\documentclass[twoside,11pt]{article}
\usepackage{jair, theapa, rawfonts}

\usepackage{amsmath}
\usepackage{amsthm}
\usepackage{amsfonts}
\usepackage{amssymb}
\usepackage{algorithm}
\usepackage[noend]{algpseudocode}
\usepackage{algorithmicx}
\usepackage{booktabs}
\usepackage[capitalise,noabbrev]{cleveref}
\usepackage{diagbox}
\usepackage[inline]{enumitem}
\usepackage{etoolbox}
\usepackage[misc]{ifsym}
\usepackage{lineno}
\usepackage{mathtools}
\usepackage{multirow}
\usepackage{soul}
\usepackage{subfigure}
\usepackage{tabularx}
\usepackage{thmtools}
\usepackage[usenames,svgnames,table]{xcolor}
\colorlet{mycolor}{blue!25}
\usepackage{thm-restate}
\usepackage{comment}

\algrenewcommand\algorithmicindent{0.75em}

\makeatletter
\newcommand{\algmargin}{\the\ALG@thistlm}
\makeatother
\newlength{\whilewidth}
\algdef{SE}[parWHILE]{parWhile}{EndparWhile}[1]
{\parbox[t]{\dimexpr\linewidth-\algmargin}{%
		\hangindent\whilewidth\strut\algorithmicwhile\ #1\ \algorithmicdo\strut}}{\algorithmicend\ \algorithmicwhile}%
\algnewcommand{\parState}[1]{\State%
	\parbox[t]{\dimexpr\linewidth-\algmargin}{\strut #1\strut}}

\usepackage{tikz}
\usetikzlibrary{fit,calc}

\colorlet{mygray}{gray!40}

\newcommand*\circled[1]{\tikz[baseline=(char.base)]{
		\node[shape=circle,draw=red,inner sep=1pt] (char) {#1};}}
\newcommand*\bcircled[1]{\tikz[baseline=(char.base)]{
		\node[shape=circle,draw=black,inner sep=1pt] (char) {#1};}}
		
\declaretheorem[name=Claim]{claim}

\crefname{thm}{Theorem}{Theorems}
\crefname{prop}{Proposition}{Propositions}

\newcommand{\impord}[1]{\vartriangleright_{#1}} 

\newcommand{\ma}{\mathcal{A}}

\newcommand{\mr}{\mathcal{R}}
\newcommand{\mP}{\mathcal{P}}

\newcommand{\ram}{random} 

\newcommand{\sd}[1]{\nobreak\succeq^{sd}_{#1}\allowbreak}

\newcommand{\fsdef}{\text{sd-}\allowbreak\text{envy-}\allowbreak\text{freeness}}
\newcommand{\Fsdef}{\text{Sd-}\allowbreak\text{envy-}\allowbreak\text{freeness}}

\newcommand{\sdef}{\text{sd-}\allowbreak\text{EF}}
\newcommand{\Sdef}{\text{Sd-}\allowbreak\text{EF}}
\newcommand{\sdefa}{\text{sd-}\allowbreak\text{EF}}

\newcommand{\fepopt}{\text{ex-post }\allowbreak\text{Pareto-}\allowbreak\text{efficiency}}

\newcommand{\epopt}{\text{ep-}\allowbreak\text{PE}}

\newcommand{\fsdopt}{\text{sd-}\allowbreak\text{Pareto-}\allowbreak\text{efficiency}}
\newcommand{\Fsdopt}{\text{Sd-}\allowbreak\text{Pareto-}\allowbreak\text{efficiency}}

\newcommand{\sdopt}{\text{sd-}\allowbreak\text{PE}}
\newcommand{\sdopta}{\text{sd-}\allowbreak\text{PE}}

\newcommand{\fsdsp}{\text{sd-}\allowbreak\text{weak-}\allowbreak\text{strategy}\allowbreak\text{proofness}}
\newcommand{\Fsdsp}{\text{Sd-}\allowbreak\text{weak-}\allowbreak\text{strategy}\allowbreak\text{proofness}}

\newcommand{\sdsp}{\text{sd-}\allowbreak\text{WSP}}
\newcommand{\sdspa}{\text{sd-}\allowbreak\text{WSP}}

\newcommand{\fsdssp}{\text{sd-}\allowbreak\text{strategy}\allowbreak\text{proofness}}
\newcommand{\Fsdssp}{\text{Sd-}\allowbreak\text{strategy}\allowbreak\text{proofness}}

\newcommand{\sdssp}{\text{sd-}\allowbreak\text{SP}}
\newcommand{\Sdssp}{\text{Sd-}\allowbreak\text{SP}}
\newcommand{\sdsspa}{\text{sd-}\allowbreak\text{SP}}

\newcommand{\ucs}{U}			

\newcommand{\vx}{{\bf{x}}}
\newcommand{\vy}{{\bf{y}}}

\newcommand{\fsdwef}{\text{sd-}\allowbreak\text{weak-}\allowbreak\text{envy-}\allowbreak\text{freeness}}
\newcommand{\Fsdwef}{\text{Sd-}\allowbreak\text{weak-}\allowbreak\text{envy-}\allowbreak\text{freeness}}

\newcommand{\sdwef}{\text{sd-}\allowbreak\text{WEF}}

\newcommand\xiaoxi[1]{{\color{brown} \footnote{\color{brown}xiaoxi: #1}} }
\newtheorem{definition}{Definition}

\newtheorem{lemma}{Lemma}
\newtheorem{example}{Example}

\newtheorem{corollary}{Corollary}
\newtheorem{remark}{Remark}
\newenvironment{sketch}{{\par\bf\noindent Proof sketch.}}{\nobreak\hfill$\Box$}

\newcommand{\am}{\text{EBM}}
\newcommand{\abm}{\text{ABM}}
\newcommand{\pcr}{\text{PRE}}
\newcommand{\upre}{\text{UPRE}}
\newcommand{\rk}[2]{rk(#1,#2)}

\newcommand{\fetep}{\text{strong }\allowbreak\text{equal }\allowbreak\text{treatment }\allowbreak\text{of }\allowbreak\text{equals}}
\newcommand{\Fetep}{\text{Strong }\allowbreak\text{equal }\allowbreak\text{treatment }\allowbreak\text{of }\allowbreak\text{equals}}
\newcommand{\etep}{\text{SETE}}

\newcommand{\ffhr}{\text{favoring-}\allowbreak\text{higher-}\allowbreak\text{ranks}}
\newcommand{\Ffhr}{\text{Favoring-}\allowbreak\text{higher-}\allowbreak\text{ranks}}
\newcommand{\fhr}{\text{FHR}}

\newcommand{\fefr}{\text{ex-post }\allowbreak\text{favoring-}\allowbreak\text{higher-}\allowbreak\text{ranks}}
\newcommand{\efr}{\text{ep-}\allowbreak\text{FHR}}
\newcommand{\efhr}{\text{ep-}\allowbreak\text{FHR}}
\newcommand{\Efhr}{\text{Ep-}\allowbreak\text{FHR}}

\newcommand{\ffhcr}{\text{favoring-}\allowbreak\text{eagerness-}\allowbreak\text{for-}\allowbreak\text{remaining-}\allowbreak\text{items}}
\newcommand{\fhcr}{\text{FERI}}

\newcommand{\fefcr}{\text{ex-post }\allowbreak\ffhcr{}}

\newcommand{\efcr}{\text{ep-}\allowbreak\fhcr{}}
\newcommand{\Efcr}{\text{Ep-}\allowbreak\fhcr{}}

\newcommand{\fsdrf}{\text{ex-ante }\allowbreak\ffhr{}}
\newcommand{\Fsdrf}{\text{Ex-ante }\allowbreak\fhr{}}
\newcommand{\sdrf}{\text{ea-}\allowbreak\fhr{}}
\newcommand{\sdrfa}{\text{ea-}\allowbreak\fhr{}}

\newcommand{\fsdcrf}{\text{ex-ante }\allowbreak\ffhcr{}}

\newcommand{\sdcrf}{\text{ea-}\allowbreak\fhcr{}}
\newcommand{\Sdcrf}{\text{Ea-}\allowbreak\fhcr{}}
\newcommand{\sdcrfa}{\text{ea-}\allowbreak\fhcr{}}

\newcommand{\fopt}{\text{Pareto-}\allowbreak\text{efficiency}}
\newcommand{\opta}{\text{PE}}
\newcommand{\opt}{\text{PE}}

\newcommand{\tp}[2]{\nobreak top(#1,#2)\allowbreak}

\newcommand{\frkm}{\text{rank-}\allowbreak\text{maximality}}

\newcommand{\rkm}{\text{RM}}

\newcommand{\ferkm}{\text{ex-post }\allowbreak\text{rank-}\allowbreak\text{maximality}}

\newcommand{\erkm}{\text{ep-}\allowbreak\text{RM}}

\newcommand{\fcm}{\text{FCM}}
\newcommand{\ffcm}{\text{first-}\allowbreak\text{choice }\allowbreak\text{maximality}}
\newcommand{\Ffcm}{\text{First-}\allowbreak\text{choice }\allowbreak\text{maximality}}
\newcommand{\ffcma}{\text{first-}\allowbreak\text{choice }\allowbreak\text{maximal}}

\newcommand{\fpop}{\text{popularity}}
\newcommand{\fpopa}{\text{popular}}
\newcommand{\pop}{\text{POP}}

\newcommand{\tps}[2]{T_{#1,#2}}

\newcommand{\Mod}[1]{\ (\mathrm{mod}\ #1)}

\newcommand{\mU}{u}
\newcommand\Myperm[2][^n]{\prescript{#1\mkern-2.5mu}{}P_{#2}}
\newcommand\Mycomb[2][^n]{\prescript{#1\mkern-0.5mu}{}C_{#2}}

\def\oldproof{true}

\renewcommand{\cite}{\shortcite}
\renewcommand{\citeA}{\shortciteA}
\renewcommand{\citeR}{\shortciteR}

\ShortHeadings{Favoring Eagerness for Remaining Items}
{Guo, Sikdar, Xia, Wang, \& Cao}

\begin{document}

\title{Favoring Eagerness for Remaining Items:\\ Designing Efficient, Fair, and Strategyproof Mechanisms}

\author{\name Xiaoxi Guo \email guoxiaoxi@pku.edu.cn \\
       \addr  Key Laboratory of High Confidence Software Technologies
        (MOE), School of Computer Science, Peking University,
        Beijing 100871, China
       \AND
       \name Sujoy Sikdar \email ssikdar@binghamton.edu \\
       \addr  Department of Computer Science, Binghamton University,\\
        4400 Vestal Parkway East, Binghamton 13902, New York, USA
       \AND
       \name Lirong Xia \email xialirong@gmail.com \\
       \addr  Department of Computer Science, Rensselaer Polytechnic Institute,\\
       110 Eighth Street, Troy 610101, New York, USA
       \AND
        \name Yongzhi Cao{\rm\textsuperscript{\Letter}} \email caoyz@pku.edu.cn \\
       \addr  Key Laboratory of High Confidence Software Technologies
        (MOE), School of Computer Science, Peking University,
        Beijing 100871, China
       \AND
        \name Hanpin Wang \email whpxhy@pku.edu.cn \\
       \addr   School of Computer Science and Cyber Engineering, Guangzhou
        University,\\ Guangzhou 510006, China\\
       Key Laboratory of High Confidence Software Technologies
        (MOE), School of Computer Science, Peking University,
        Beijing 100871, China
       }


\maketitle

\begin{abstract}
In the assignment problem, the goal is to assign indivisible items to agents who have ordinal preferences, efficiently and fairly, in a strategyproof manner. In practice, {\em\ffcm{}}, i.e., assigning a maximal number of agents their top items, is often identified as an important efficiency criterion and measure of agents' satisfaction. In this paper, we propose a natural and intuitive efficiency property, {\em \ffhcr{}} (\fhcr{}), which requires that each item is allocated to an agent who ranks it highest among remaining items, thereby implying \ffcm{}. Using \fhcr{} as a heuristic, we design mechanisms that satisfy ex-post or ex-ante variants of \fhcr{} together with combinations of other desirable properties of efficiency (\fopt{}), fairness (\fetep{} and \fsdwef{}), and strategyproofness (\fsdsp{}). 
We also explore the limits of \fhcr{} mechanisms in providing stronger efficiency, fairness, or strategyproofness guarantees through impossibility results.

\end{abstract}

\section{Introduction}\label{sec:intro}
In the {\em assignment problem}~\cite{hylland1979efficient,zhou1990conjecture}, $n$ agents have unit demands and {\em strict ordinal preferences} for $n$ items, each with unit supply, and the goal is to compute an {\em assignment} which allocates each agent with one unit of items and (approximately) maximizes agent satisfaction. This serves as a useful model for a variety of problems involving houses~\cite{shapley1974cores}, dormitory rooms~\cite{chen2002improving}, school choice without priorities~\cite{miralles2009school}, and computational resources in cloud computing~\cite{Ghodsi11:Dominant,Ghodsi12:Multi,Grandl15:Multi}. 
Due to the wide applicability of the assignment problem, there is a rich literature pursuing the design of assignment mechanisms satisfying desirable properties of efficiency, fairness, and strategyproofness. However, many of these properties are incompatible with each other, and trade-offs must be made.

In several practical assignment problems, whether a maximal number of agents are allocated their respective top items, or {\em \ffcm{}} (\fcm), is identified as an important measure of agents' satisfaction with an assignment. For example, in school choice programs, the percentage of students admitted to their most preferred school is often prominently reported in mass media as a measure of student welfare, and is therefore also an important consideration for school administrators~\cite{dur2018first}.
Often, additional efficiency guarantees are also desired such as {\em \fopt{}} (\opt{}), which requires that an assignment cannot be improved upon so that some agents are better off and no agent is worse off. 
In light of these considerations, we seek to address the following question in this paper: 
{\em Can we design mechanisms that satisfy important efficiency criteria (such as \fcm{} and \opt{} simultaneously), while also providing desirable fairness and strategyproofness guarantees?}

The desire for efficiency has motivated the design of mechanisms that (approximately) maximize total satisfaction by a natural heuristic which seeks to allocate each item to an agent who ranks it as highly as possible. A prominent example of this is the famous Boston mechanism, which proceeds iteratively by allocating as many items as possible to agents who rank it as their first choice, then allocating as many items as possible to agents who rank the items in second position, and so on. In fact, \citeA{Kojima2014:Boston} showed that the Boston mechanism is characterized by a formalization of this natural heuristic, the {\em \ffhr{}} (\fhr) efficiency property implying both \fcm{} and \opt{}, which requires that each item is allocated to an agent that ranks it highest, unless every such agent is allocated an item she ranks higher.

However, the Boston mechanism has long been criticized for failing to provide strategyproofness~\cite{abdulkadiroglu2006changing,pathak2017really,pathak2008leveling,roth2005pairwise} which is often considered equally important to \fcm{} in practical applications like school choice and kidney exchange. Several works have attempted to address this failing by proposing variants of the Boston mechanism. 
Most notably,~\citeA{Mennle2021:Partial} showed that some members of {\em adaptive Boston mechanisms} (\abm{},~\citeR{alcalde1996implementation}) satisfies strategyproofness, but do not consider the question of fairness.

When items are indivisible, even relatively basic fairness notions such as the equal treatment of agents with identical preferences can only satisfied by a random mechanism, such as the Boston mechanism where ties between agents are broken using a lottery. \citeA{Ramezanian2021:Ex-post} showed that for any lottery, the expected output of the Boston mechanism satisfies {\em ex-post} \fhr{} (\efhr{}). However, they also showed that no \efhr{} mechanism can satisfy either the fairness property {\em \fsdef{}} (\sdef), or the strategyproofness property {\em \fsdssp{}} (\sdssp). \Efhr{} is also incompatible with a combination of the weaker strategyproofness property {\em \fsdsp{}} (\sdsp) and the basic fairness property of {\em \fetep{}} (\etep{}). Here, \sdef{} is an extension of envy-freeness~\cite{foley1966resource,VARIAN197463} which requires that no agent considers her allocation to be dominated by that of another agent when allocations are compared using the notion of {\em stochastic dominance} (sd,~\citeR{Bogomolnaia01:New}); \sdssp{} and \sdsp{} require that no agent can manipulate the outcome of the mechanism to her benefit by misreporting her preferences~\cite{Bogomolnaia01:New}; and \etep{} requires that agents who share a common prefix in their rankings of items are allocated items in the shared prefix with equal probability~\cite{nesterov2017fairness}.

\subsection{Our Contributions}
We begin by showing that \efhr{} is not compatible with \etep{} and \sdwef{} (a mildly weaker variant of \sdef{}) in \Cref{prop:impefr}, which complements the impossibility results by \citeA{Ramezanian2021:Ex-post} of the incompatibility of \fhr{} with \sdef{} and with \etep{} and \sdsp{}. Together, this means that mechanisms that satisfy \efhr{} do not provide an avenue to answer our question (see \Cref{sec:app:imp}).

Our main conceptual contribution is a natural alternative principle for the design of assignment mechanisms: that each item be allocated to an agent most {\em ``eager''} for it, i.e., ranks it highest {\em among remaining items}. This forms the basis of a novel efficiency property, \ffhcr{} (\fhcr{}), which implies both \fcm{} and \opt{}.

We provide an affirmative answer to the question we seek to address in the paper through our main technical contributions. Using \fhcr{} as a heuristic, we design two mechanisms which  satisfy desirable combinations of efficiency, fairness, and strategyproofness properties (defined formally in Section~\ref{sec:properties}, and summarized in~\Cref{tab:abbrp}):


\paragraph{An ex-post \fhcr{} (\efcr{}), fair, and strategyproof mechanism.} The {\em eager Boston mechanism} (\am{},~\Cref{alg:am}) we design is efficient, fair, and strategyproof. \am{} satisfies ex-post \fhcr{} (\efcr), which implies ex-post \fcm{} and ex-post \opt{}, and also satisfies \etep{}, \sdwef{}, and \sdsp{} (\Cref{thm:amp}). 
\am{} bears a close resemblance to the \abm{} family of mechanisms (\Cref{alg:abm}), but as we show, \am{} is not a member of \abm{} (\Cref{rmk:amnotabm}), although they are closely related: Every \efcr{} assignment, including the output of \am{}, can be computed by some member of \abm{}, and every member of \abm{} satisfies \efcr{} (\Cref{thm:abmchar}).

\paragraph{An ex-ante \fhcr{} (\sdcrf{}) and fair mechanism.} We identify the {\em uniform probabilistic respecting eagerness mechanism} (\upre{},~\Cref{dfn:upre}) and show that it satisfies  \etep{} and \sdwef{}~(\Cref{thm:uprep}). In addition, \upre{} satisfies ex-ante \fhcr{} (\sdcrf{}) which implies ex-post \fcm{}, \epopt{}, and \sdopt{}. This is because \upre{} belongs to the family of {\em probabilistic respecting eagerness mechanisms} (\pcr{},~\Cref{alg:pcr}), and as we show that: Every member of \pcr{} satisfies \sdcrf{}, and in addition, every \sdcrf{} assignment must be the output of some member of \pcr{} (\Cref{thm:familychar}).

\vspace{0.5em}
In addition, we explore if \efcr{} or \sdcrf{} is compatible with stronger notions of fairness (\sdef{} over \sdwef) and strategyproofness (\sdssp{}, a stronger variant of \sdsp{}), and find that that no mechanism can satisfy the following combinations of properties:  
\efcr{} and \sdef{}~(\Cref{prop:impefcr1}); \sdcrf{} and \sdef{}~(\Cref{prop:impsdcfr1}); \efcr{}, \etep{} and \sdssp{}~(\Cref{prop:impefcr2}); \sdcrf{}, \etep{}, and \sdsp{}~(\Cref{prop:impsdcfr2}); and \efcr{}, \sdcrf{}, and \etep{}~(\Cref{prop:impefcrsdcfr}).

\subsection{Related Work}

\citeA{Chen2021:Theprobabilistic} proposed another extension of \fhr{}, ex-ante~\fhr{} (\sdrf{})\footnote{\citeA{Chen2021:Theprobabilistic} named this property sd-rank-fairness. We rename it here to emphasize its connection with \fhr{}.} and provided the probabilistic rank mechanism which satisfies \sdrf{}. Since \sdrf{} implies \efr{}, it suffers the same incompatibility with fairness and strategyproofness as \efhr{} does. Apart from \fhr{}, {\em \frkm{}} (\rkm{},~\citeR{irving2006rank,paluch2013capacitated}) is another popular efficiency property that implies \fcm{}, which has been widely studied for assigning schools to students~\cite{abraham2009matching}, assigning papers to referees~\cite{garg2010assigning}, and rental items to customers~\cite{abraham2006assignment}. However, since \rkm{} is stronger than \fhr{}~\cite{belahcene2021combining}, once again, the incompatibility with fairness and strategyproofness extends to \rkm{}.

Looking beyond mechanisms that attempt to allocate items to agents who rank them highest, RP and PS are famous mechanisms widely studied in the literature due to their fairness and strategyproofness guarantees (see \cref{tab:properties}). However, both RP and PS fail to satisfy \fcm{} (See~\Cref{sec:app:results:properties}), and therefore they do not provide a positive answer to the question we study in the paper.


\begin{table}[tp]
    \fontsize{9pt}{\baselineskip}\selectfont
    \renewcommand{\tabcolsep}{0.3em}
	\centering
	\begin{tabular}{|c|ccc|ccc|ccc|cc|}
		\hline
		\multirow{2}{*}{} & \multicolumn{3}{c|}{ex-post efficiency} & \multicolumn{3}{c|}{ex-ante efficiency}  &  \multicolumn{3}{c|}{ex-ante fairness} &  \multicolumn{2}{c|}{strategyproofness}\\
		
		& \efcr{} & \efr{} & \epopt{}
		& \sdcrf{}{} & \sdrf{} & \sdopt{}
		& \sdef{} & \sdwef{} & \etep{}
		& \sdssp{} & \sdsp{}\\\hline
		
		RP
		& \cellcolor{mycolor}N$^{P\ref{prop:rp}}$ & N$^\texttt{a}$ & Y$^\texttt{c}$
		& \cellcolor{mycolor}N$^{P\ref{prop:rp}}$ & N$^\texttt{b}$ & N$^\texttt{c}$
		& N$^\texttt{c}$ & Y$^\texttt{c}$ & Y$^\texttt{d}$
		& Y$^\texttt{c}$ & Y$^\texttt{c}$\\
		
		PS
		& \cellcolor{mycolor}N$^{P\ref{prop:ps}}$ & N$^\texttt{a}$ & Y$^\texttt{a}$
		& \cellcolor{mycolor}N$^{P\ref{prop:ps}}$ & N$^\texttt{b}$ & Y$^\texttt{c}$
		& Y$^\texttt{c}$ & Y$^\texttt{c}$ & Y$^{\texttt{c},\texttt{d}}$
		& N$^\texttt{c}$ & Y$^\texttt{c}$\\\hline
		
		BM$^*$
		& \cellcolor{mycolor}N$^{P\ref{prop:nbm}}$ & Y$^\texttt{a}$ & Y$^\texttt{a}$
		& \cellcolor{mycolor}N$^{P\ref{prop:nbm}}$  & N$^\texttt{b}$ & N$^\texttt{b}$
		& \cellcolor{mycolor}N$^{P\ref{prop:nbm}}$ & \cellcolor{mycolor}N$^{P\ref{prop:nbm}}$ & \cellcolor{mycolor}Y$^{P\ref{prop:nbm}}$
		& N$^\texttt{a}$ & N$^\texttt{a}$\\

        ABM$^*$
		& \cellcolor{mycolor}Y$^{P\ref{prop:abm}}$ & \cellcolor{mycolor}N$^{P\ref{prop:abm}}$ & Y$^\texttt{e}$
		& \cellcolor{mycolor}N$^{P\ref{prop:abm}}$ & \cellcolor{mycolor}N$^{P\ref{prop:abm}}$ & \cellcolor{mycolor}N$^{P\ref{prop:abm}}$
		& \cellcolor{mycolor}N$^{P\ref{prop:abm}}$ & ? & \cellcolor{mycolor}Y$^{T\ref{prop:abm}}$
		& N$^\texttt{e}$ & Y$^\texttt{f}$\\
		
		\am{}
		& \cellcolor{mycolor}Y$^{T\ref{thm:amp}}$ & \cellcolor{mycolor}N$^{P\ref{prop:ebm}}$ & Y$^{\texttt{e}}$
		& \cellcolor{mycolor}N$^{P\ref{prop:ebm}}$ & \cellcolor{mycolor}N$^{P\ref{prop:ebm}}$ & \cellcolor{mycolor}N$^{P\ref{prop:ebm}}$
		& \cellcolor{mycolor}N$^{P\ref{prop:ebm}}$ & \cellcolor{mycolor}Y$^{T\ref{thm:amp}}$ & \cellcolor{mycolor}Y$^{T\ref{thm:amp}}$
		& \cellcolor{mycolor}N$^{P\ref{prop:ebm}}$ & \cellcolor{mycolor}Y$^{T\ref{thm:amp}}$\\\hline

		PR
		& \cellcolor{mycolor}N$^{P\ref{prop:pr}}$ & Y$^{\texttt{a},\texttt{b}}$ & Y$^\texttt{b}$
		& \cellcolor{mycolor}N$^{P\ref{prop:pr}}$ & Y$^\texttt{b}$ & Y$^\texttt{b}$
		& N$^\texttt{b}$ & N$^\texttt{b}$ & \cellcolor{mycolor}Y$^{P\ref{prop:pr}}$
		& N$^\texttt{b}$ & N$^\texttt{b}$\\
		
		\upre{}
		& \cellcolor{mycolor}N$^{P\ref{prop:upre}}$
		& \cellcolor{mycolor}N$^{P\ref{prop:upre}}$
		& \cellcolor{mycolor}Y$^{C\ref{cor:pcrp}}$
		& \cellcolor{mycolor}Y$^{T\ref{thm:familychar}}$
		& \cellcolor{mycolor}N$^{P\ref{prop:upre}}$
		& \cellcolor{mycolor}Y$^{C\ref{cor:pcrp}}$
		& \cellcolor{mycolor}N$^{P\ref{prop:upre}}$
		& \cellcolor{mycolor}Y$^{T\ref{thm:uprep}}$
		& \cellcolor{mycolor}Y$^{T\ref{thm:uprep}}$
		& \cellcolor{mycolor}N$^{P\ref{prop:upre}}$
		& \cellcolor{mycolor}N$^{P\ref{prop:upre}}$\\
		\hline
	\end{tabular}
	\caption{Properties of RP, PS, BM, \am{}, PR and \pcr{}.}\label{tab:properties}
    \begin{flushleft}
    \footnotesize{Note: A `Y' indicates that the mechanism at that row satisfies the property at that column, and an `N' indicates that it does not. 
    Results annotated with `\texttt{a}' follow from~\protect\citeA{Ramezanian2021:Ex-post}, `\texttt{b}' from~\protect\citeA{Chen2021:Theprobabilistic}, `\texttt{c}' from~\protect\citeA{Bogomolnaia01:New}, `\texttt{d}' from~\protect\citeA{nesterov2017fairness},  `\texttt{e}' from~\protect\citeA{Dur2019modified}, and `\texttt{f}' from~\protect\citeA{Mennle2021:Partial}
    respectively.
    A result annotated with T, P or C refers to a Theorem, Proposition or Corollary in this paper (or Appendix~\ref{sec:app:results}), respectively.}
    
    \footnotesize{*Here, we refer to the expected outputs of the BM and ABM when the priority order over agents is drawn from a uniform distribution over priority orders.}
    \end{flushleft}
\end{table}

\Cref{tab:properties}
compares the properties of \am{} and \upre{} to the properties of the random priority mechanism (RP,~\citeR{Abdulkadiroglu98:Random}), probabilistic serial mechanism (PS,~\citeR{Bogomolnaia01:New}), Boston mechanism (BM,~\citeR{abdulkadirouglu2003school,Kojima2014:Boston}), adaptive Boston mechanism (\abm,~\citeR{alcalde1996implementation,Dur2019modified}), and probabilistic rank mechanism~(PR,~\citeR{Chen2021:Theprobabilistic}). 
\Cref{fig:nwaxioms} shows the relationship between efficiency properties based on \fhcr{} to extensions of \opt{} or \fhr{}.

We note that \citeA{Harless2018:Immediate} proposed the immediate division$^+$ mechanism and proved it satisfies \sdwef{}. This mechanism appears similar to \upre{}, although we are unable to prove or disprove their equivalence. In our paper, we define the family of \pcr{} mechanisms (of which \upre{} is a member) and prove that it is characterized by the newly-proposed property \sdcrf{}, which has not been considered earlier to the best of our knowledge. 
In addition, with the impossibility results we proved for \sdcrf{}, we show the limit of the family of \pcr{}, including \upre{}, on guarantees of efficiency, fairness, and strategyproofness.

\section{Preliminaries}
An instance of the {\em assignment problem} is given by a tuple $(N,M)$ and a {\em preference profile} $R$, where $N=\{1,\dots,n\}$ is a set of $n$ {\em agents}, and $M=\{o_1,\dots,o_n\}$ is a set of $n$ {\em items} with a single unit of {\em supply} of each item.

\vspace{1em}
\noindent{\bf Preferences.} A {\em preference profile} $R=(\succ_j)_{j\in N}$ specifies the ordinal preference of each agent $j\in N$ as a strict linear order over $M$, and $\succ_{-j}$ denotes the collection of preferences of agents in $N\setminus\{j\}$.
Let $\mr$ be the set of all the preference profiles.
For any $j\in N$, we use $\rk{\succ_j}{o}$ to denote the rank of item $o$ in $\succ_j$, and $\tp{\succ_j}{S}$ to denote the item ranked highest in $\succ_j$ among $S\subseteq M$. 
We also use $\rk{j}{o}$ and $\tp{j}{S}$ for short if it is clear in the context.
For any linear order $\succ$ over $M$ and item $o$,
$\ucs(\succ,o)=\{o'\in M\mid o'\succ o\}\cup\{o\}$
represents the items weakly preferred to $o$.
For any pair of agents $j,k\in N$, the common prefix of their preferences $\succ_{j,k}$ is the preference over the first several items which have the same upper contour set in $\succ_j$ and $\succ_k$. 
Formally, $\succ_{j,k}$ is a strict linear preference over $M'\subseteq M$ such that
\begin{enumerate*}[label=(\roman*)]
    \item for any $o\in M'$,  $\rk{j}{o}=\rk{k}{o}=\rk{\succ_{j,k}}{o}\le |M'|$, and 
    \item $\tp{j}{M\setminus M'}\neq\tp{k}{M\setminus M'}$.
\end{enumerate*}

\vspace{1em}
\noindent{\bf Allocations, Assignments, and Mechanisms.} A {\em \ram{} allocation} is a stochastic $n$-vector $p=[p_o]_{o\in M}$ describing the probabilistic share of each item. Let $\Pi$ be the set of all the possible \ram{} allocations. A {\em \ram{} assignment} is a doubly stochastic $n\times n$ matrix $P=[p_{j,o}]_{j\in N, o\in M}$. For each agent $j\in N$, the $j$-th row of $P$, denoted $P_j$, is agent $j$'s \ram{} allocation, and for each item $o\in M$, $p_{j,o}$ is $j$'s probabilistic share of $o$. We use $\mP$ to denote the set of all possible \ram{} assignments. A {\em deterministic assignment} $A:N\to M$ is a one to one mapping from agents to items, represented by a binary doubly stochastic $n\times n$ matrix. 
For each agent $j\in N$, we use $A(j)$ to denote the item allocated to $j$, and for each item $o\in M$, $A^{-1}(o)$ to denote the agent allocated $o$. 
Let $\ma$ denote the set of all the deterministic assignment matrices. By the Birkhoff-Von Neumann theorem, every \ram{} assignment $P\in\mP$ describes a probability distribution over $\ma$.

A {\em mechanism} $f\colon\mr\to\mP$ is a mapping from preference profiles to \ram{} assignments.
For any profile $R\in\mr$, we use $f(R)$ to refer to the \ram{} assignment output by $f$. For every agent $j\in N$, we use $f(R)_{j}$ to denote agent $j$'s \ram{} allocation, and for every item $o\in M$, we use $f(R)_{j,o}$ to denote $j$'s share of $o$.

\subsection{Economic Efficiency}\label{sec:eff}

\paragraph{\fopt{} (\opt{}).} A deterministic assignment $A$ satisfies \opt{} if no agent can be assigned a better item without assigning any other agent a worse item, i.e., there does not exist another $A'$ and a set $N'\subseteq N$ with $N'\neq\emptyset$ such that $A'(j)\succ_j A(j)$ for any $j\in N'$ and $A'(k)=A(k)$ for $k\in N\setminus  N'$.

\paragraph{\Ffcm{} (\fcm{}).} A deterministic assignment $A$ satisfies \fcm{} if it assigns a maximal number of agents their top ranked items, i.e., there does not exist another $A'$ such that $\lvert \{j\in N\mid \rk{j}{A'(j)}=1\}\rvert>\lvert\{ j\in N\mid \rk{j}{A(j)}=1\}\rvert$.


\paragraph{\Ffhr{} (\fhr{}).} A deterministic assignment $A$ satisfies \fhr{}, if every item is allocated to an agent that ranks it highest unless every such agent is allocated an item she ranks higher. 
Formally, $A$ satisfies \fhr{} if for any agents $j,k\in N$, $\rk{j}{A(j)}\le\rk{k}{A(j)}$ or $\rk{k}{A(k)}<\rk{k}{A(j)}$.

\begin{example}\label{eg:rkmfhr}
    \rm 
	Consider the preference profile $R$ in \Cref{fig:fhr}. 
	
	\begin{figure}[htb]
		\centering
		\renewcommand{\tabcolsep}{0.1em}
		\renewcommand\arraystretch{1.7}
    		\begin{tabular}{cccccccccccc}
    		     $\succ_1$: &~\circled{$a$} & $\succ_1$ & $b$ & $\succ_1$ & $c$ &$\succ_1$ & $~d~$ &$\succ_1$ & $~e~$ &$\succ_1$ & $f$ \\
    		     $\succ_2$: &~\circled{$b$} &$\succ_2$ &$a$ & $\succ_2$ &$c$ &$\succ_2$ &$d$ &$\succ_2$ & $e$ &$\succ_2$ & $f$ \\
    		     $\succ_3$: &~\circled{$c$} & $\succ_3$ & $e$ &$\succ_3$& $d$ &$\succ_3$& $f$& $\succ_3$& $a$ &$\succ_3$& $b$ \\
    		     $\succ_4:$ &~$c$&$\succ_4$&$ \circled{$e$}$&$\succ_4$&$ d$&$\succ_4$&$ f$&$ \succ_4$&$a$&$\succ_4$&$ b$ \\
    		     $\succ_5:$ &~$c$&$\succ_5$&$ e$&$\succ_5$&$ \circled{$d$}$&$\succ_5$&$ f$&$\succ_5$&$a$&$\succ_5$&$ b$ \\
                 $\succ_6:$ &~$c$ & $\succ_6$&$ a$&$\succ_6$&$ b$&$\succ_6$&$ d$&$\succ_6$&$e$&$ \succ_6$&$ \circled{$f$}$\\
    		\end{tabular}
		\caption{A linear preference profile $R$.}
		\label{fig:fhr}
	\end{figure}
	
	In any assignment that satisfies~\fhr{}, by definition, each item must be assigned to one of the agents who ranks it on the top if such agents exist.	
	Therefore, $a$ and $b$ go to agents $1$ and $2$, respectively.
	Notice that agents $3$-$6$ all rank $c$ on top.
	If $c$ is allocated to agents $3$-$5$, then by \fhr{}, agent $6$ cannot be assigned either item $d$ or item $e$, since for any $j\in\{3,4,5\}$, $\rk{6}{d}>\rk{j}{d}$ and $\rk{6}{e}>\rk{j}{e}$. The items circled in red represent one such deterministic assignment which satisfies \fhr{}.
\hfill$\square$
\end{example}

By the Birkhoff-Von Neumann theorem, all of the properties for deterministic assignments can naturally be extended to \ram{} assignments: A \ram{} assignment satisfies {\em ex-post} $X$ if it is a convex combination of deterministic assignments satisfying property $X$.
In this paper, we also say that a mechanism $f$ satisfies a property $Y$, if for every profile $R\in\mr$, $f(R)$ satisfies $Y$.

Besides the efficiency notions above, we also introduce ex-ante notions for random assignments. 
One of the notion is based on the stochastic dominance (sd), which extends an agent's preference over single items to lotteries over items~\cite{segal2020fair} and helps comparing random allocations and assignments.

\begin{definition}\cite{Bogomolnaia01:New}\label{dfn:sd}
	Given a preference relation $\succ$ over $M$, the {\em stochastic dominance} relation associated with $\succ$, denoted by $\sd{\null}$, is a partial ordering over $\Pi$ such that for any pair of \ram{} allocations $p,q\in\Pi$, $p$ (weakly) {\em stochastically dominates} $q$, denoted by $p\sd{\null} q$, if for any $o\in M$, $\sum_{o'\in\ucs(\succ,o)}p_{o'}\ge\sum_{o'\in\ucs(\succ,o)}q_{o'}$.
\end{definition}

\paragraph{\Fsdopt{} (\sdopt{}).} A random assignment $P$ satisfies \sdopt{} if $P$ is not stochastically dominated by other \ram{} assignments, i.e., there does not exist a \ram{} assignment $Q\neq P$ such that $Q_j\sd{j}P_j$ for every $j\in N$.

\paragraph{\Fsdrf{} (\sdrf{})} A random assignment $P$ satisfies \sdrf{}, if the shares of every item are allocated to agents that rank it highest unless every such agent's demand is satisfied. Formally, $P$ satisfies \sdrf{} if for every agent $j\in N$ and every $o\in M$ such that $p_{j,o}>0$, it holds that  for every $k\in N$ such that $\rk{k}{o}<\rk{j}{o}$, $\sum_{o'\in\ucs(k,o)}p_{k,o'}=1$.

\begin{remark}
\rm
    For deterministic assignments, \fhr{} implies \fcm{}~\cite{dur2018first,Kojima2014:Boston} and \opt{}~\cite{Ramezanian2021:Ex-post}, and \fcm{} and \opt{} do not imply each other.
    As for random assignments, \sdrf{} implies \sdopt{} and \efhr{}~\cite{Chen2021:Theprobabilistic}, while both \sdopt{} and \efhr{} implies \epopt{}~\cite{Bogomolnaia01:New,Ramezanian2021:Ex-post}.
\end{remark}

\begin{table}[t]
	
	\centering
	\begin{tabular}{l|l|l}
		Abbr. & full names & category\\\hline
		\sdcrf & \fsdcrf{} & ex-ante efficiency\\
		\sdrf{} & \fsdrf{} & ex-ante efficiency\\
		\efcr{} & \fefcr{} & ex-post efficiency\\
		\efr{} & \fefr{} & ex-post efficiency\\
		\epopt{} & \fepopt{} & ex-post efficiency\\
		\fhcr{} & \ffhcr{} & efficiency$^*$\\
		\fhr{} & \ffhr{} & efficiency$^*$\\
		\opt{} &\fopt{} & efficiency$^*$\\
		\sdef{} & \fsdef{} & ex-ante fairness\\
		\sdopt{} & \fsdopt{} & ex-ante efficiency\\
		\sdssp{} &\fsdssp{} & strategyproofness\\
		\sdwef{} & \fsdwef{} & ex-ante fairness\\
		\sdsp{} &\fsdsp{} & strategyproofness\\
		\etep{} &\fetep{} & ex-ante fairness\\
	\end{tabular}
	\caption{Acronyms for properties used in this paper.}
    \begin{flushleft}
      \footnotesize{Note: Properties annotated with $^*$ are for deterministic assignments}
    \end{flushleft}
	\label{tab:abbrp}
\end{table}

\subsection{Fairness and Strategyproofness}\label{sec:properties}

\paragraph{\Fetep{} (\etep{}).} A random assignment $P$ satisfies \etep{} if any two agents have the same allocation over items appearing in the common prefix of their preferences. Formally, for every pair of $j,k\in N$, $p_{j,o}=p_{k,o}$ for any $o$ appearing in $\succ_{j,k}$.

\paragraph{\Fsdef{} (\sdef{}).} A random assignment $P$ is \sdef{}, if every agent's allocation weakly stochastically dominates the others', i.e., $P_j\sd{j}P_{k}$ for every pair of $j,k\in N$.
	
\paragraph{\Fsdwef{} (\sdwef{}).} A random assignment $P$ is \sdwef{}, if no agent's allocation is dominated by others', i.e., $P_{k}\sd{j}P_j\implies P_j=P_k$ for every pair of $j,k\in N$.

\begin{remark}\rm
    \Sdef{} implies \sdwef{}~\cite{Bogomolnaia01:New} and \etep{}~\cite{nesterov2017fairness}, while \sdwef{} and \etep{} do not imply each other.
\end{remark}

\paragraph{\Fsdssp{} (\sdssp{}).} When an agent reports the true preference, a mechanism $f$ satisfying \sdssp{} always outputs an allocation which weakly dominates the ones when she misreports. Formally, for every $R\in\mr$, it holds that $f(R)\sd{j}f(R')$ for every $j\in N$ and $R'=(\succ'_j,\succ_{-j})$,

\paragraph{\Fsdsp{} (\sdsp{}).} A mechanism $f$ satisfying \sdsp{} guarantees that when an agent misreports her preference, she would not receive an allocation dominating the one when she truly reports. Formally, for every $R\in\mr$, it holds that $f(R')\sd{j}f(R)\implies f(R')_j=f(R)_j$ for every $j\in N$, and $R'=(\succ'_j,\succ_{-j})$.

\begin{remark}\rm
    \Sdssp{} implies \sdsp{}~\cite{Bogomolnaia01:New}.
\end{remark}


\section{Incompatibility of \fhr{} with Fairness}
\label{sec:app:imp}
In this section, we show that \fhr{} mechanisms are unable to satisfy desirable properties of fairness. In~\Cref{prop:impefr}, we show that requiring \efr{} together with \etep{} leads to a violation of \sdwef{}, meaning that no \fhr{} mechanisms can satisfy all of these properties simultaneously. This complements the results by \citeA{Ramezanian2021:Ex-post} which showed that \efr{} is not compatible with either \sdef{} or \sdssp{}, and that no mechanism satisfies \efr{}, \etep{}, and \sdsp{}.
Together these negative results demonstrate that \efhr{} mechanisms cannot provide an answer to the question proposed in~\Cref{sec:intro}.

\begin{restatable}{prop}{propimpefr}{}\label{prop:impefr}
     No mechanism simultaneously satisfies \fefr{} (\efr{}), \fsdwef{} (\sdwef{}), and \fetep{} (\etep{}).
\end{restatable}

\begin{proof}

	We prove it using the instance with preference $R$ in~\Cref{fig:fhr}.
	Let $P$ be the random assignment satisfying \efr{} and \etep{}.
	
	First, we look into the deterministic assignments satisfying \fhr{}.
	By \fhr{} implying \fcm{}, if an item is ranked top by some agents, then it should be assigned to one of them.
	Therefore, agents $1$ and $2$ get $a$ and $b$ respectively, which means that $p_{1,a}=p_{2,b}=1$, and one of agents $3$-$6$ gets $c$.
	Since agents $3$-$5$ share the same preference, there are two kinds of assignments satisfying \fhr{}:
	\begin{enumerate}[label=\rm(\roman*),wide,labelindent=0pt,topsep=0em,itemsep=0pt]
	    \item if agent $6$ gets $c$, then $\{d,e,f\}$ can be assigned arbitrarily among agents $3$-$5$;
	    \item if agent $6$ does not gets $c$, then she does not get $e$ or $d$ since $\rk{6}{d}>\rk{j}{d}$ and $\rk{6}{e}>\rk{j}{e}$ with $j\in\{3,4,5\}$, which also means that $p_{6,d}=p_{6,e}=0$.
	\end{enumerate}
	
	Then, by \etep{}, agents $3$-$5$ have the same allocation, and $p_{j,c}=p_{k,c}=1/4$ for any $j,k\in\{3,4,5,6\}$.
	With the observation above, $P$ can only be the assignment int the following.
	
    \begin{center}
    \centering
        \begin{tabular}{r|cccccc}
        	\multicolumn{7}{c}{Assignment $P$} \\
        	&  a & b & c & d & e & f\\\hline
        	$1$ & $1$ & $0$ & $0$ & $0$ & $0$ & $0$\\
        	$2$ & $0$ & $1$ & $0$ & $0$ & $0$ & $0$\\
        	$3$-$5$ & $0$ & $0$ & $1/4$ & $1/3$ & $1/3$ & $1/12$\\
        	$6$ & $0$ & $0$ & $1/4$ & $0$ & $0$ & $3/4$\\
        \end{tabular}
    \end{center}
    
	Assingment $P$ is not \sdwef{} because $\sum_{o'\in\ucs(\succ_6,o)}p_{6,o'}\le\sum_{o'\in\ucs(\succ_6,o)}p_{1,o'}$ holds for any $o\in M$, and it is strict when $o\in\{e,d\}$.
\end{proof}

Since \sdrf{} implies \efr{}, we can extend~\Cref{prop:impefr} to \sdrf{}~(\Cref{cor:impsdrf}).
We also discuss in~\Cref{sec:app:results:rkm} the compatibility of \frkm{}, which also implies~\fcm{}, with fairness, but the result is still negative.

\begin{corollary}\label{cor:impsdrf}
     No mechanism simultaneously satisfies \fsdrf{} (\sdrf{}), \fetep{} (\etep{}), and \fsdwef{} (\sdwef{}).
\end{corollary}

\section{Ex-post Favoring Eagerness for Remaining Items}
Motivated by the desire for \fcm{} mechanisms that are also fair and strategyproof, we propose {\bf\ffhcr{} (\fhcr{})}, an efficiency property which implies both \fcm{}~(\Cref{rm:fhcr}) and \opt{}~(\Cref{prop:fhcr}). 
As we will show in \Cref{sec:ep:ebm}, the ex-post variant of \fhcr{} is compatible simultaneously with fairness (\sdwef{} and \etep{}) and strategyproofness (\sdsp{}).

Informally, a deterministic assignment satisfies \fhcr{} (\Cref{dfn:fhcr}) if it can be decomposed in a manner that every item ranked highest by some agents is allocated to one such agent, subject to which, every remaining item is allocated to a remaining agent who ranks it highest among remaining items if such an agent exists, and so on. 

\begin{definition}[\bf\fhcr{}]\label{dfn:fhcr}
	Given any deterministic assignment $A$, we define for each $r\in\{1,2,\dots\}$ a set of items $\tps{A}{r}=\{o\in M:o=\tp{j}{\allowbreak M\setminus\bigcup_{r'<r}\tps{A}{r'}}\text{ for some }j\in N\text{ with }A(j)\notin\bigcup_{r'<r}\tps{A}{r'}\}$.
	
	The assignment $A$ satisfies {\bf \ffhcr{}} if for every $r\in\{1,2,\dots\}$ and every item $o\in\tps{A}{r}$, it holds that the item $o$ is assigned to an agent most eager for it, i.e, $o=\tp{A^{-1}(o)}{\allowbreak M\setminus\bigcup_{r'\in\{1,2,\dots,r-1\}}\tps{A}{r'}}$.
\end{definition}

\Cref{dfn:fhcr} suggests the following heuristic for designing an \fhcr{} mechanism: In each iteration, remove agents already allocated an item. Then eliminate the allocated items
from the preference lists of every remaining agent. Now allocate each remaining item to an agent who ranks it as the top remaining item according to their preferences over {\em remaining} items, if such an agent exists, using a tie-breaking rule if there are multiple such agents. In contrast, in each iteration of the Boston mechanism which characterizes \fhr{} assignments, the remaining items that are ranked in the highest position by a remaining agent. This notion of iteratively making decisions based on preferences over remaining alternatives is similar in spirit to that of single transferable voting rules~\cite{hare1861treatise} that are resistant to strategic manipulation~\cite{bartholdi1991single} in social choice, and the iterated elimination of dominated strategies for solving strategic games in game theory~\cite{martin2004osborne}. In this vein, \fhcr{} is a natural alternative to \fhr{} since it also implies \fcm{} and \opt{}.

\begin{remark}\label{rm:fhcr}
\fhcr{} implies \fcm{}.
Specifically, in any \fhcr{} assignment $A$, when $r=1$, it requires that for every item $o$ ranked on the top by some agents, i.e., $o\in \tps{A}{1}=\{o\in M:o=\tp{j}{M}\text{ for some }j \in N\}$, item $o$ is allocated to one such agent, i.e., $o=\tp{A^{-1}(o)}{M}$. 
\end{remark}

Although \fhr{} and \fhcr{} both imply \fcm{}, they do not imply each other as we show in~\Cref{eg:fcr}.

\begin{example}\label{eg:fcr}{\rm[\fhr{}$\not\Rightarrow$\fhcr{}, \fhcr{}$\not\Rightarrow$\fhr{}]}
    \rm
	Consider again the profile in~\Cref{fig:fhr}. Let $A$ be the \fhr{} assignment indicated by the circled items, and $A^*$ be the following assignment, where $j\gets o$ means agent $j$ is allocated item $o$:
	\begin{equation*}
		\begin{split}
			A^*:&1\gets a,2\gets b,3\gets c,4\gets e,5\gets f,6\gets d.\\
		\end{split}
	\end{equation*}
	
	It is easy to see that $A$ violates \fhcr{} because item $d\in\tps{A}{2}$ due to the fact that $\tps{A}{1}=\{a,b,c\}$, $d=\tp{6}{M\setminus\tps{A}{1}}$, and $A(6)\notin\tps{A}{1}$; but $A^{-1}(d)=5$ and $d\neq\tp{5}{M\setminus\tps{A}{1}}=e$.
	
	Besides, we show that $A^*$ satisfies \fhcr{}:
	\begin{itemize}[label=-,leftmargin=*,itemsep=0pt,topsep=0pt]
	\item For $r=1$, it is easy to see that for every $o\in\tps{A^*}{1}=\{a,b,c\}$, $o=\tp{j}{M}$ for each $j$ with $A^*(j)=o$.
	\item For $r=2$, $\tps{A^*}{2}=\{e,d\}$. Items $e$ and $d$ are allocated to agents most eager for them among the remaining items $M'=M\setminus{}\tps{A^*}{1}$, i.e., 
	$e=\tp{4}{M'}=A^*(4)$ and $d=\tp{6}{M'}=A^*(6)$.
	\item For $r=3$, $\tps{A^*}{3}=\{f\}$. Since $M''=M\setminus{}\tps{A^*}{1}\cup\tps{A^*}{2}=\{f\}$, we have that $f=\tp{5}{M''}=A^*(5)$ trivially.
	\end{itemize}
	But $A^*$ violates \fhr{} because $A^*(6)=d$, $\rk{6}{d}>\rk{5}{d}$, and $d\succ_{5}A^*(5)$.
\hfill$\square$
\end{example}

\Cref{prop:fhcr} shows that \fhcr{} is a stronger efficiency property than \opt{}, which means that \efcr{} implies \epopt{}.
We also discuss in~\Cref{sec:app:results:relation} the relation of \fhcr{} to \fpop{}~\cite{abraham2007popular} which is also a famous efficiency property.

\begin{restatable}{prop}{propfhcr}{}\label{prop:fhcr}
     {\rm[\fhcr{}$\Rightarrow$\opt{}, \opt{}$\not\Rightarrow$\fhcr{}]} A deterministic assignment satisfying \ffhcr{} (\fhcr{}) also satisfies \fopt{} (\opt{}), but not vice versa.
\end{restatable}
\begin{proof}
	\noindent{\bf (\fhcr{}~$\Rightarrow{}$\opt{})}
	Consider an arbitrary preference profile $R$, and let $A$ be any deterministic assignment that satisfies \fhcr{}. Suppose for the sake of contradiction that $A$ is Pareto dominated by another assignment. Then, since agents have strict preferences, there must exist an assignment $A'$ that Pareto dominates $A$ and can be obtained from $A$ by agents in an {\em improving cycle} exchanging items along the cycle, while all other agents' allocations remain unchanged. More formally, there exists an assignment $A'$ such that a set of $h\le n$ agents $N'=\{j_1,j_2,\cdots,j_h\}$ are involved in an improving cycle where for any $i=1,\dots,h$, $A'(j_i)=A(j_{i+1 \Mod h})\succ_i A(j_i)$, and for every agent $j\in N\setminus N'$, $A'(j)=A(j)$.
	
	For ease of exposition, let the agents in the improving cycle be $N'=\{1,\dots,h\}$, and for any $i=1,\dots,n$, let $o_{i}=A(i)$.
	Without loss of generality, let $o_1$ be the items that belongs to the set $\tps{A}{r}$ with the smallest possible value of $r$ among $\{o_1,\dots,o_h\}$. 
	Then, by $A$ satisfying \fhcr{},
	\begin{equation}\label{eq:prop:fhcr:1}
	    o_1=\tp{A(1)}{M\setminus\bigcup_{r'<r}\tps{A}{r'}}.
	\end{equation}
	
	By our choice of $r$, item $o_{2}\in M\setminus\bigcup_{r'<r}\tps{A}{r'}$, and~Eq~(\ref{eq:prop:fhcr:1}) implies that $o_1\succ_1 o_2$.
	However, by our assumption that $A'$ Pareto dominates $A$, we must have that $o_2=A'(1)\succ_1 A(1)=o_1$, a contradiction. Therefore, any deterministic assignment satisfying \fhcr{} is also \opta{}.
	
	\paragraph{(\opt{}$\not\Rightarrow{}$\fhcr{})}
	For the instance with the following profile $R$ from~\citeA{Ramezanian2021:Ex-post}, the deterministic assignment $A$ is \opt{} since it is an outcome of RP with the priority order $2\impord{}1\impord{}3$.
	
	\vspace{1em}\noindent
	\begin{minipage}{\linewidth}
		\centering
		\begin{minipage}{0.4\linewidth}
			\begin{center}
				$\succ_1$: $a\succ_1 b\succ_1 c$,\\
				$\succ_2$: $a\succ_2 c\succ_2 b$,\\
				$\succ_3$: $b\succ_3 a\succ_3 c$.
			\end{center}
		\end{minipage}
		\begin{minipage}{0.4\linewidth}
			\centering
			\begin{center}
				\centering
				\begin{tabular}{c|ccc}
					\multicolumn{4}{c}{Assignment $A$} \\
					&  a & b & c\\\hline
					1 & $0$ & $1$ & $0$\\
					2 & $1$ & $0$ & $0$\\
					3 & $0$ & $0$ & $1$\\
				\end{tabular}
			\end{center}
		\end{minipage}
	\end{minipage}\vspace{1em}
	
	We see that $b\in \tps{A}{1}=M$ since $\tp{3}{M}=b$.
	However, $A^{-1}(b)=1$ and $b\neq\tp{1}{M}=a$, which violates \fhcr{}.
\end{proof}

\subsection{\am{} Satisfies \efcr{}, \sdwef{}, \etep{}, and \sdsp{}}\label{sec:ep:ebm}

In this section, we define the {\em eager Boston mechanism} (\am{},~\Cref{alg:am}), and prove that it is efficient (\efcr{} and therefore ex-post \fcm{} and \epopt{}), fair (\sdwef{} and \etep{}), and strategyproof (\sdsp{}). \am{} proceeds in multiple rounds using \fhcr{} as a heuristic to allocate items.
In each round, each unsatisfied agent $j$ applies for the item that she is most eager for, i.e., her top remaining item $o$.
We use $N_o$ to refer to the set of agents who apply for $o$. Every agent in $N_o$ gets $o$ with probability $1/\lvert N_o\rvert $, with the winner determined by a random lottery winner generator $G$. Given a set of agents $S\subseteq N$, $G(S)$ is a single agent drawn from $S$ uniformly at random. At the end of each round, for every item $o$ with $N_o\neq\emptyset$, both the item $o$ and the winner $G(N_o)$ are removed.
We illustrate the execution of \am{} in~\Cref{eg:amnfhr}.
The outcome \am$(R)$ is a deterministic assignment, which can be computed in polynomial time if $G$ runs in polynomial time as we show in~\Cref{sec:app:results:time}.

\begin{algorithm}[htb]
	\begin{algorithmic}[1]
		\State {\bf Input:} An assignment problem $(N,M)$, a strict linear preference profile $R$, and a lottery winner generator $G$.
		\State $M'\gets M$. $N'\gets N$. $A\gets 0^{n\times n}$.
		\While{$M'\neq\emptyset$}
		\For{each $o\in M'$}
		\State$N_o\gets\{j\in N'\mid \tp{j}{M'}=o\}$.
		\State Run a lottery over $N_o\neq\emptyset$ to pick an agent $j_o=G(N_o)$, and allocate $o$, $A_{j_o,o}\gets1$.
        \EndFor
		\State $M'\gets M'\setminus\{o\in M'\mid N_o\neq\emptyset\}$.
		$N'\gets N'\setminus{\cup_{o\in M'}\{j_o\}}$.
		\EndWhile
		\State \Return $A$
	\end{algorithmic}
	\caption{\label{alg:am} Eager Boston mechanism (\am{})}
\end{algorithm}

\begin{example}\label{eg:amnfhr}\rm
	We execute \am{} on the instance in~\Cref{fig:fhr}. The table below shows for each round, which item each agent applies for, and a `/' represents the fact that an agent does not apply for any item since she has already been allocated one. The circled items represent the allocation of an item to the lottery winner.
		\begin{center}
		\centering
		\begin{tabular}{c|cccccc}
			\diagbox{Round}{Agent}&  1 & 2 & 3 & 4 & 5 & 6\\
			\hline
			1 & \bcircled{a} & \bcircled{b} & \bcircled{c} & c & c & c\\
			2 & / & / & / & \bcircled{e} & e & \bcircled{d}\\
			3 & / & / & / & / & \bcircled{f} & /\\
			\hline
		\end{tabular}
	\end{center}

	\noindent{-}~At round $1$, agents $1$ and $2$ apply for $a$ and $b$, respectively, and win them since they are the only applicants, while agents $3$~-~$6$ apply for $c$ and enter a lottery with equal chances of winning.
	
	\noindent{-}~If~agent $3$ wins $c$ at round $1$, then at round $2$, agents $4$ and $5$ apply for $e$, while agent $6$ applies for $d$ alone and gets it.
	
	\noindent{-}~If~agent $4$ wins $e$ at round $2$, agent $5$ applies for and gets $f$ at round $3$.
	
	Then, \am{} outputs the assignment $A^*$ in~\Cref{eg:fcr}.
\hfill$\square$
\end{example}

The expected outcome of \am{} is a random assignment, which we refer to as $\mathbb E(\text{\am}(R))$. Therefore, \am{} can also be viewed as a random mechanism. We prove that \am{} satisfies ex-post \fhcr{} (\efcr{}, \Cref{dfn:efcr}), \sdwef{}, and \sdsp{} in \Cref{thm:amp}.
{\bf All the missing proofs can be found in \Cref{sec:app:proof}.} 

\begin{definition}\label{dfn:efcr}
    A random assignment $P$ satisfies \fefcr{} (\efcr{}) if it is a convex combination of \fhcr{} deterministic assignments, i.e., if $P=\sum_{A\in \ma'}\alpha_A*A$, where $\ma'\subseteq\ma$, $\sum_{A\in\ma'}\alpha_A=1$ and every $A\in\ma'$ satisfies \fhcr{}.
\end{definition}

\begin{restatable}{thm}{thmamp}{}\label{thm:amp}
	\am{} satisfies \fefcr{} (\efcr{}), \fsdwef{} (\sdwef{}),  \fetep{} (\etep{}), and \fsdsp{} (\sdsp{}).
\end{restatable}
\begin{sketch}
    Given any profile $R$, let $P=\mathbb E(\am{}(R))$. For convenience, we refer to each possible execution of \am{}, i.e., each way in which lottery winners are picked, as a possible {\em world} below.
    
    \vspace{0.5em}\noindent{\bf(\efcr{})}
    Let $A=\text{\am}(R)$.
	We show that the following two conditions hold for each $r\ge1$:
	\begin{enumerate}[label=(\arabic*),leftmargin=*,topsep=0.5em,itemsep=0pt]
	\item the set of items assigned at each round $r$ of~\Cref{alg:am}, i.e., $\{o\in M'|N_o\neq\emptyset\}$, are exactly those in $\tps{A}{r}$, and
	\item the assignment $A$ allocates every item $o\in \tps{A}{r}$ to an agent who ranks $o$ as the top item in the set of remaining items $M'$, i.e. an agent in $\{j\in N'|\tp{j}{M'}=o\}$.  
	\end{enumerate}
	When $r=1$, we obtain condition (1) trivially with $M'=M$ and $N'=N$. 
	Condition (2) holds because at the beginning of round~$1$, we have that for any $o\in\tps{A}{1}$, $o=\tp{j_o}{\tps{A}{1}}$ with $j_o=A^{-1}(o)$ by Line~6 of~\Cref{alg:am}, which means that $N_o\neq\emptyset$ and $o$ is assigned to $j_o$ who ranks $o$ as the top item among $M$. 
	Before round $2$, by~Line~7, every such item $o$ and its winner $j_o$ are removed from $M'$ and $N'$ respectively . With the updated $M'$ and $N'$, we can obtain the two conditions hold for $r=2$ with a similar analysis.
	The proof follows by repeating a similar argument at each subsequent round.
	In this way we see that $A$ is \fhcr{}, and therefore $P$ is \efcr{}.

    \vspace{0.5em}\noindent{\bf(\sdwef{})}
	For any pair of agents $j,k$ with $P_{k}\sd{j}P_j$, we show that the following two conditions hold at any round $r$ during the execution of~\Cref{alg:am}, where $k$ has not been allocated an item yet:
	\begin{enumerate}[label=(\arabic*),leftmargin=*,topsep=0.5em,itemsep=0pt]
		\item if $j$ applies for $o$, then $k$ also applies for $o$, and
		\item if $j$ gets some item $o$ at round $r'<r$, then $k$ applies for the item that $j$ ranks highest among remaining items.
	\end{enumerate}
	Condition (1) shows that at any round where both agents are unassigned, they apply for the same item and therefore have the same chance to win it.
	Condition (2) shows that if $j$ gets an item at an earlier round, then $k$ applies to the item $j$ would have applied to had $k$ been allocated an item in an earlier round.
	The proof proceeds by comparing the probabilities of the worlds in which $j$ and $k$ get $o$ respectively, and shows that they are equal for every item $o$, considered one by one according to the preference order  $\succ_j$, from which it follows that $P_{k}=P_j$.

	\vspace{0.5em}\noindent{\bf (\etep{})}.
	We prove it by comparing the probabilities that agent $j$ and $k$ get each $o\in\ucs(\succ_{j,k},o_m)$.
	First, for the world $w$ where $j$ gets $o$ at round $r$ while $k$ gets $o'\in\ucs(\succ_{j,k},o_m)$ at round $r'$, we can find out another world $w'$ where only $j$ and $k$ swap their items.
	It follows that $Pr(w)= Pr(w')$ since the other lotteries keep the same as $w$.
	Then for the worlds $W_j$ where $j$ gets $o$ at round $r$ and $k$ does not get items in $\succ_{j,k}$, we can also construct another set of worlds $W_k$ such that: $k$ gets $o$ at round $r$, and $j$ participates in lotteries instead from round $r+1$ to the last round that $k$ applies for items in $\succ_{j,k}$.
	In this way, we also obtain that $Pr(W_j)= Pr(W_k)$.
	With both cases hold we make the proof.
	
	\vspace{0.5em}\noindent{\bf (\sdsp{})}
	Let $R'=(\succ'_j,\succ_{-j})$ be the profile when agent $j$ misreports her preferences as $\succ'_j$, $Q=\am(R')$, and assume that $Q_j\sd{j} P_j$.
	The proof proceeds by considering each item $o$ according to the order $\succ_j$, and shows that if $j$ applies for $o$ at round $r$ in some world $w$ for $\am{}(R)$, then $j$ also applies for $o$ at round $r$ in any world with the lotteries and winners before round $r$ identical to those of $w$ for $\am{}(R')$.
	This means that despite misreporting, $j$ applies for the same items as she does when truthfully reporting her preferences, and therefore $j$ has the same probability to win each item.
	It follows that $p_{j,o}=q_{j,o}$ for each $o\in M$, and therefore, that if $Q_j\sd{} P_j$, then $Q_j=P_j$.
\end{sketch}

\subsection{\Efcr{} and Adaptive Boston Mechanism}\label{sec:ep:abm}

We now show that not only is every member of adaptive Boston mechanism (\abm{}) guaranteed to output an \efcr{} assignment, but also that every \efcr{} assignment can be computed by some member of \abm{}, meaning that the output of \am{} must also be the output of some member of \abm{} which depends on the instance of the assignment problem. As we show in~\Cref{rmk:amnotabm}, although \am{} appears similar to \abm{}, \am{} does not belong to the family of \abm{} mechanisms.

Each algorithm in \abm{} (\Cref{alg:abm}) is specified by a probability distribution $\pi$ over the priority orderings of agents, and is denoted \abm{}$^{\pi}$, which computes an assignment as follows.
First, a priority order $\impord{}$, a strict linear order over $N$ is picked according to the probability distribution $\pi$. Then, items are allocated to agents in multiple rounds. In each round, each unsatisfied agent applies for a remaining item that she is most eager for. Each remaining item $o$, if it has applicants, is assigned to the agent $j_o$ who is ranked highest in $\impord{}$ among all the applicants. At the end of each round, every such item $o$ and the corresponding $j_o$ are removed from $M$ and $N$, respectively. 

\begin{algorithm}[htb]
	\begin{algorithmic}[1]
		\State {\bf Input:} An assignment problem $(N,M)$, a strict linear preference profile $R$, a probability distribution $\pi$ over all priority orderings of agents.
		\State $M'\gets M$. $N'\gets N$. $A\gets 0^{n\times n}$.
        \State Randomly choose a priority order $\impord{}$ according to $\pi$.
		\While{$M'\neq\emptyset$}
		\For{each $o\in M'$}
		\State$N_o\gets\{j\in N'\mid\tp{j}{M'}=o\}$.
		\State Allocate $o$ to agent $j_o$ which is ranked highest in $\impord{}$ among $N_o$, i.e., $A_{j_o,o}\gets1$.
        \EndFor
		\State $M'\gets M'\setminus\{o\in M'\mid N_o\neq\emptyset\}$.
		$N'\gets N'\setminus{\cup_{o\in M'}\{j_o\}}$.
		\EndWhile
		\State \Return $A$
	\end{algorithmic}
	\caption{\label{alg:abm} Adaptive Boston mechanism (\abm)}
\end{algorithm}

\Cref{thm:abmchar} shows that \efcr{} characterizes the family of \abm{} algorithms. Throughout, we will use $\pi(\impord{})$ to denote the probability of a priority order $\impord{}$ according to the distribution $\pi$. If  $\pi(\impord{})=1$ for a certain $\impord{}$, we will use \abm{}$^{\impord{}}$ to refer to the corresponding algorithm for convenience.

\begin{restatable}{thm}{thmabmchar}{}\label{thm:abmchar}
	Given a profile $R$, a random assignment $P$ satisfies \fefcr{} (\efcr{}) if and only if there exists a probability distribution over all the priorities $\pi$ such that $P=\mathbb E(\text{\abm{}}^{\pi}(R))$.
\end{restatable}

\begin{proof}

{\bf\noindent(Satisfaction)}
The proof is similar to proving \am{} satisfies \efcr{}, and is provided in~\Cref{sec:app:proof} for the sake of completeness.

\vspace{0.5em}{\bf\noindent(Uniqueness)}
Consider an arbitrary random assignment $P$ satisfying \efcr{} for a preference profile $R$. Then, $P$ can be decomposed into a set $\ma'\subseteq\ma$ of deterministic assignments satisfying \fhcr{} with positive probability, i.e., $P=\sum_{A_i\in \ma'}\alpha_i*A_i$, where $\alpha_i>0$. 

Consider any $A\in\ma'$. Since $A$ satisfies \fhcr{}, by \Cref{dfn:fhcr}, there exist non-empty sets $\tps{A}{1},\dots,\tps{A}{K}$, such that for each $r\in\{1,\dots,K\}$, $\tps{A}{r}=\{o\in M:o=\tp{j}{M\setminus\bigcup_{r'<r}\tps{A}{r'}}\}$ for some $j\in N$ with $A(j)\not\in\bigcup_{r'<r}\tps{A}{r'}$.

Consider any priority ordering $\impord{}$ where for any $r',r\in\{1,\dots,K\}$ with $r'<r$, and any pair of items $o'\in\tps{A}{r'}$ and $o\in\tps{A}{r}$, it holds that agent $A^{-1}(o')$ has higher priority than $A^{-1}(o)$, denoted as $A^{-1}(o')\impord{}A^{-1}(o)$. It is easy to see that since $A$ is deterministic, and every agent receives exactly one item, at least one such priority ordering always exists. 

Let $B=\abm{}^{\impord{}}(R)$. We claim that at any round $r$ during the execution of $\abm{}^{\impord{}}(R)$, every item in $o\in\tps{A}{r}$ is allocated to $A^{-1}(o)$, i.e., $B^{-1}(o)=A^{-1}(o)$. It is easy to see that the claim is true for $r=1$. 
Since $A$ satisfies \fhcr{}, for any item $o\in\tps{A}{1}$, which is the set of items that are ranked on top by some agent, $A^{-1}(o)$ ranks $o$ as her top item, i.e., $A^{-1}(o)\in N_o=\{j\in N\mid \tp{j}{M}=o\}$, and therefore she applies for $o$ at round $1$.
Due to the construction of $\impord{}$, $A^{-1}(o)$ must have the highest priority among $N_o$ and obtain item $o$, i.e., $B^{-1}(o)=A^{-1}(o)$. 

Now, assume that it holds that at any round $r'<r$, every $o'\in\tps{A}{r'}$ is allocated to $A^{-1}(o')$, i.e., $B^{-1}(o')=A^{-1}(o')$. 
We show that at round $r$, any $o'\in\tps{A}{r}$ is allocated to $A^{-1}(o)$, i.e., $B^{-1}(o)=A^{-1}(o)$. 
Assume for the sake of contradiction that there exists an item $o\in\tps{A}{r}$, such that $B^{-1}(o)=k\neq j=A^{-1}(o)$. By our assumption about rounds $r'<r$, both $j$ and $k$ have not been assigned an item in an earlier round by $\abm{}^{\impord{}}(R)$. Notice that by Line 6 of \Cref{alg:abm}, $\tp{k}{M\setminus\bigcup_{r'<r}\tps{A}{r'}}=o$, since every item in $\bigcup_{r'<r}\tps{A}{r'}$ is allocated in an earlier round by our assumption. Also, since $A$ is \fhcr{}, we also have that $\tp{j}{M\setminus\bigcup_{r'<r}\tps{A}{r'}}=o$. Therefore, both $j$ and $k$ apply for item $o$ during round $r$ of $\abm{}^{\impord{}}(R)$. Then, it must hold that $k\impord{} j$ since $k$ is assigned $o$ in round $r$ of the execution of $\abm{}^{\impord{}}(R)$. 

However, by the construction of $\impord{}$ and the assumption that $j=A^{-1}(o)$, $k\impord{} j$ implies that there exists some $r^*<r$ such that $k$ gets an item in $\tps{A}{r^*}$. Then, there must exist some item $o^*\in\tps{A}{r^*}$ where $r^*<r$ such that $k=A^{-1}(o^*)$. It also means that $o^*\neq o=B(k)$, and therefore $B^{-1}(o^*)\neq=k=A^{-1}(o^*)$, a contradiction to our assumption that $B^{-1}(o')=A^{-1}(o')$ for every $o'\in\tps{A}{r'}$ with $r'<r$. Thus, by induction, it holds that $B=A$.

We have shown that for every deterministic assignment $A_i\in\ma'$, there exists a priority order $\impord{i}$ such that the output of $\abm{}^{\impord{i}}(R)=A_i$. Then, for the \efcr{} assignment $P=\sum_{A_i\in\ma'}\alpha_i A_i$, is the output of a member of the family of \abm{} algorithms specified by the probability distribution $\pi$ over priority orderings where for any $A_i\in\ma'$, $\pi(\impord{i})=\alpha_i$, i.e., $\abm{}^{\pi}(R)=P$. This completes the proof.
\end{proof}

\begin{remark}\label{rmk:amnotabm}
\rm
\citeA{Mennle2021:Partial} proved that \abm{}$^\pi$ is \sdsp{} if $\pi(\impord{})>0$ for every $\impord{}$.
However, as we show in~\Cref{sec:app:results:abm}, there is no such a distribution $\pi$ that $\am{}(R)=\abm{}^{\pi}(R)$ for every preference profile $R$, meaning that \am{} is not a member of \abm{}. Therefore, the sufficient condition proposed by~\citeA{Mennle2021:Partial} for checking if a member of \abm{} satisfies \sdsp{} is not applicable for the proof of \am{} satisfying \sdsp{}.
\end{remark}

\section{Ex-ante Favoring Eagerness for Remaining Items}

In this section, we design a family of mechanisms that satisfy the ex-ante variant of \fhcr{}, {\bf \fsdcrf{} (\sdcrf{})}, which implies ex-post \fcm{} (\Cref{rm:sdcrf}), \epopt{}, and \sdopt{} (\cref{prop:sdcrf}).
We further prove that a member of this family of mechanisms satisfies \sdwef{} and \etep{}.

Intuitively, a random assignment satisfies \sdcrf{}, if the shares of every remaining item are distributed among only the agents who are most eager for it unless every such agent's demand has been satisfied by better items.

\begin{definition}[\bf\sdcrf{}]\label{dfn:sdcrf}
	Given a random assignment $P$, we defined $M_{P,0}=\emptyset$ and for each item $o\in M$, $E_{P,0}(o)=\emptyset$. Then, for each $r\in\{1,2,\dots$\}, we define
	\begin{enumerate}[label=\rm(\roman*),wide,labelindent=0pt,topsep=0em,itemsep=0pt]
	    \item the set of items with positive supply after excluding the shares owned by agents in $\bigcup_{r'<r}E_{P,r'}(o)$, $M_{P,r}=\{o\in M:\sum_{k\in\bigcup_{r'<r}E_{P,r'}(o)}p_{k,o}<1\}$, and
	    \item for each item $o\in M_{P,r}$, $E_{P,r}(o)=\{j\in N:o=\tp{j}{M_{P,r}}\}$ to be the set of all agents eager for it.
	\end{enumerate}
	
	A random assignment $P$ satisfies {\bf\fsdcrf{} (\sdcrf{})} if for every $r\in\{1,2,\dots\}$ and item $o\in M_{P,r}$, it holds for any agent $j\in N$ that if there exists an $r'<r$ such that $j\in E_{P,r'}(o)$, agent $j$ is satisfied by items weakly preferred to $o$, i.e., $\sum_{o'\in\ucs(\succ_j,o)}p_{j,o'}=1$.
\end{definition}

\Sdcrf{} is a natural ex-ante extension of \fhcr{} for random assignments. 
Indeed, it is easy to see that for deterministic assignments, \sdcrf{} is equivalent to \fhcr{}. 
We also note that \sdcrf{} and \efcr{} do not imply each other. 
Please see~\Cref{sec:app:results:relation} for more details.

\begin{remark}\label{rm:sdcrf}
\Sdcrf{} implies ex-post \fcm{}.
Specifically, in any \sdcrf{} assignment $P$, for item $o\in M$ which is ranked top by some agents, i.e., $o\in M_{P,1}$ and $E_{P,1}(o)\neq\emptyset$, we have that $o\notin M_{P,r}$ with $r>1$ by \sdcrf{} and the fact that $\sum_{o'\in\ucs(\succ_j,o)}=0$ for $j\in E_{P,1}(o)$, and it follows that every such $o$ is allocated to one of agents in $E_{P,1}(o)$ who rank $o$ top among all the items in deterministic assignments that constitute a convex combination for $P$.
\end{remark}


\Cref{prop:sdcrf} below shows that \sdcrf{} implies \sdopt{}.

\begin{restatable}{prop}{propsdcrf}{}\label{prop:sdcrf}
     {\rm[\sdcrf{}$\Rightarrow$\sdopt{}, \sdopt{}$\not\Rightarrow$\sdcrf{}]}
     A random assignment satisfying \fsdcrf{} (\sdcrf{}) also satisfies \fsdopt{} (\sdopt{}), but not vice versa.
\end{restatable}
\begin{proof}
	{\bf(\sdcrf{}~$\Rightarrow{}$\sdopt{})}
	Assume for the sake of contradiction that $P$ is \sdcrfa{}, but not \sdopta{}.
	By assumption and~\Cref{lem:BM} (in \Cref{sec:app:results:properties}), we can find a set of agents $\{j_1,j_2,\cdots,j_h\}$ and items $M^*=\{o_1,o_2,\cdots,o_h\}$ such that $o_{i+1\Mod h}\succ_{j_i} o_i$ with $p_{j_i,o_i}>0$ with $i\le h$.
	For each $o_i$, let $r_i$ be the round where $j_i$ consumes it.
	We note that $r_i$ is unique for each $o_i$, because if item $o_i$ is consumed by $j_i$ at that round, then either $o_i$ is consumed to exhausted or $j_i$ is satisfied and does participate the later consumption.
	Without loss of generality, let $o_{i'}=\arg\min_{o_i\in M^*}r_i$.
	Then we have that all the items in $M^*$ are available at round $r_{i'}$ and $j_{i'}\in N_{o_{i'}}$, which means that $\tp{j_{i'}}{M'}=o_{i'}$ where $M'$ is the set of all the available items at that round in~\Cref{alg:pcr}.
	Since $M^*\subseteq M'$, we have that $o_{i'}\succ_{i'}o_{i'+1}$, a contradiction to the assumption.
	
	\vspace{0.5em}\noindent{\bf (\sdopt{}$\not\Rightarrow{}$\sdcrf{})}
	For the following preference profile $R$, the  assignment $P$ is the outcome of PS which satisfies \sdopt{}:
	
	\vspace{1em}\noindent
	\begin{minipage}{\linewidth}
		\centering
		\begin{minipage}{0.4\linewidth}
			\begin{center}
				$\succ_1$: $a\succ_1 b\succ_1 c$,\\
				$\succ_2$: $a\succ_2 c\succ_2 b$,\\
				$\succ_3$: $b\succ_3 a\succ_3 c$.
			\end{center}
		\end{minipage}
		\begin{minipage}{0.4\linewidth}
			\centering
			\begin{center}
				\centering
				\begin{tabular}{c|ccc}
					\multicolumn{4}{c}{Assignment $P$}\\
					&  a & b & c\\\hline
					1 & $1/2$ & $1/4$ & $1/4$\\
					2 & $1/2$ & $0$ & $1/2$\\
					3 & $0$ & $3/4$ & $1/4$\\
				\end{tabular}
			\end{center}
		\end{minipage}
	\end{minipage}\vspace{1em}
	
	We see that $E_{P,1}(b)=\{3\}$, $\sum_{o\in\ucs(\succ_3,b)}=3/4<1$, and $b\in M_{P,2}$, which violates \sdcrf{}.
\end{proof}

\subsection{\upre{} Satisfies \sdcrf{}, \sdwef{}, and \etep{}}\label{sec:ea:pre}


We propose the family of probabilistic respecting eagerness mechanisms (\pcr{}) defined in Algorithm~\ref{alg:pcr}, and show that the uniform probabilistic respecting eagerness mechanism (\upre), a member of \pcr{}, satisfies the desirable fairness notions \sdwef{} and \etep{} (\Cref{thm:uprep}). We also prove that not only does every \pcr{} mechanism satisfy \sdcrf{}, but also every \sdcrf{} assignment must be the output of some member of \pcr{} (\Cref{thm:familychar}).
Therefore, \upre{} satisfies \sdcrf{}, \sdwef{}, and \etep{}. 

Each member of the \pcr{} family of mechanisms is specified by a parameter $\omega=(\omega_j)_{j\in N}$, and denoted $\pcr{}_\omega$. Each $\omega_j$ is an eating speed function which maps each time instance $t$ to a rate of consumption for agent $j$ such that $\int_{0}^{1}\omega_{j}(t){\rm d}t=1$, $\omega_{j}(t)\ge0$ for $t\in [0,1]$, and $\omega_{j}(t)=0$ for $t>1$, i.e., agent $j$ consumes exactly one unit of item during one unit of time.
At the beginning of execution, we set $s(o)=1$ to refer to the supply of item $o$, and set $t_j=0$ to indicate the elapsed time each agent $j$ has spent on consumption.
At each round $r$, each agent $j$ determines $\tp{j}{M'}$, her top item among the set $M'$ consisting of every item $o$ which remains available, i.e., all items with $s(o)>0$.
For each item $o\in M'$, $N_o$ is the set of agents for whom $o$ is the top item. All the agents in $N_o$ consume $o$ together for $\gamma_\omega(N_o, (t_j)_{j\in N}, s(o))$ units of time. For any $N'\subseteq N$, elapsed consumption times $(t_j)_{j\in N}$, and supply $s$, we define:
\begin{equation}\label{eq:precon}
	\begin{split}
		&\gamma_\omega(N', (t_j)_{j\in N}, s)=\min\Big\{\{\rho\mid\sum_{k\in N'}\int^{t_k+\rho}_{t_k}\omega_{k}(t){\rm d}t= s\}\\
		&\cup\{\rho\in[0,1]\mid\sum_{k\in N'}\int^{t_k+\rho}_{t_k}\omega_{k}(t){\rm d}t=\sum_{k\in N'}\int^{1}_{t_k}\omega_{k}(t){\rm d}t\}\Big\},
	\end{split}
\end{equation}
Notice that for any agent $j$, $\int^{1}_{t_j}\omega_{j}(t){\rm d}t$ refers to her remaining demand.
In words, Eq~(\ref{eq:precon}) requires that agents in $N_o$ stop their consumption when either the supply of $o$ is exhausted, or all of them are satisfied.
Then the amount that agent $j$ consumes at this round is the shares of $o$ she gets in the final outcome, and we update the supply $s(o)$ and the elapsed time $t_j$.

\begin{algorithm}[htb]
	\begin{algorithmic}[1]
		\State {\bf Input:} An assignment problem $(N,M)$, a strict linear preference profile $R$, a collection of eating functions $\omega=(\omega_{j})_{j\in N}$.
		\State $M'\gets M$, $P\gets 0^{n\times n}$, $s(o)\gets 1$ for every $o$, and $t_j\gets 0$ for every $j$.
		\While{$M'\neq\emptyset$}
		\State $N_{o}\gets\{j\in N\mid \tp{j}{M'}=o\}$.
		\For{each item $o\in M'$}
		\State Agents in $N_{o}$ {\bf consume} $o$.
		\begin{enumerate}[label=6.\arabic*:,itemindent=2.5em]
			\item        $\rho_o\gets\gamma_{\omega}(N_o,(t_j)_{j\in N}, s(o))$.
			\item For each $j\in N_{o}$, $p_{j,o}\gets \int^{t_j+\rho_o}_{t_j}\omega_{j}(t){\rm d}t$.
		\end{enumerate}
		\EndFor
		\State $s(o)\gets s(o)-\sum_{k\in N_{o}}\int^{t_k+\rho_o}_{t_k}\omega_{k}(t){\rm d}t$.
		\State For each $j\in N_{o}$, $t_j\gets t_j+\rho_o$.
		\State $M'\gets M'\setminus\{o\in M'\mid s(o)=0\}$.
		\EndWhile
		\State \Return $P$
	\end{algorithmic}
	\caption{\label{alg:pcr} Probabilistic respecting eagerness (\pcr{})}
\end{algorithm}

We define the {\em uniform probabilistic respecting eagerness} mechanism (\upre{}) to be the member of \pcr{} (defined in~\Cref{alg:pcr}) in which every agent consumes items uniformly at the same speed during the time period $[0,1]$. 
We prove in~\Cref{thm:uprep} that \upre{} satisfies \sdwef{} and \etep{}. 
As a member of \pcr{}, \upre{} also satisfies \sdcrf{} by~\Cref{thm:familychar}. 
We demonstrate the execution of \upre{} in \Cref{eg:pcrnefr}, and also show that \upre{} runs in polynomial time in~\Cref{sec:app:results:time}.

\begin{definition}[\bf\upre{}]\label{dfn:upre}
    The {\bf uniform probabilistic respecting eagerness mechanism} (\upre) is a member of \pcr{}, where every agent eats at a uniform eating speed of one unit of item per one unit of time, i.e., for each $j\in N$,
    \begin{equation}\label{eq:eatingfunction}
		 \omega_j(t)=
		\begin{cases}
			
			1, & \mbox{$t\in[0,1]$,}\\
			
			0, & \mbox{$t>1$.}
			
		\end{cases}
	\end{equation}
\end{definition}

\begin{example}\label{eg:pcrnefr}\rm
    In \Cref{fig:ref6ag} and the discussion below, we illustrate the execution of \upre{} applied to the profile in~\Cref{fig:fhr}.
	
	\noindent{-}~At round $1$, agents $1$ and $2$ consume $a$ and $b$ respectively, and other agents consume $c$.
	After consumption, agents $1$ and $2$ fully get $a$ and $b$, respectively, and the other four agents each get $1/4$ units of $c$.
	The supply of each consumed item is updated as $s(a)=s(b)=s(c)=0$.
	
	\noindent{-}~At round $2$, agents $3$~-~$5$ consume $e$ and each get $1/3$ units such that $s(e)=0$, and agent $6$ consumes $d$ till satisfied and gets $3/4$ units, leaving $s(d)=1/4$.
	
	\noindent{-}~At round $3$, agents $3$~-~$5$ consume $d$ and each get $1/12$ units each such that $s(d)=0$.
	
	\noindent{-}~At round $4$, agents $3$~-~$5$ consume $f$ and get $1/3$ units each.
	
	The final output is assignment $P$ in~\Cref{prop:impefr}. 
\hfill$\square$
\end{example}

\begin{figure}[t]
	\centering
	\includegraphics[width=0.8\linewidth]{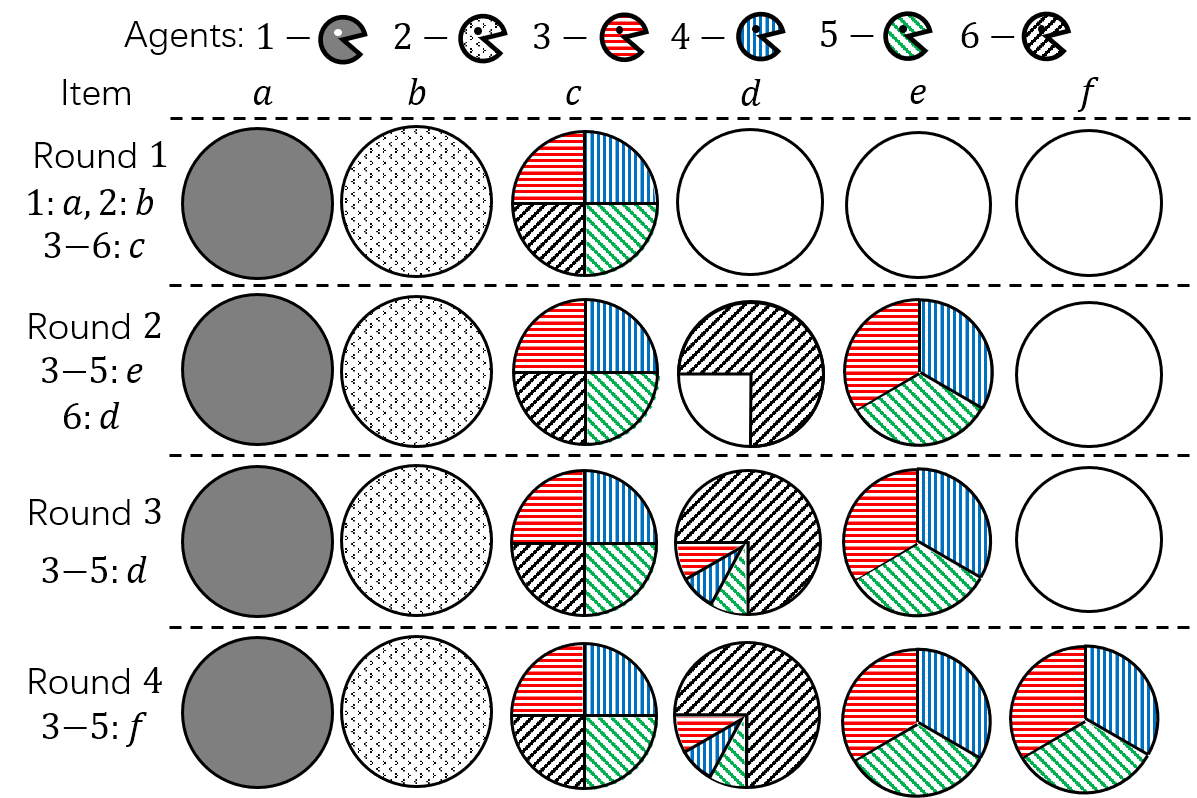}
	\caption{An example of the execution of a member of \pcr{} (\upre).}
	\label{fig:ref6ag}
\end{figure}

\begin{restatable}{thm}{thmuprep}{}\label{thm:uprep}
	\upre{} satisfies \fsdwef{} (\sdwef{}) and \fetep{} (\etep{}).
\end{restatable}
\begin{sketch}
	{\bf (\sdwef{})}
	Given a profile $R$, let $P=\text{\upre}(R)$.
	For agents $j$ and $k$ with $P_k \sd{j} P_j$, we show that at each round $r$, agent $k$ applies for the same item $o$ as $j$ because otherwise,
	\begin{enumerate}[label=(\arabic*),leftmargin=*,labelindent=0pt,topsep=0em,itemsep=0pt]
	\item if $o$ is consumed to exhaustion by $j$ and other agents, then agent $k$ can never obtain shares of $o$ at later rounds, and
	\item if $j$ is satisfied upon consuming $o$, i.e., satisfied by shares of items in $\ucs(j,o)$, then agent $k$ must get shares of an item ranked below $o$ (according to $\succ_j$).
	\end{enumerate}
	
	Both cases above imply that $P_j \sd{j} P_k$ while $P_j\neq P_k$, a contradiction. Further, this means that at the end of each round $r$, both agents have the same elapsed consumption time, i.e., $t_k=t_j$, and therefore both agents consume the same shares of items since they have the same eating speeds.
	Together, this means $p_{j,o}=p_{k,o}$ for any $o$ they have consumed, and it follows that $P_k=P_j$.
	
	\noindent{\bf (\etep{})}
	We show that before consuming items not in $\succ_{j,k}$, agent $j$ and $k$ consume the same item $o$ at each round $r$.
	If agent $k$ consumes $o'\neq o$, then let $o\succ o'$ without loss of generality, which means that $k$ does not consume the top item at round $r$, a contradiction to the execution of \upre{}.
	Since the agents consume the same item at a round for the same length of time, we obtain that $p_{j,o}=p_{k,o}$ for each $o$ appearing in $\succ_{j,k}$.
\end{sketch}


\begin{restatable}{thm}{thmfamilychar}{}\label{thm:familychar}
	Given a profile $R$, a random assignment $P$ satisfies \fsdcrf{} (\sdcrf{}) if and only if there exists an eating speed function $\omega$ such that $P=\pcr_{\omega}{}(R)$.
\end{restatable}

\begin{sketch}
    {\bf (Satisfaction)}
    Let $P=\pcr{}_{\omega}(R)$ where $\omega$ is any collection of eating functions.
	We show that at every round $r$,
	\begin{enumerate}[label=(\arabic*),wide,leftmargin=*,labelindent=0pt,topsep=0em,itemsep=0pt]
	    \item $M_{P,r}$ and $E_{P,r}(o)$ are exactly the sets of items with remaining shares and agents who consume $o$, respectively, and
	    \item if $o$ is available for later rounds, then all agents in $E_{P,r}(o)$ are satisfied.
	\end{enumerate}
    We see that (1) is trivially true for round $1$.
    Agents in $E_{P,1}(o)$ are those in the set $N_o$ on~Line~4 of~\Cref{alg:pcr}.
    Also, they consume $o$, and stop as soon as either $o$ is exhausted, or they are satisfied according to the consumption process, which means (2) is true for round $1$.
    By~Lines~7 and~9, we see that $s(o)$ is updated to ensure that $M'$ does not contain item $o'$ with $\sum_{k\in E_{P,1}(o')}p_{k,o'}=0$.
	With the updated $M'$, the condition (1) holds for $r=2$ trivially, and we can prove condition (2) according to the selection of items to be consumed at round $2$.
	The proof follows from an inductive argument along similar lines.
	
	\vspace{0.5em}\noindent{\bf (Uniqueness)}
	For any \sdcrf{} assignment $Q$, we find out a member of \pcr{} with the following eating function such that $P=$\pcr{}$_\omega(R)$ coincides with $Q$.
	\begin{equation*}
		\omega_j(t)=
		\begin{cases}
			
			n~\cdot~q_{j,o}, & {\mbox{$t\in[\frac{r-1}{n},\frac{r}{n}],$ where $r=\min(\{\hat{r}\mid j\in E_{Q,\hat{r}}(o)\})$, }}\\
			
			0, & \mbox{others.}
			
		\end{cases}
	\end{equation*}
	
    Constructing such a eating function is to ensure the following conditions successively for each round $r$ in the execution of \pcr{}$_\omega(R)$:
    
    \begin{enumerate}[label=(\arabic*),leftmargin=*,labelindent=0em,topsep=0em,itemsep=0pt]
        \item items with remaining shares are the same, i.e., $M_{Q,r}=M_{P,r}$, and each agent is eager for the same item, i.e., $E_{Q,r}(o)=E_{P,r}(o)$;
        \item for any an unsatisfied agent $j$ which is going to consume item $o$ at round $r$, she starts consumption at time $t_j=(r-1)/n$, and exactly consumes for $\rho_o=1/n$ long; and
        \item after the consumption, for any $o\in M_{Q,r}$, each agent $j$ eager for it obtains $q_{j,o}$ units of shares of item $o$, i.e., $p_{j,o}=q_{j,o}$.
    \end{enumerate}
    
    It is easily to see that the condition (1) above hold for $r=1$.
    As for condition (2) when $r=1$, we have $t_j=(r-1)/n=0$, which is initially set in~\Cref{alg:pcr}, and $N_o=E_{P,1}(o)=E_{Q,1}(o)$ by Line~4.
    Then $\sum_{j\in N_o}p_{j,o}=s(o)=1$ because otherwise, $o$ has remaining shares for the later round while agents in $E_{Q,1}(o)$ are not satisfied, a violation to $Q$ satisfying \sdcrf{}.
    By construction of $\omega$, $\sum_{j\in E_{Q,1}(o)}\int_0^{1/n}\omega_j(t){\rm d}t = \sum_{j\in E_{Q,1}(o)}q_{j,o} = 1$, and we can infer that the consumption time for $o$ $\rho_o=1/n$.
    
    For $r=1$, conditions (1) and (2) easily lead to condition (3).
    Moreover, with conditions (1) and (2) established for $r$, we can obtain them for $r+1$ following a similar analysis, and therefore condition (3) holds for $r+1$.
    In this way, we have $p_{j,o}=q_{j,o}$ for any $o\in M_{Q,r}$ with any $r\ge1$.
    It implies that $Q=P$, which completes the proof.
\end{sketch}

With~\Cref{prop:sdcrf} and~\Cref{thm:familychar}, we can conclude that every member of \pcr{} also satisfies \sdopt{} and \epopt{} (\Cref{cor:pcrp}).

\begin{corollary}\label{cor:pcrp}
For any collection of eating speed functions $\omega$, $\pcr{}_\omega$ satisfies \fepopt{} (\epopt{}) and \fsdopt{} (\sdopt{}).
\end{corollary}

\section{Impossibility Results}

\begin{figure}[b]
    \centering
    
    \subfigure[\efcr{} mechanisms]{
    \includegraphics[width=0.45\linewidth]{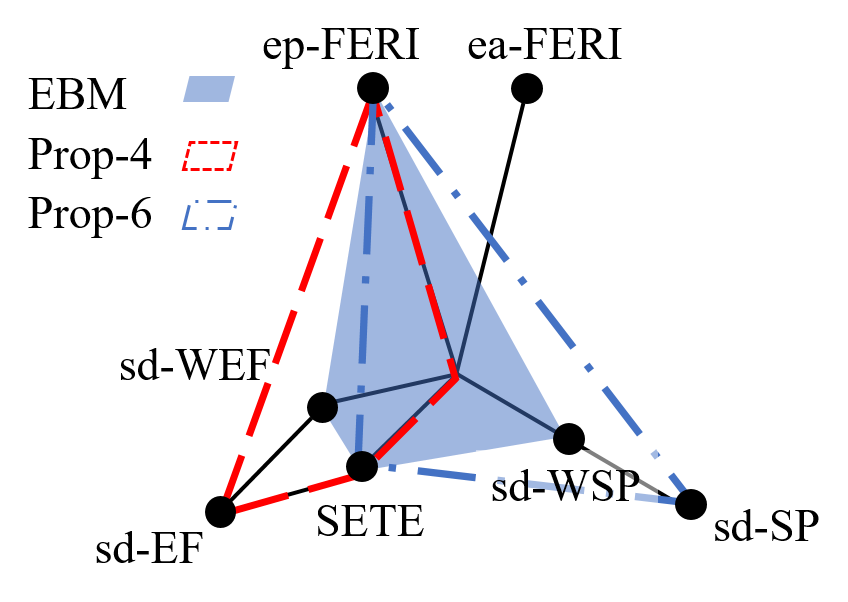}
    }
    \subfigure[\sdcrf{} mechanisms]{
    \includegraphics[width=0.45\linewidth]{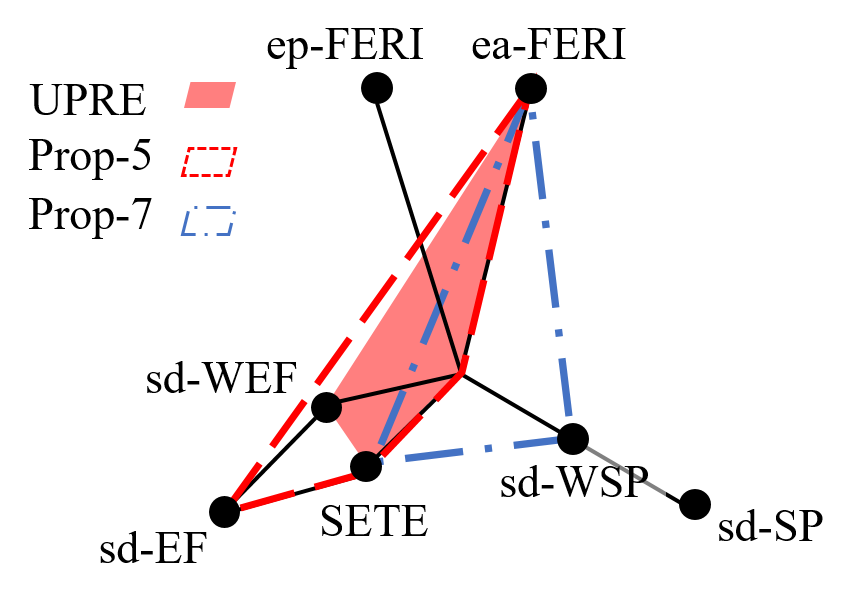}
    }
    \caption{The impossibility results for \fhcr{}  mechanisms.}
    \label{fig:imp}
\end{figure}

A natural question that follows the positive results in Sections~\ref{sec:ep:ebm} and~\ref{sec:ea:pre}, is whether it is possible to design \efcr{} or \sdcrf{} mechanisms which provide a stronger guarantee of either efficiency, fairness, or strategyproofness. In this section, we show that it is impossible to design mechanisms that satisfy some combinations of properties involving stronger variants of either fairness or strategyproofness. We summarize these results in the form of a spider graph in \Cref{fig:imp}, where the nodes represent properties, and a broken line joining the nodes represents a combination of properties that are impossible to satisfy by any assignment mechanism. The properties satisfied by our mechanisms are represented by the nodes at the corners of the shaded polygons. Each of the arms radiating away from the center of the graph represent a type of efficiency, fairness, or strategproofness property, with nodes farther away from the center of the graph representing stronger properties. For example, \sdef{} implies both the \sdwef{} and \etep{} fairness properties.
Our impossibility results for the design of \sdcrf{} mechanisms with stronger  guarantees requires a technical lemma as an auxiliary tool, so we provide the detailed proofs in~\Cref{sec:app:proof}.

Propositions~\ref{prop:impefcr1} and~\ref{prop:impsdcfr1} show that we cannot improve the fairness guarantee of \efcr{} or \sdcrf{} mechanisms from \sdwef{} to \sdef{}.

\begin{restatable}{prop}{propimpefcr1}{}\label{prop:impefcr1}
No mechanism simultaneously satisfies \fefcr{} (\efcr{}) and \fsdef{} (\sdef{}).
\end{restatable}
\begin{proof}
    We prove it with the following preference profile $R$.
	By \fhcr{}, one of agents in $\{1,2,3\}$ gets $a$, and agent $4$ must get $b$.
	If assignment $Q$ satisfies \efcr{} and \etep{} which is implied by \sdef{}, then it is in the following form.	
	
	\vspace{1em}\noindent
	\begin{minipage}{\linewidth}
		\centering
		\begin{minipage}{0.4\linewidth}
			\begin{center}
			    Preference Profile $R$\\
				$\succ_1$: $a\succ_1 c\succ_1 b\succ_1 d$,\\
				$\succ_2$: $a\succ_2 c\succ_2 b\succ_2 d$,\\
				$\succ_3$: $a\succ_3 b\succ_3 c\succ_3 d$,\\
				$\succ_4$: $b\succ_4 a\succ_4 d\succ_4 c$.\\
			\end{center}
		\end{minipage}
		\begin{minipage}{0.4\linewidth}
			\centering
			\begin{center}
				\centering
				\begin{tabular}{c|cccc}
					\multicolumn{5}{c}{Assignment $Q$}\\
					&  a & b & c & d\\\hline
					1 & $\frac{1}{3}$ & $0$ & ? & ?\\
					2 & $\frac{1}{3}$ & $0$ & ? & ?\\
					3 & $\frac{1}{3}$ & $0$ & ? & ?\\
					4 & $0$ & $1$ & $0$ & $0$\\
				\end{tabular}
			\end{center}
		\end{minipage}
	\end{minipage}\vspace{1em}

	Then we do not have $Q_3\sd{3}Q_4$ since $\sum_{o'\in\ucs(\succ_3,b)}q_{3,o'}<\sum_{o'\in\ucs(\succ_3,b)}q_{4,o'}$, a contradiction to \sdef{}.
\end{proof}

\begin{restatable}{prop}{propimpsdcfri}{}\label{prop:impsdcfr1}
No mechanism simultaneously satisfies \fsdcrf{} (\sdcrf{}) and \fsdef{} (\sdef{}).
\end{restatable}


As for strategyproofness, \Cref{prop:impefcr2} shows with the weak fairness requirement of \etep{}, that \sdsp{} cannot be improved to \sdssp{} for any \efcr{} mechanism, and \Cref{prop:impsdcfr2} shows that even \sdsp{} cannot be satisfied by \sdcrf{} mechanisms.

\begin{restatable}{prop}{propimpefcr2}{}\label{prop:impefcr2}
No mechanism simultaneously satisfies \fefcr{} (\efcr{}), \fetep{} (\etep{}), and \fsdssp{} (\sdssp{}).
\end{restatable}
\begin{proof}
  We prove it with the preference profile $R$ in \Cref{prop:impefcr1}.
	If agent $3$ misreports her preference as agent $4$, i.e., $R'=(\succ'_3,\succ_{-3})$ with $\succ'_3=\succ_4$, then one of agents $1$ and $2$ gets $a$, and one of agents $3$ and $4$ get $b$ by \fhcr{}.
	For the remaining items $M'=\{c,d\}$, $\tp{1}{M'}=\tp{2}{M'}=c$ and $\tp{\succ'_3}{M'}=\tp{4}{M'}=d$, which means that agent $1$ (or $2$) gets $c$ when she does not get $a$, and agent $3$ (or $4$) gets $d$ when she does not get $b$.
	Then the following assignment $Q'$ is the only one satisfying \efcr{} and \etep{} for $R'$.

	\vspace{1em}\noindent
	\begin{minipage}{\linewidth}
		\centering
		\begin{minipage}{0.4\linewidth}
			\begin{center}
			    Preference Profile $R'$\\
				$\succ_1$: $a\succ_1 c\succ_1 b\succ_1 d$,\\
				$\succ_2$: $a\succ_2 c\succ_2 b\succ_2 d$,\\
				$\succ'_3$: $b\succ_3 a\succ_3 d\succ_3 c$,\\
				$\succ_4$: $b\succ_4 a\succ_4 d\succ_4 c$.\\
			\end{center}
		\end{minipage}
		\begin{minipage}{0.4\linewidth}
			\centering
				\begin{center}
		\centering
		\begin{tabular}{c|cccc}
			\multicolumn{5}{c}{Assignment $Q'$}\\
			&  a & b & c & d\\\hline
			1 & $\frac{1}{2}$ & $0$ & $\frac{1}{2}$ & $0$\\
			2 & $\frac{1}{2}$ & $0$ & $\frac{1}{2}$ & $0$\\
			3 & $0$ & $\frac{1}{2}$ & $0$ & $\frac{1}{2}$\\
			4 & $0$ & $\frac{1}{2}$ & $0$ & $\frac{1}{2}$\\
		\end{tabular}
	\end{center}
		\end{minipage}
	\end{minipage}\vspace{1em}

	Comparing $Q'$ with $Q$ in \Cref{prop:impefcr1} which is under the true preference $R$, we see that $Q_3$ does not dominate $Q'_3$ on $\ucs(\succ_3,b)$, which means that no \efcr{} and \etep{} mechanism can be \sdssp{}.
\end{proof}

\begin{restatable}{prop}{propimpsdcfrii}{}\label{prop:impsdcfr2}
No mechanism simultaneously satisfies \fsdcrf{} (\sdcrf{}), \fetep{} (\etep{}), and \fsdsp{} (\sdsp{}).
\end{restatable}

\Cref{prop:impefcrsdcfr} shows that the simultaneous satisfaction of \efcr{} and \sdcrf{} cannot be achieved given the fairness requirement \etep{}.

\begin{restatable}{prop}{propimpefcrsdcfr}{}\label{prop:impefcrsdcfr}
No mechanism simultaneously satisfies \fefcr{} (\efcr{}), \fsdcrf{} (\sdcrf{}), and \fetep{} (\etep{}).
\end{restatable}

\section{Conclusion and Future Work}
In this paper, we provide the first random mechanisms that are guaranteed to output ex-post \fcm{} assignments simultaneously with other desirable properties of efficiency (\epopt{} and \sdopt{}), fairness (\etep{} and \sdwef{}), and strategyproofness (\sdsp{}) properties.
As we summarize in \Cref{fig:cmp}, our positive results expand the envelope for \fcm{} and \opt{} mechanisms along the dimensions of both fairness and strategyproofness, while our impossibility results, summarized in~\Cref{fig:imp}, help define the limits of what may be possible.

\begin{figure}[htb]
    \centering
    \includegraphics[width=0.5\linewidth]{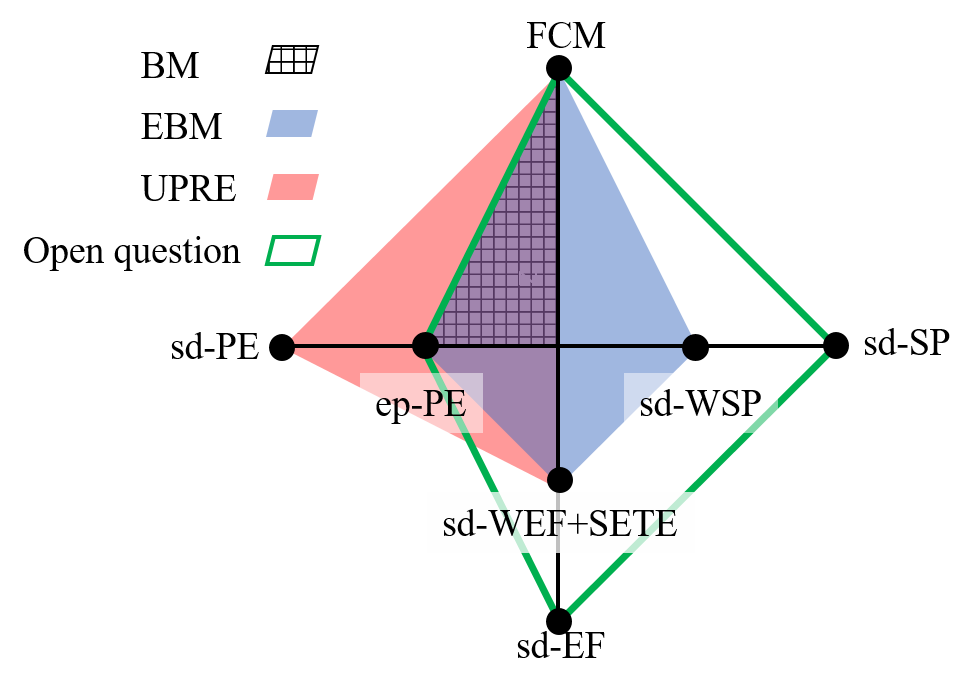}
    \caption{The properties of \fcm{}  mechanisms on \opt{}, envy-freeness, and strategyproofness.}
    \label{fig:cmp}
\end{figure}

We hope that our results encourage the search for \fcm{} and \opt{} mechanisms which also satisfy other efficiency, fairness, and strategyproofness desiderata (like the ``Open question'' in~\Cref{fig:cmp}).
For example, if \fcm{} is deemed to be an indispensable property, an interesting open question is whether it is possible to design mechanisms that satisfy relaxations of one property, such as Pareto optimality among \fcm{} assignments, while satisfying stronger properties of fairness and strategyproofness, along lines of similar work on constrained-optimal mechanisms in the school choice literature~\cite{abdulkadirouglu2003school,abdulkadirouglu2017minimizing}.


Another natural direction for future work concerns generalizations of the assignment problem. 
An immediate question is whether our results may be extended to settings where ties or incomparability between items are allowed in agents' preferences~\cite{Katta06:Solution,wang2020multi}.
When agents may demand multiple items~\cite{Heo14:Probabilistic,kojima2009:Random,wang2020multi,budish2011combinatorial}, the question of whether a mechanism can satisfy a natural extension of \fcm{} together with other desiderata also remains open.

\begin{acks}
LX acknowledges NSF \#1453542 and \#1716333 for support. YC acknowledges NSFC under Grants 62172016  and 61932001 for support. HW acknowledges NSFC under Grant 61972005 for support.
\end{acks}


%

\appendix

\section{Additional Results}

\label{sec:app:results}

\subsection{Relationships between Efficiency Properties}\label{sec:app:results:relation}

We summarize the efficiency properties we discuss in the main body with~\Cref{fig:nwaxioms}.

\begin{figure}[htb]
	\centering
	\includegraphics[width=0.5\linewidth]{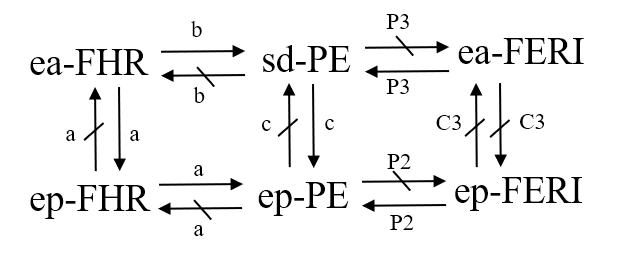}
	\caption{Relationship between \efcr{}, \sdcrf{}{}, \efr{}, \sdrf{}, \epopt{}, and \sdopt{}.}
	\begin{flushleft}
	  {\small Note: The property A points to another property B means that A implies B.
	An arrow annotated with $\texttt{a}$, $\texttt{b}$, or $\texttt{c}$ refers to a result due to~\protect\citeA{Ramezanian2021:Ex-post},~\protect\citeA{Chen2021:Theprobabilistic}, and~\protect\citeA{Bogomolnaia01:New}, respectively, and an edge annotated with P refers to a Proposition in this paper.}
	\end{flushleft}
	\label{fig:nwaxioms}
\end{figure}

\begin{restatable}{prop}{propdeterfhcrsdcrf}{}\label{prop:deterfhcrsdcrf}
	A deterministic assignment satisfies \ffhcr{} (\fhcr{}) if and if only it satisfies \fsdcrf{} (\sdcrf{}).
\end{restatable}


\begin{proof}
    The sets $\tps{A}{r}$, $M_{A,r}$, and $E_{A,r}(o)$ in the following have already been defined in \Cref{dfn:fhcr} and \Cref{dfn:sdcrf}.

	\vspace{0.5em}\noindent{\bf (\fhcr{}$\Rightarrow{}$\sdcrf{})}
	For any $A$ satisfying \fhcr{}, we show that for any $r$, the following two conditions hold:
	\begin{enumerate}[label={\bf Condition (\arabic*):},align=left,wide,labelindent=0pt]
		\item $M_{A,r}=M\setminus\bigcup_{r'<r}\tps{A}{r'}$, and
		\item for any item $o\in M_{A,r^*}$ with $r^*>r$, $\sum_{\hat{o}\in\ucs(\succ_j,o)}A_{j,\hat{o}}=1$ for any $j\in E_{A,r}(o)$.
	\end{enumerate}
	
	\paragraph{Base Case.}
	Condition (1) trivially holds for $r=1$.
	We show condition (2) holds by contradiction.
	Assume that there exist an item $o\in M_{A,r^*}$ with $r^*>1$ and an agent $j\in E_{A,1}(o)$ such that $\sum_{\hat{o}\in\ucs(\succ_j,o)}A_{j,\hat{o}}=0$.
	Since $j\in E_{A,1}(o)$, we know that $o=\tp{j}{M_{A,1}}=\tp{j}{M}$, which also means that $o\in\tps{A}{1}$.
	By $o\in M_{A,r^*}$, we know that $\sum_{j\in\bigcup_{r'\le r}E_{A,r'}(o)}A_{j,o}=0<1$, which also means that no agent $j'\in E_{A,1}(o)=\{k\in N\mid o=\tp{k}{M}\}$ gets $o$, a contradiction to $A$ satisfying \fhcr{}.

	\paragraph{Induction Step.}
	Supposing that they hold for each of $r'=1,\dots,r-1$, we prove that they also hold for $r$.
	
	\vspace{0.5em}\noindent{\bf Condition (1) holds for $r$.} For any $r'<r$, by \fhcr{}, $o'=\tp{A^{-1}(o')}{M_{A,r'}}$ for every item $o'\in \tps{A}{r'}$, which means that $A^{-1}(o')\in E_{A,r'}(o')$ and $\sum_{j\in\bigcup_{r'<r}E_{A,r'}(o')}A_{j,o'}=1$.
	Therefore $o'\notin M_{A,r}$.
	It is easy to prove the opposite direction that $o\in \tps{A}{r}$ for any $o\notin M_{A,r}$ with the similar analysis.
	Together they mean lead to condition (1) for $r$.
	
	\vspace{0.5em}\noindent{\bf Condition (2) holds for $r$.}
	For any item $o\in M_{A,r^*}$ with $r^*>r$, we have that $\sum_{j\in\bigcup_{r'\le r}E_{A,r'}(o)}A_{j,o}<1$.
	Since $A$ is deterministic, it means that $A_{j',o}=0$ for any $j'\in E_{A,r'}(o)$ with $r'\le r$.
	Besides, $o\in M_{A,r^*}\subseteq M_{A,r}$, and then by \fhcr{} and (1) for $r$, we obtain that $o\notin\tps{A}{r}$, which means that for any $j\in E_{A,r}(o)$, i.e. $o=\tp{j}{M_{A,r}}$, we have that $A(j)\in\tps{A}{r'}$ with some $r'<r$.
	Again by (1) for $r'$, we know that $A(j)=\tp{j}{M_{A,r'}}\succ_j o$ since $M_{A,r}\subseteq M_{A,r'}$, and therefore $\sum_{o'\in\ucs(\succ_j,o)}A_{j,o'}=1$, i.e. condition (2) for $r$.
	
    With that condition (2) holds for any $r$, we obtain the \fhcr{} assignment $A$ satisfies \sdcrf{}.
	
	\vspace{0.5em}\noindent{\bf (\sdcrf{}$\Rightarrow{}$\fhcr{})}
	For any $A$ satisfying \sdcrf{}, we show that for any $r$, the following two conditions hold:
	\begin{enumerate}[label=(\arabic*),leftmargin=*,labelindent=0pt,topsep=0.5em,itemsep=0pt]
		\item $M\setminus\bigcup_{r'<r}\tps{A}{r'}=M_{A,r}$, and
		\item $o=\tp{A^{-1}(o)}{M_{A,r}}$ for every item $o\in \tps{A}{r}$.
	\end{enumerate}
	
	\paragraph{Base Case.}
	Condition (1) trivially holds for $r=1$.
	We show condition (2) holds by contradiction.
	If $\tp{A^{-1}(o)}{M_{A,1}}=\tp{A^{-1}(o)}{M}\neq o$, then $A^{-1}(o)\notin E_{A,1}(o)$, and $\sum_{k\in E_{P,1}(o)}p_{k,o}<1$.
	It follows that $o\in M_{A,2}$, while for any $j\in E_{P,1}(o)\neq\emptyset$, $\sum_{o'\in(\succ_j,o)A_{j,o'}}=A_{j,o}=0<1$, a contradiction to $A$ satisfying \sdcrf{}.
	
	\paragraph{Induction Step.} Supposing that both condition hold for each of $r'=1,\dots,r-1$, we prove that they also hold for $r$.
	
	\vspace{0.5em}\noindent{\bf Condition (1) holds for $r$.} For every $o\in M_{A,r}$, we have that $\sum_{j\in\bigcup_{r'< r}E_{A,r}(o)}A_{j,o}<1$, which implies that $A_{j,o}=0$ for any $j\in\bigcup_{r'<r}E_{A,r'}(o)$ since $A$ is deterministic.
	By the assumption that $A$ is \sdcrf{}, $\sum_{\hat{o}\in\ucs(\succ_j,o)}p_{j,\hat{o}}=1$ for any $j\in\bigcup_{r'<r}E_{A,r'}(o)$. Now,
	\begin{itemize}[leftmargin=*,topsep=0pt,itemsep=0pt]
	    \item If $E_{A,r'}(o)=\emptyset$ for any $r'<r$, then there does not exist $j$ with $\tp{j}{M_{A,r'}}=o$ and therefore $o\notin\tps{A}{r'}$.
	    \item If $E_{A,r'}(o)\neq\emptyset$ for some $r'<r$, then we have that $o\notin\tps{A}{r'}$, since condition (2) holds for $r'$ and $A_{j,o}=0$ for any $j\in E_{A,r'}(o)$.
	\end{itemize}
	Then, we have that $o\notin\bigcup_{r'<r}\tps{A}{r'}$.
	It is easy to prove the opposite direction that $o\in M_{A,r}$ for any $o\notin\bigcup_{r'<r}\tps{A}{r'}$ with a similar argument.
	Together they mean that $M_{A,r}= M\setminus\bigcup_{r'<r}\tps{A}{r'}$, i.e., condition (1) holds for $r$.

	\vspace{0.5em}\noindent{\bf Condition (2) holds for $r$.}
	For any item $o\in \tps{A}{r}$, assume for the sake of contradiction that $o\neq\tp{A^{-1}(o)}{M_{A,r}}$, which means that, $A(j)\neq o$ for any $j$ with $o=\tp{j}{M_{A,r}}$, i.e., $j\in E_{A,r}(o)$.
	
	Since $o\in \tps{A}{r}$, then $o\in M_{A,r}$ by condition (1) for $r$ which we have just proved, and therefore $\sum_{j\in\bigcup_{r'< r}E_{A,r'}(o)}A_{j,o}=0<1$ since $A$ is deterministic.
	Together with the assumption, we have that $\sum_{j\in\bigcup_{r'< r+1}E_{A,r'}(o)}A_{j,o}=0<1$, which means that $o\in M_{A,r+1}$.
	Again by $o\in \tps{A}{r}$, we know that there exists $j'\in E_{A,r}(o)$, i.e., $o=\tp{j'}{M_{A,r}}$ with $A(j')\notin\bigcup_{r'<r}\tps{A}{r'}$.
	It also means that $A(j')\in M_{A,r}$ by condition (1) for $r$, and therefore $o\succ_{j'}A(j')$.
	Then $\sum_{\hat{o}\in\ucs(\succ_{j'},o)}p_{j',\hat{o}}=0$ with $o\in M_{A,r+1}$ and $j'\in E_{A,r}(o)$, a contradiction to the assumption that $A$ satisfies \sdcrf{}.
	By the contradiction, we have that $o=\tp{A^{-1}(o)}{M_{A,r}}$ for every item $o\in \tps{A}{r}$, i.e., condition (2) holds for $r$.
	
	Together we have condition (2) established for any $r$, which means that the \sdcrf{} assignment $A$ satisfies \fhcr{}.
\end{proof}

\begin{corollary}\label{cor:frenfea}
     {\em[\efcr{}$\not\Rightarrow{}$\sdcrf{}, \sdcrf{}$\not\Rightarrow{}$\efcr{}]}
     A random assignment satisfying \fsdcrf{} (\sdcrf{}) do not need to satisfy \fefcr{} (\efcr{}), and vice versa.
\end{corollary}
\begin{proof}
    {\bf(\efcr{}$\not\Rightarrow{}$\sdcrf{})} It follows from the fact that \am{} satisfies \efcr{} (\Cref{thm:amp}), but \am{} does not satisfy \sdcrf{} (\Cref{prop:ebm}).

	\vspace{0.5em}\noindent{\bf (\sdcrf{}$\not\Rightarrow{}$\efcr{})} It follows from the fact that \upre{} satisfies \sdcrf{} (\Cref{thm:familychar}), but \upre{} does not satisfy \efcr{} (\Cref{prop:upre}).
\end{proof}

\begin{definition}
    \cite{abraham2007popular,belahcene2021combining}
    Given an instance, a deterministic assignment $A$ satisfies \fpop{} (\pop{}) if there does not exist another $A'$ such that $|\{j\in N:A'(j)\succ_j A(j)\}|>|\{j\in N:A(j)\succ_j A'(j)\}|$.
\end{definition}

\begin{restatable}{lemma}{lempop}{}\label{lem:pop}
     {\rm\cite{abraham2007popular}} An assignment $A$ is popular if and only if every agent is assigned {\rm(1)} her most preferred item, or {\rm(2)} the most preferred item which is not ranked first by any agent.
\end{restatable}

\begin{restatable}{prop}{proppop}{}\label{prop:pop}
     A \fpopa{}~(\pop{}) deterministic assignment satisfies \ffhcr{} (\fhcr{}), but not vice versa.
\end{restatable}

\begin{proof}
{\bf ($\pop{}\Rightarrow\fhcr{}$)}
By~\Cref{lem:pop}, an assignment $A$ is popular if and only if for every agent $j$, $A(j)=\tp{j}{M}$ or $A(j)=\tp{j}{M'}$ where $M'=\{o\in M|o=\tp{k}{M} \text{ for any }k\in N\}$.

By~\Cref{dfn:fhcr}, we observe that:
\begin{enumerate}[label=\rm(\roman*),wide,labelindent=0pt,topsep=0em,itemsep=0pt]
    \item For $r=1$, $\tps{A}{1} = \{\tp{j}{M}|j\in N\}$ for any $j\in N$. By the characterization above, any $o\in\tps{A}{1}$ is assigned to one of the agents who rank it as their top item in $M$, i.e. $o=\tp{A^{-1}(o),M}$, which satisfies \fhcr{} for $r=1$.
    \item For $r=2$, $\tps{A}{2}=\{o\in M:o=\tp{j}{\allowbreak M\setminus\tps{A}{1}},\allowbreak A(j)\notin\tps{A}{1}\}$.
    We note that $M'=M\setminus\tps{A}{1}$ because by~(\romannumeral1), any $o$ satisfies $o\in\tps{A}{1}$ if and only if $o=\tp{k}{M}$ for some agent $k$.
    According to the characterization of \pop{}, for any $j$ with $A(j)\notin\tps{A}{1}$, we trivially have $A(j)\neq=\tp{j}{M}$ and therefore $A(j)=\tp{j}{M'}$.
    It also means that for any $o\in\tps{A}{2}$, $o=\tp{A^{-1}(o)}{M'}$, which satisfies \fhcr{} for $r=2$.
    \item for $r>2$, $\tps{A}{r}=\emptyset$ because $A(j)=\tp{j}{M}\in\tps{A}{1}$ or $A(j)=\tp{j}{M'}\in\tps{A}{2}$, and we have that $A$ satisfies \fhcr{} trivially.
\end{enumerate}

\vspace{0.5em}\noindent{\bf ($\fhcr{}\not\Rightarrow\pop{}$)}
Consider the profile in~\Cref{fig:fhr} and the assignment $A^*$ in~\Cref{eg:fcr} which satisfies \fhcr{}.
Now, consider the assignment $A'$ below:
\begin{equation*}
	\begin{split}
		A':~&1\gets a,~2\gets b,~3\gets f,~4\gets c,~5\gets e,~6\gets d.\\
	\end{split}
\end{equation*}
Notice that $A'$ is more popular than $A^*$: $A'(4)\succ_4 A^*(4)$, $A'(5)\succ_4 A^*(5)$, while $A^*(3)\succ_4 A'(3)$, and $A'(j)=A^*(j)$ for $j\in\{1,2,6\}$. Therefore, \fhcr{} does not imply \pop{}.
\end{proof}

\subsection{\rkm{} is Not Compatible with Fairness and Strategyproofness}\label{sec:app:results:rkm}

We also consider the the compatibility of \frkm{} (\rkm{}) with fairness and strategyproofness. 

\begin{definition}
    \cite{irving2006rank}
    Given an instance, a deterministic assignment $A$ satisfies \frkm{}  (\rkm{}) if there is no assignment $A'$ such that its {\em signature} dominates $A$, where the signature of $A$ is an $n$-vector $\vx=(x_r)_{r\le n}$ such that for each $r\in[n]$, the $r$-th component is the number of agents who are allocated their $r$-th ranked item, and a signature $\vx$ dominates signature $\vy$ if there exists $r'$ such that $x_{r'}>y_{r'}$ and for every $r''<r'$, $x_{r''}\ge y_{r''}$.
\end{definition}

Since \rkm{} implies \fhr{}~\cite{belahcene2021combining}, it automatically means that ex-post \rkm{} (\erkm{}) suffers the same incompatibility with fairness and strategyproofness as \fhr{}. In fact, \erkm{} is incompatible even with \sdwef{} alone as we show in~\Cref{prop:imprkm}.

\begin{restatable}{prop}{propimprkm}{}\label{prop:imprkm}
	No mechanism satisfies \ferkm{} (\erkm{}) and \fsdwef{} (\sdwef{}) simultaneously.
\end{restatable}
\begin{proof}
	Let $R$ be the profile in~\Cref{fig:fhr}.
	By \rkm{} implying \fcm{}, $a$ and $b$ must be assigned to agents $1$ and $2$, respectively.
	Although agents $3$-$6$ all rank $c$ on top, in any \rkm{} assignment, item $c$ can only be allocated to agent $6$ and $\{d,e,f\}$ to agents $3$-$5$, which leads to a signature $\vy=(3,1,1,1,0,0)$.
	Otherwise, if $c$ is not assigned to agent $6$ in an \rkm{} assignment, then by \rkm{} implying \fhr{}, agent $6$ can only get $f$ since $\rk{6}{d}>\rk{j}{d}$ and $\rk{6}{e}>\rk{j}{e}$ for $j\in\{3,4,5\}$, and agents $3$-$5$ get $\{c,d,e\}$, which results in the signature $\vx=(3,1,1,0,0,1)$ dominated by $\vy$, a contradiction. Then, since any random assignment $P$ satisfying \erkm{} is a convex combination over the set of possible \rkm{} assignments, we have that for any item $o$, $\sum_{o'\in\ucs(\succ_3,o)}p_{3,o'}\le1=p_{6,c}=\sum_{o'\in\ucs(\succ_6,o)}p_{6,o'}$, and it is strict for items other than $f$, a violation of \sdwef{}.
\end{proof}

\subsection{Properties that RP, PS, BM, \abm{}, \am{}, PR, and \upre{} Fail to Satisfy}\label{sec:app:results:properties}

\Cref{prop:rp,prop:ps} show that RP and PS are not \ffcma{}, and therefore they do not satisfy \efcr{} since \fhcr{} implies \ffcm{}.

\begin{restatable}{prop}{proprp}{}\label{prop:rp}
	 RP does not satisfy \ffcm{} (\fcm{}) ex-post.
\end{restatable}
\begin{proof}
    We show it with the instance with following profile:
	\begin{center}
		$\succ_1:~a,b,c$;\quad$\succ_2:~a,c,b$;\quad$\succ_3:~b,a,c$.
	\end{center}
	The following $A$ is one possible output of RP when the priority order is $2\impord{}1\impord{}3$.
	In $A$, we see that only agent $2$ gets her first choice.
	However, there exists another assignment $A'$ as shown below where both agents $1$ and $3$ get their first choices, which means that RP does not satisfy \fcm{}.
			
	\vspace{1em}\noindent
	\begin{minipage}{\linewidth}
		\centering
		\begin{minipage}{0.4\linewidth}
			\begin{center}
				\centering
				\begin{tabular}{c|ccc}
					\multicolumn{4}{c}{Assignment $A$}\\
					&  a & b & c\\\hline
					$1$ & $0$ & $1$ & $0$\\
					$2$ & $1$ & $0$ & $0$\\
					$3$ & $0$ & $0$ & $1$\\
				\end{tabular}
			\end{center}
		\end{minipage}
		\begin{minipage}{0.4\linewidth}
			\centering
			\begin{center}
				\centering
				\begin{tabular}{c|ccc}
					\multicolumn{4}{c}{Assignment $A'$}\\
					&  a & b & c\\\hline
					$1$ & $1$ & $0$ & $0$\\
					$2$ & $0$ & $0$ & $1$\\
					$3$ & $0$ & $1$ & $0$\\
				\end{tabular}
			\end{center}
		\end{minipage}
	\end{minipage}\vspace{1em}
\end{proof}

\begin{restatable}{prop}{propps}{}\label{prop:ps}
	 PS does not satisfy \ffcm{} (\fcm{}) ex-post.
\end{restatable}
\begin{proof}
    {\bf(neither \efcr{} nor \fcm{})}
    we continue to use the instance with the profile in~\Cref{prop:rp}.
	The following $P$ is the outcome of PS.
	We see that $A$ in~\Cref{prop:rp} must be among the deterministic assignments which constitute the convex combination for $P$, which means that PS does not satisfy \fcm{}, and therefore it is not \efcr{}.
	\begin{center}
				\centering
				\begin{tabular}{c|ccc}
					\multicolumn{4}{c}{Assignment $P$}\\
					&  a & b & c\\\hline
					$1$ & $1/2$ & $1/4$ & $1/4$\\
					$2$ & $1/2$ & $0$ & $1/2$\\
					$3$ & $0$ & $3/4$ & $1/4$\\
				\end{tabular}
			\end{center}
\end{proof}

\begin{restatable}{prop}{propnbm}{}\label{prop:nbm}
BM with a uniform probability distribution over all the priority orders of agents does not satisfy \fefcr{} (\efcr{}), \fsdcrf{} (\sdcrf{}) or \fsdwef{} (\sdwef{}), but satisfies \fetep{} (\etep{}).
\end{restatable}
\begin{proof}
    We refer to BM with a uniform probability distribution as BM$^u$.

    \vspace{0.5em}\noindent{\bf(not \efcr{})} For the instance with profile in~\Cref{fig:fhr}, the assignment indicated by circled item is one possible outcome of BM given the priority $1\impord{}2\impord{}3\impord{}4\impord{}5\impord{}6$, which does not satisfy \fhcr{} as we discuss in~\Cref{eg:fcr}, and therefore it is not \efcr{}.

    \vspace{0.5em}\noindent{\bf (not \sdcrf{})} This follows from the fact that it is not \sdopta{}~\cite{Chen2021:Theprobabilistic}.

    \vspace{0.5em}\noindent{\bf (not \sdwef)} This follows from~\Cref{prop:impefr} and the fact that it satisfies \etep{} (shown below) and \efr{}~\cite{Ramezanian2021:Ex-post}.

    \vspace{0.5em}\noindent{\bf (\etep{})}
    Let $P=\mathbb E(\text{BM}^u(R))$ for any given profile $R$.
    For any agents $j,k$ and their common prefix $\succ_{j,k}$, given a priority order $\impord{}$ with $j\impord{} k$, if $j$ gets an item $o$ appearing in $\succ_{j,k}$, then it is easy to see that $k$ gets $o$ given $\impord{}'$ which just swaps the positions of $j$ and $k$ in $\impord{}$.
	Due to the assumption, we know that any such pair of priorities $\impord{}$ and $\impord{}'$ have the equal probability to be drawn, and therefore we have that $p_{j,o}=p_{k,o}$, which means \etep{}.
\end{proof}

We recall~\Cref{lem:BM} from \citeA{Bogomolnaia01:New} used in~\Cref{prop:ebm} and the proof of~\Cref{prop:sdcrf}.

\begin{lemma}\label{lem:BM}
	\cite{Bogomolnaia01:New}
	Given a preference profile $R$ and a random assignment $P$, let $\tau(P,R)$ be the relation over all the items such that: if there exists an agent $j$ such that $o_a\succ_j o_b$ and $p_{j,o_b}>0$, then $o_a\tau(P,R)o_b$.
	The  random assignment $P$ is \sdopta{} if and only if $\tau(P,R)$ is acyclic.
\end{lemma}

\begin{restatable}{prop}{propebm}{}\label{prop:ebm}
\am{} with a uniform probability distribution over all the priority orders of agents does not satisfy \fefr{} (\efr{}), \fsdopt{} (\sdopt{}), \fsdcrf{} (\sdcrf{}), \fsdrf (\sdrf{}), \fsdef{} (\sdef{}) or \fsdssp{} (\sdssp{}).
\end{restatable}
\begin{proof}
    {\bf(not \efr{})} For the profile in~\Cref{fig:fhr}, one of its possible outcome is the assignment $A^*$ in~\Cref{eg:fcr} which does not satisfy \fhr{}, and therefore \am{} is not \efr{}.

    \vspace{0.5em}\noindent{\bf (not \sdopt{})} We show it by the instance with following $R$:
	\begin{center}
		\begin{tabular}{l}
			$\succ_1$: $a,b,c,$others;\quad$\succ_2$: $a,b,d,$others;\quad$\succ_3$: $a,b,e,$others;\\
			$\succ_4$: $a,b,f,$others;\quad$\succ_{5\text{-}7}$: $a,b,g,c,d,x,y,$others;\\
			$\succ_{8\text{-}10}$: $a,b,h,e,f,y,x,$others.
		\end{tabular}
	\end{center}
	The following are two possible outcomes of \am{} where $j\gets o$ means that agent $j$ gets item $o$:
	\begin{equation*}
		\begin{split}
			A:&1\gets a,2\gets b,3\gets e,4\gets f,5\gets g,\\
			&6\gets c,7\gets d,8\gets h,9\gets y,10\gets x.\\
			A':&1\gets c,2\gets d,3\gets a,4\gets b,5\gets g,\\
			&6\gets x,7\gets y,8\gets h,9\gets e,10\gets f\\
		\end{split}
	\end{equation*}
	Let $P=\mathbb E(\text{AM}(R))$.
	Then $p_{7,y}>0$ and $p_{10,x}>0$.
	With $x\succ_7 y$ and $y\succ_{10} x$ and~\Cref{lem:BM}, we have that $x\tau(P,R)y$ and $y\tau(P,R)x$, which means that $P$ is not \sdopta{}.
	
	\vspace{0.5em}\noindent{\bf (neither \sdcrfa{} nor \sdrfa{})} It follows from the fact that \am{} is not \sdopta{}.
	
	\vspace{0.5em}\noindent{\bf (neither \sdefa{} nor \sdsspa{})} This follows from \Cref{prop:impefcr1,prop:impefcr2} and the fact that it satisfies \efcr{} and \etep{}.
\end{proof}

\begin{restatable}{prop}{propabm}{}\label{prop:abm}
\abm{} satisfies \etep{}, but does not satisfy \fefr{} (\efr{}), \fsdopt{} (\sdopt{}), \fsdcrf{} (\sdcrf{}), \fsdrf (\sdrf{}), or \fsdef{} (\sdef{}).
\end{restatable}

\begin{proof}

We refer to \abm{} with a uniform probability distribution as \abm{}$^u$.

\vspace{0.5em}{\bf\noindent(not \efhr{})} It follows from the fact that no mechanism satisfies \efhr{} with \etep{} and \sdsp{}~\cite{Ramezanian2021:Ex-post}.

\vspace{0.5em}\noindent{\bf (not \sdopt{}, \sdcrfa{}, or \sdrfa{})}
Consider the preference profile $R$ , and the two deterministic assignments $A$, $A'$ in~\Cref{prop:ebm} which are \fhcr{} since \am{} satisfies \efcr{}.
According to the proof of~\Cref{thm:abmchar}, we know that any \fhcr{} assignment is a possible output of \abm{} with certain priority.
Then $A$ and $A'$ are among the deterministic assignments that constitute convex combinations for $P=\mathbb E(\text{\am}^u(R))$ as the priority is chosen uniformly.
From the proof of~\Cref{prop:ebm}, we know that $P$ is not \sdopt{}, therefore not \sdcrf{} and \sdrf{}.

\vspace{0.5em}{\bf\noindent(not \sdef{})} It follows from~\Cref{prop:impefcr1} and the fact that ABM satisfies \efcr{}.

\vspace{0.5em}{\bf\noindent(\etep{})} Let $R$ be any preference profile. 
For any two agents $j$ and $k$ and a given priority $\impord{}$ with $j\impord{}k$, let $\impord{}'$ be another priority which only swaps $j$ and $k$ in $\impord{}$, $A=\text{ABM}^{\impord{}}(R)$, and  $A'=\text{ABM}^{\impord{}'}(R)$.
If $A(j)$ appears in $\succ_{j,k}$, it is easy to see that $A'(k)=A(j)$.
Furthermore, if $A(k)$ also appears in $\succ_{j,k}$, then $A'(j)=A(k)$.
Let $P=\mathbb E(\text{\am}^u(R))$.
Since the pair of priorities like $\impord{}$ and $\impord{}'$ have the same probability to be chosen, we see that for any item $o$ appearing in $\succ_{j,k}$, $p_{j,o}=p_{k,o}$, which means that $P$ satisfies \etep{}.
\end{proof}

\begin{restatable}{prop}{proppr}{}\label{prop:pr}
 PR does not satisfy \fefcr{} (\efcr{}) or \fsdcrf{} (\sdcrf{}), but satisfies \etep{}.
\end{restatable}
\begin{proof}
    {\bf(not \efcr{})}
    We show it by the instance with $R$ in~\Cref{fig:fhr}.
    Let $P=\text{PR}(R)$ is shown in the following.
	\begin{center}
		\centering
		\begin{tabular}{r|cccccc}
			\multicolumn{7}{c}{Assignment $P$}\\
			&  a & b & c & d & e & f\\\hline
			$1$ & $1$ & $0$ & $0$ & $0$ & $0$ & $0$\\
			$2$ & $0$ & $1$ & $0$ & $0$ & $0$ & $0$\\
			$3$-$5$ & $0$ & $0$ & $1/4$ & $1/3$ & $1/3$ & $1/12$\\
			$6$ & $0$ & $0$ & $1/4$ & $0$ & $0$ & $3/4$\\
		\end{tabular}
	\end{center}
	Let  $A$ be the deterministic assignment indicated by circled items in~\Cref{fig:fhr}, and we see that $A$ must be among the deterministic assignments which constitute the convex combination for $P$.
	We know that $A$ does not satisfy \fhcr{} as shown in~\Cref{eg:fcr}, which means that $P$ does not satisfy \efcr{}.
	
	\vspace{0.5em}\noindent{\bf (not \sdcrf{})}
	We continue to use the instance with $R$ in~\Cref{fig:fhr}. In the assignment $P$ above, it is easy to obtain that $M_{P,2}=\{d,e,f\}$ since agents in $E_{P,1}(o)$ for $o\in\{a,b,c\}$ owns all the shares of $o$. ($M_{P,r}$ and $E_{P,r}(o)$ are defined in~\Cref{dfn:sdcrf}).
	Then we have that $E_{P,2}(d)=\{6\}$ and $\sum_{j\in\bigcup_{r<3}E_{P,d}}p_{j,d}=p_{6,d}=0<1$, which means that $d\in M_{P,3}$ while $\sum_{o\in\ucs(6,d)}p_{6,o}=1/4<1$, which violates \sdcrf{}.
	
	\vspace{0.5em}\noindent{\bf (\etep{})}
    PR satisfies equal-rank envy-freeness by~\cite{Chen2021:Theprobabilistic} which requires that in $P=\text{PR}(R)$ for the given preference profile $R$, for any agents $j,k$ and item $o$ with $\rk{j}{o}=\rk{k}{o}$, $\sum_{o'\succ_j,o}p_{j,o'}+p_{k,o}\le\sum_{o'\in\ucs(\succ_j,o)}p_{j,o'}$.
    Then for any item $o$ appearing in $\succ_{j,k}$, $\rk{j}{o}=\rk{k}{o}$, and therefore $\sum_{o'\succ_j,o}p_{j,o'}+p_{k,o}=\sum_{o'\in\ucs(\succ_j,o)}p_{j,o'}$, i.e., $p_{j,o}=p_{k,o}$, which means that PR satisfies \etep{}. 	
\end{proof}

\begin{restatable}{prop}{propupre}{}\label{prop:upre}
\upre{} does not satisfy \fefcr{} (\efcr{}), \fefr{} (\efr{}), \fsdrf{} (\sdrf{}), \fsdef{} (\sdef{}), or \fsdssp{} (\sdsp{}).
\end{restatable}
\begin{proof}
    {\bf(not \efcr{})} This follows from~\Cref{prop:impefcrsdcfr} and the fact that it satisfies \sdcrf{} by~\Cref{thm:familychar} and \etep{} by~\Cref{thm:uprep}.

    \vspace{0.5em}\noindent{\bf (not \efr{})}
	For the profile $R$ in~\Cref{fig:fhr}, the following assignment $P$ is the outcome of \upre{}.
	\begin{center}
    	\centering
    	\begin{tabular}{r|cccccc}
    		\multicolumn{7}{c}{Assignment $P$}\\
    		&  a & b & c & d & e & f\\\hline
    		$1$ & $1$ & $0$ & $0$ & $0$ & $0$ & $0$\\
    		$2$ & $0$ & $1$ & $0$ & $0$ & $0$ & $0$\\
    		$3$-$5$ & $0$ & $0$ & $1/4$ & $1/12$ & $1/3$ & $1/3$\\
    		$6$ & $0$ & $0$ & $1/4$ & $3/4$ & $0$ & $0$\\
    	\end{tabular}
    \end{center}
	The deterministic assignment $A^*$ in the following, where $j\gets o$ means agent $j$ is allocated item $o$, is the one in~\Cref{eg:fcr} which is not \fhcr{}.
	It is easy to see that $A^*$ must be among those which constitute the convex combination for $P$, which means that $P$ does not satisfy \efr{}.
	\begin{equation*}
		\begin{split}
			A^*:&1\gets a,2\gets b,3\gets c,4\gets e,5\gets f,6\gets d\\
		\end{split}
	\end{equation*}
	
	\vspace{0.5em}\noindent{\bf (not \sdrf{})} We continue to use the instance with $R$ in~\Cref{fig:fhr}. In the assignment $P$ above , where $p_{6,d}>0$ but $\rk{3}{d}<\rk{2}{d}$ and $\sum_{o\in\ucs(\succ_3,d)}p_{3,o}<1$, which violates \sdrf{}.
	
	\vspace{0.5em}\noindent{\bf (not \sdefa{}, not \sdspa{})}  It follows from \Cref{prop:impsdcfr1,prop:impsdcfr2} and the fact that it satisfies \sdcrf{} and \etep{}.
\end{proof}

\subsection{Running Time Analysis for \am{} and \upre{}}\label{sec:app:results:time}

\begin{restatable}{prop}{propebmtime}{}\label{prop:ebmtime}
Given a profile $R$, the deterministic assignment $\text{\am{}}(R)$ can be computed in polynomial time in the number of agents, if $G$ is a polynomial time algorithm\footnote{Such algorithm exists, like Xorshift RNG~\cite{thomson1958modified} and linear congruential generator~\cite{JSSv008i14}}.
\end{restatable}
\begin{proof}
	In~\Cref{alg:am},~Line~2 is the initial setting which takes $O(n^2)$ time, and the {\bf While} loop is executed at most $n$ times because there is at least one item is allocated in each round.
	Below, we analyze the time for each line in the main body of the {\bf While} loop.
	
	For Line~5, identifying the top item for each agent among $M'\subseteq M$, takes $O(n\cdot n)$ time.
	
	For Line~6, it issues a lottery for each $o$  over $N_o$ and the implementation runs in in polynomial time with respect to $n$ by the condition.
	Here the range is $\lvert N_o\rvert<n$, which means the implementation of lottery is in polynomial time with respect to $n$.
	
	Line~7 take $O(n')$ time where $n'$ is the number of items being allocated at that round, and it takes $O(n)$ time in total since at most $n$ items are allocated in one run.
	Together we have that~\Cref{alg:am} runs in polynomial time.
\end{proof}

\begin{remark}
\rm
	Although it is in polynomial time for \am{} to output a deterministic assignment as shown in~\Cref{prop:ebmtime}, there is no guarantee for the time complexity of computing the expected results of \am{}. We conjecture that it is  \#P-complete to compute the expected output of \am{}, just as computing the expected result of RP, where the priority order is generated randomly and uniformly is \#P-complete~\cite{Saban13:Complexity}.
\end{remark}

\begin{restatable}{prop}{propupretime}{}\label{prop:upretime}
Given a profile $R$, the random assignment $\text{\upre{}}(R)$ can be computed in polynomial time in the number of agents.
\end{restatable}
\begin{proof}
	Recall that \upre{} is~\Cref{alg:pcr} using~Eq~(\ref{eq:eatingfunction}) as eating functions, and then $\int_{t_j}^{t_j+\rho}\omega_j{\rm d}t=\min(\rho,1-t_j)$ when $t_j<1$, which can be computed in $O(1)$ time.
	In~\Cref{alg:pcr}, Line~2 is the initial setting which takes $O(n^2)$ time.
	The {\bf While} loop is executed at most $n$ times because an agent consumes different items in each round.
	Below, we analyze the time for each line in the main body of the {\bf While} loop.
	
	On {\bf Line~4}, identifying the top item for each agent among $M'\subseteq M$, takes $O(n\cdot n)$ time.
	
	For~{\bf Line~6.1}, for each $o$,	we can compute $\rho_o$ with the following steps:
	\begin{enumerate}[label={\bf Step \arabic*.},leftmargin=*,topsep=0.5em,itemsep=0em]
		\item Sort agents in $N_o$ by $1-t_j$ in increasing order and obtain the sequence $j_{1},j_{2},\dots,j_{n_o}$ where $n_o=\lvert N_o\rvert$, which takes $O(n^2)$ time since $n_o\le n$;
		\item For each $i\in\{1,\dots,n_o\}$, test if $\sum_{k\in N_{o}}\int^{t_k+\rho}_{t_k}\omega_{k}(t){\rm d}t\le s(o)$ with $\rho=1-t_{j_i}$ and stop when it is not. This takes $O(n)$ time;
		\item If $i'<n_o$ is the maximum value for which the test in step (2) passes, then $\rho_o\ge 1-t_{j_{i'}}$ and computing $\rho_o=\max(\{\rho\mid \rho\cdot (n_o-i')+\sum_{1}^{i'}(1-t_{i})\le s(o)\})\allowbreak=\frac{s(o)-\sum_{i=1}^{i'}(1-t_{i})}{n_o}$ takes $O(1)$ time; if $i=n_o$, then $\rho_o= 1-t_{j_{n_o}}$.
	\end{enumerate}
	In this way, we have that Line~6.1 runs in $O(n^2)$ time.
	
	Line~6.2 needs us to perform one integration for each agent and can also be computed in polynomial time since each integration can be done in $O(1)$ time.
	
	Line~7 updates the supply of each $o\in M'$, Line~8 updates $t_j$ for each agent $j$, and~Line~9 checks if $s(o)=0$, each of which needs addition/subtraction for no more than $n$ times.
	
	Together we have that \upre{} runs in polynomial time.
\end{proof}

\subsection{\am{} is Not A Member of \abm{}}\label{sec:app:results:abm}

    
    We show that \am{} is not a member of \abm{} by proving that there does not exist a distribution $\pi$ over all the priority orders such that \abm{}$^\pi(R)=\am{}(R)$ for any preference profile $R$.
    
    For an instance of assignment problems with $N=\{1,2,\dots,5\}$ and $M=\{a,b,\dots,e\}$, there are $\Myperm[5]{5}=120$ priority orders in total.
    For the preference $R^*$ where all the agents have an identical preference, we trivially have that there are $\Myperm[5]{5}=120$ possible outcomes of \am{}, each with the same probability.
    By~\Cref{thm:abmchar}, we also have that each possible $\am{}(R)$ corresponds to a unique priority order.
    Then $\mathbb E(\am{}(R^*))$ coincides with $\mathbb E(\text{ABM}^{\mU}(R^*))$ where $\mU$ is the uniform distribution over all the priority orders.
    
    Now we show that \am{} is different from \abm{}$^{\mU}$ by comparing the probability of agent $1$ obtaining $c$ in the outcomes of \am{} and \abm{} applied to the profile $R$ for $5$ agents below:
    \begin{equation*}
            1\text{-}2: a\succ c\succ \{others\},\quad 3\text{-}5: b\succ c\succ \{others\}.
    \end{equation*}

    First we consider \am{}.
    Let $A=\am{}(R)$ be a possible outcome of $\am{}$ and $P=\mathbb E(\am{}(R))$. 
    In the first round, \am{} issues a lottery for $a$ among agents $1$-$3$ and one for $b$ among $4,5$.
    Then $Pr(A(1)\neq a)=1/2$.
    In the second round, \am{} issues a lottery for $c$, and the number of participants is always $3$, which means that $Pr(A(1)=c\mid A(1)\neq a)=1/3$.
    It follows that $p_{1,c}=Pr(A(1)=c)=1/6$.
    
    Then we consider ABM$^{\mU}$.
    Let $Q=\mathbb E(\text{ABM}^{\mU}(R))$.
    Agent $1$ gets $c$ when in the priority order,
    \begin{enumerate}[label=(\arabic*)]
        \item agent $2$ is ranked above agent $1$,
        \item at least two agents $j,k\in\{3,4,5\}$, i.e., the agents who do not get $b$, are ranked below agent $1$.
    \end{enumerate}
    The priority orders satisfying this where $j=4$ and $k=5$ are the ones corresponding to the topological orderings in\Cref{fig:tp2}, of which there are $6$ as listed below:
    \begin{equation*}
        \begin{split}
            3\impord{}2\impord{}1\impord{}4\impord{}5,\quad
            3\impord{}2\impord{}1\impord{}5\impord{}4,\\
            2\impord{}3\impord{}1\impord{}4\impord{}5,\quad
            2\impord{}3\impord{}1\impord{}5\impord{}4,\\
            2\impord{}1\impord{}3\impord{}4\impord{}5,\quad
            3\impord{}1\impord{}3\impord{}5\impord{}4.\\
        \end{split}
    \end{equation*}
  With fact that $j,k$ are chosen in $\{3,4,5\}$, there are $6\cdot \Mycomb[2]{3}=18$ priority orders where agent $1$ gets $c$, i.e., $q_{1,c}=18/120=3/20$.
  
      \begin{figure}[ht]
	\centering
	\includegraphics[width=.5\linewidth]{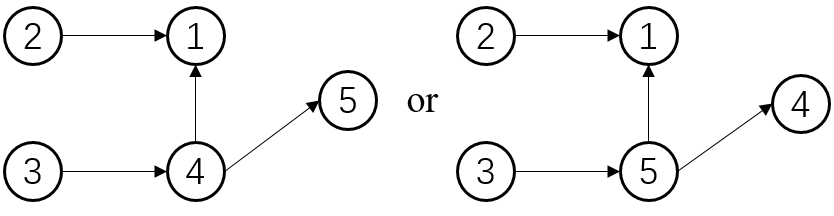}
	\caption{Possible topological orders of priority orders where agent $1$ gets $c$.}\label{fig:tp2}
    \end{figure}

  Together we see that $q_{1,c}\neq p_{1,c}$, which means that $Q\neq P$.
  Therefore, while \abm{}$^\pi$ is uniquely equal to \am{} for the preference profile $R^*$ where agents have an identical preference, \abm{}$^\pi$ is not equal to \am{} for $R$ above. It follows that \am{} is not a member of \abm{}.

\section{Omitted Proofs}\label{sec:app:proof}

\thmamp*

\ifx\oldproof\undefined
\begin{proof}
	Consider an arbitrary instance $(N,M)$ and a strict linear preference profile $R$, and let $P=\mathbb E(\text{\am}(R))$.
	
	\vspace{0.5em}\noindent{\bf Part 1: $\mathbb E(\text{\am}(R))$ is \efcr{}.}\quad

	Let $A=\text{\am}(R)$ be any one of the possible outcomes of \am{} applied to $R$. We show that $A$ is \fhcr{}. 
	We prove by mathematical induction that the items to be assigned at each round $r$ of~\Cref{alg:am} are exactly those in $\tps{A}{r}$ (defined in~\Cref{dfn:fhcr}), and that $A$ is \fhcr{} by showing that every item $o\in \tps{A}{r}$ is assigned to an agent most eager for it in round $r$, i.e.,
	\begin{equation}\label{eq:thm:amchar:1}
		o=\tp{A^{-1}(o)}{M\setminus\bigcup_{r^*<r}\tps{A}{r^*}}.
	\end{equation}
	
	\vspace{0.5em}\noindent{\bf Base case.} When $r=1$, any $o\in \tps{A}{1}$ satisfies that $o=\tp{j}{M}$ for some $j$.
	In~\Cref{alg:am}, all such agents are in $N_o$ on Line~5, and $o$ is assigned to one of them by Line~6 at that round.
	It means that $o=\tp{A^{-1}(o)}{M}$, which is equivalent to Eq~(\ref{eq:thm:amchar:1}) when $r=1$.
	
	\vspace{0.5em}\noindent{\bf Inductive step.} Consider the case that $r>1$. Suppose that for each $r'<r$, $\am{}(R)$ assigns the items in $\tps{A}{r'}$ during round $r'$ and it holds for every item $o'\in\tps{A}{r'}$, that $o'=\tp{A^{-1}(o')}{M\setminus\bigcup_{r^*<r'}\tps{A}{r^*}}$. We will show that at the end of round $r$, for every $o\in \tps{A}{r}$, $o=\tp{A^{-1}(o)}{M\setminus\bigcup_{r^*<r}\tps{A}{r^*}}$.
	By the assumption and Line~3, we have that at the beginning of round $r$:
	\begin{itemize}[leftmargin=*,itemsep=0em,topsep=1em]
	    \item the set of remaining items is $M'=M\setminus\bigcup_{r^*<r}\tps{A}{r^*}$, since items in $\bigcup_{r^*<r}\tps{A}{r^*}$ are allocated before round $r$, and
	    \item for any agent $k\in N'$ who has not received any item yet, it holds that $A(k)\notin \tps{A}{r'}$ for any $r'<r$.
	\end{itemize}
	Therefore, if there exists an agent $j\in N_o\subseteq N'$ (i.e., $o=\tp{j}{M'}$) on Line~5, then an agent in $N_o$ gets $o$ by Line~6, i.e., $A^{-1}(o)\in N_o$, which implies Eq~(\ref{eq:thm:amchar:1}).
	
	By the induction hypothesis we have that Eq~(\ref{eq:thm:amchar:1}) holds for any $o\in \tps{A}{r}$ with $r\ge1$, which means that $A$ satisfies \fhcr{}.
	It follows that $\mathbb E(\text{\am}(R))$ is \efcr{}, and therefore \epopt{} by~\Cref{prop:fhcr}.
	
	\vspace{0.5em}\noindent{\bf Part 2: $\mathbb E(\text{\am}(R))$ is \sdwef{}.}\quad
	
	Before proceeding with the proof we introduce some notation for convenience:
	\begin{itemize}[label=-,itemindent=1em,leftmargin=0pt,topsep=0pt,itemsep=0pt]
	    \item For ease of exposition, we use a possible world, denoted $w$, to represent one possible execution of \am{}, and $\text{\am}^w(R)$ to be the corresponding deterministic assignment.
    	It is easy to see that if $w\neq w'$, then $\text{\am}^w(R)\neq \text{\am}^{w'}(R)$.
    	Let $W(R)$ be the set of all possible worlds for the given instance with $R$, and $W$ for short when $R$ is clear given the context.
    	The probability of $w$, denoted $Pr(w)$, can be computed according to lotteries in each round.

	    \item We use $l$ to refer to a lottery and $N(l)$ be the set of agents who participate in $l$.
	    Let $L(w,r)$ denote the set of lotteries in round $r$ of world $w$ ($w$ can be omitted when clear), and $r(w)$  be the total rounds of $w$.
    	Specially, $l^w_o$ refers to the lottery for item $o$ in $w$.
    	Then we have that 
    	$$Pr(w)=\Pi_{l\in L(r),r\le r(w)}\frac{1}{\lvert N(l)\rvert}=\Pi_{o\in M}\frac{1}{\lvert N(l^w_o)\rvert}$$ 
    	since every item can only be allocated once through a lottery.
    	Let $Pr(W')=\sum_{w\in W'}Pr(w)$ for $W'\subseteq W$.
    	If $W'$ is the set of all the worlds with the same lotteries and winners for the first $r$ rounds, then $Pr(W')=\Pi_{l\in L(r'),r'\le r}\frac{1}{\lvert N(l)\rvert}$.
    	For $P=\mathbb E(\text{\am}(R))$, we have that $p_{j,o}=Pr(\{w\in W\mid \am^w(R)(j)=o\})$, i.e., the probability of all the worlds where $j$ gets items $o$.

    	\item For ease of exposition, we use $M^w_{r},N^w_{r},N^w_{o,r}$ to refer to the values of variables $M',N',N_{o}$ at the beginning of each round $r$ during the execution of~\Cref{alg:am} in the world $w$, and omit $w$ when it is clear from the context.
	\end{itemize}
	
	\vspace{1em}
	We now proceed with the proof. Consider an arbitrary pair of agents $j$ and $k$ such that $P_{k}\sd{j}P_j$ and $j\neq k$.
	Without loss of generality, let $\succ_j$ be $o_1\succ_j o_2 \succ_j \cdots \succ_j o_n$.
	We show by mathematical induction, on the rank $i=1,2,\dots,n$ with respect to $\succ_j$, that the following conditions hold:
	\begin{enumerate}[label={\bf Condition (\arabic*):},align=left,wide,labelindent=0pt,itemsep=0em,topsep=0pt]
		\item if there exits a round $r$ such that $j\in N_{o_i,r}$ and $k\in N_r$ in a certain world $w\in W$, then $k\in N_{o_i,r}$,
		\item if in a certain world $w\in W$, $j$ gets some $o\succ_j o_i$ at round $r'$, then for each $r>r'$ with $k\in N_r$ and $M_{r}\cup\ucs(\succ_j,o_i)\neq\emptyset$, we have that $k\in N_{\tp{j}{M_r},r}$, and
		\item $p_{k,o_i}=p_{j,o_i}$.
	\end{enumerate}

	\vspace{0.5em}\noindent{\bf Base case.} First we prove conditions (1)-(3) for $i=1$.
	Since no $o\succ_j o_1$, condition (2) holds trivially.
	For every possible world $w\in W$, we have that $j\in N_{o_1,1}$ and $k\in N_1=N$.
	If  $k\notin N_{o_1,1}$, then $k$ does not participate in the lottery for $o$ , and she does not get $o$ in any $w$, which means that $p_{k,o_1}=0<p_{j,o_1}$, a contradiction to $P_{k}\sd{j}P_j$.
	Therefore condition (1) holds for $i=1$. It follows that $p_{j,o_1}=Pr(\{w\in W\mid \text{\am}^w(R)(j)=o_1\})=Pr(\{w\in W\mid \text{\am}^w(R)(k)=o_1\})=p_{k,o_1}$, i.e., condition (3) holds for $i=1$.
	
	\vspace{0.5em}\noindent{\bf Inductive step.} Assume that conditions (1)-(3) hold for $i'<i$. We now show that they also hold for $i$. First, we show that $p_{k,o_i}\le p_{j,o_i}$ by comparing the probabilities of worlds where $k$ gets $o_i$ with those where $j$ gets $o_i$ in cases \romannumeral1~--~\romannumeral3~below.
	We note that we only consider the worlds where $k$ is possible get $o_i$.
	
	\vspace{0.5em}\noindent{\bf Case 1:}
	Consider any world $w'\in W$ where agent $j$ does not get an item ranked higher than $o_i$ according to $\succ_j$, agent $k$ does not get an item ranked higher than $o_i$ according to $\succ_k$.
	
	\vspace{0.5em}\noindent{\bf Case 1.1:} 
    Suppose there does not exist a round $r$ that $j\in N_{o_i,r}$. Then, clearly $\text{\am}^{w'}(R)(j)\neq o_i$. Moreover, there does not exists any round $r^*$ with $k\in N_{o_i,r^*}$ either. Suppose for the sake of contradiction that such a round $r^*$ exists, then it must hold that $o_i\in M_{r^*}$. Now suppose $j$'s top item in $M_{r^*}$ is $o_{i^*}=\tp{j}{M_{r^*}}\succ_j o_i$, i.e, that $i^*<i$.
	However, by condition (1), for every $i'<i$, when $j$ applies for $o_{i'}$, $k$ does too, which implies that $i^*\ge i$ since $o_{i^*}\neq o_i$, a contradiction.
	Together we see that both agents do not get $o_i$ in $w'$, and therefore we can ignore such a world $w'$.

	\vspace{0.5em}\noindent{\bf Case 1.2:} 
	Suppose there exists a round $r$ such that $j\in N_{o_i,r}$. 
	Let $W_1$ be the set of worlds where lotteries conducted and winners are the same as $w'$ for any round $r'<r$. Therefore, in every world in $W_1$, we have the same $M_{r},N_{r},N_{o,r}$ as $w'$ at beginning of round $r$.
	By the selection of $w'$, we have that $j,k\in N_{r}$.
	Below, we compare the probabilities that $j$ and $k$ get $o_i$ in $W_1$ through a case analysis.

    \begin{itemize}[topsep=0pt,itemsep=0em,leftmargin=0pt,itemindent=1em]
	\item If $k\in N_{o_i,r}$, then $k$ participates in the lottery for $o_i$ at round $r$ and her chance to win is equal to $j$'s, which means that
	\begin{equation*}
	\begin{split}
	    & Pr(\{w\in W_1\mid \am^w(R)(j)=o_i\})\\
	    =& Pr(W_1)\cdot\frac{1}{\lvert N(l_{o_i})\rvert}\\
	    =& Pr(\{w\in W_1\mid \am^w(R)(k)=o_i\}).
	\end{split}
	\end{equation*}

	\item If $k\notin N_{o_i,r}$, then by $j\in N_{o_i,r}\neq\emptyset$, item $o_i$ is allocated to some agent in $N_{o_i,r}$, subsequently removed at the end of round $r$, and never appears in later rounds This means that $Pr(\{w\in W_1\mid \am^w(R)(j)=o_i\})> Pr(\{w\in W_1\mid \am^w(R)(k)=o_i\})=0$.
	\end{itemize}
	
	\vspace{0.5em}\noindent{\bf Case 2:}
	Consider any world where exactly one of agents $j$ and $k$ gets an item ranked higher than $o_i$ according to $\succ_j$ and $\succ_k$ respectively.
	\xiaoxi{In the old version I want say that the worlds discussed in this case satisfies that $j$ and $k$ get items better than $o_i$. But now I find it seems to deviate from my meaning if we put it as the first sentence. We can only consider case 2.1, since the worlds in case 2.2 do not allow $k$ to get $o_i$}
	
	\vspace{0.5em}\noindent{\bf Case 2.1:} Consider any world $w$ in which $j$ gets a higher ranked item $o_h\succ_j o_i$ at a round $r'$ and $k$ does not get an item ranked higher than $o_i$ in $\succ_k$. Since agent $k$ does not get an item it ranks higher than $o_i$, either $k$ gets $o_i$ or it gets an item it ranks strictly lower than $o_i$. We do not need to consider the worlds in which $k$ does not get $o_i$. We analyze an arbitrary world in which $k$ gets $o_i$ instead. Now, suppose $k\in N_{o_i,r^*}$ at a round $r^*$.
	\xiaoxi{NOTE: although the proof is easy to read with $k$ and $r^*$, but in this way we cannot prove condition(2) for rank $i$.}
	\begin{enumerate}[label={\bf Case 2.1.\arabic*.},align=left]
	
	    \item Suppose that $j\in N_{r^*}$. We will show that $k$ does not get $o_i$, meaning we do not need to consider such a world. Let $o_g$ be the top ranked item in $M_{r^*}$ according to $\succ_j$, i.e., $o_g=\tp{j}{M_{r^*}}$. Then, clearly $o_g\succ_j o_i$, i.e. $g<i$. Then, by our inductive assumption, and specifically, by condition (1), $k\in N_{o_g,r^*}$, a contradiction.\xiaoxi{move it to a higher level as 2.1? and the rest at 2.2. It might be not good to have 4 levels}
	
	    \item\label{case:sdwef212} Suppose that $j\not\in N_{r^*}$. This immediately implies that $r^*>r'$. 
	    
	    \begin{enumerate}[label={\bf{\ref{case:sdwef212}\alph*.}},align=left]
	        
	        \item Now, if $\tp{j}{M_{r^*}}=o_g\neq o_i$, then by the observation that at round $r^*>r'$, $k\in N_{r^*}$ and $M_{r^*}\cup\ucs(\succ_j,o_i)\neq\emptyset$, and by our inductive assumption and condition (2) in particular, it must hold that $k\in N_{o_g,r^*}$, a contradiction. This means we do not need to consider such a world.
	        
	        \item\label{case:sdwef212b} Now, suppose $\tp{j}{M_{r^*}}=o_i$. If $k$ loses the lottery for $o_i$ at round $r^*$, we do not need to consider such a world. Then, consider the world $w_k$ where $k$ wins $o_i$ in $r^*$. We will show that for every such world $w_k$, there exists a world $w_j$ with equal probability in which $k$ wins $o_h$ in round $r'$ and $j$ gets $o_i$ in round $r^*$.
	        
	        Recall that in $w_k$, $j$ gets $o_h\succ_j o_i$. Then, by our inductive assumption and condition (1) in particular, it must hold at every round $r<r'$ that neither $j$ nor $k$ applies to any item $o$ ranked lower than $o_h$ in $\succ_j$. It follows that $r^*>r'$.
	        
	        Consider the world $w_j$ in which:
	        \begin{itemize}
	            \item for every round $r\in\{1,\dots,r'-1,r'+1,\dots,r^*\}$, the lotteries and winners are identical to those in $w_k$, 
	            \item at round $r'$ all of the lotteries are identical to those in $w_k$, and all of the winners are identical, except for the lottery for item $o_h$,
	            \item agent $k$ wins the lottery for $o_h$ in round $r'$, and
	            \item agent $j$ wins the lottery for $o_i$ in round $r^*$.
	            \item for every round $r\in\{r^*+1\dots\}$, the lotteries and winners are identical to those in $w_k$.
	        \end{itemize}
	        Such a world $w_j$ must exist because:
	        \begin{itemize}
	            \item By our assumption, $h<i$ and by condition (1) of our inductive assumption, if at round $r'$, $j\in N_{o_h,r'}$ then $k\in N_{o_h,r'}$. This means agent $k$ is a possible winner of the lottery for $o_h$, and necessarily, $o_h\succ_k o_i$. 
	            \item Then, due to our inductive assumption that condition (2) is true for all pairs of agents for all $g<i$, it must hold that for every world $w'\in W_k$, there is a world $w''\in W_j$ such that $j$  participates in every lottery that agent $k$ participates in from round $r'+1$ to round $r^*$, and then $j$ wins the lottery for $o_i$ in round $r^*$.
	        \end{itemize} 
	        Now, notice that the number of agents applying for $o_h$ at round $r'$ and applying for $o_i$ in round $r^*$ are identical in every world in both $w_j$ or $w_k$, i.e., $|N_{o_h,r'}^{w_j}|=|N_{o_h,r'}^{w_k}|$ and $|N_{o_i,r^*}^{w_j}|=|N_{o_i,r^*}^{w_k}|$. Finally, it is easy to see that after the end of round $r^*$, the remaining items and agents are identical in $w_j$ and $w_k$. It follows from the choice of $w_k$ and construction of $w_j$ that they occur with equal probability.
	    \end{enumerate}
	\end{enumerate}
	
	\vspace{0.5em}\noindent{\bf Case 2.2:}\xiaoxi{we might not need this part, because we are consider the worlds where $k$ can get $o_i$} Consider any world $w$ in which $k$ gets an item $o_h$ ranked higher than $o_i$ in $\succ_j$, i.e., $o_h\succ_j o_i$ at a round $r'$ and $j$ does not get an item ranked higher than $o_i$ in $\succ_j$. By the argument in~\ref{case:sdwef212b}, for every such world $w_j$ where $k$ gets $o_h$ and $j$ gets $o_i$, there is a unique corresponding world $w_k$ where $j$ gets $o_h$ and $k$ gets $o_i$ and $w_j$ and $w_k$ occur with equal probability.
	

	\vspace{0.5em}\noindent{\bf Case 3:}\xiaoxi{we might not need case this either. Agent $k$ always does not get a better item than $o_i$ w.r.t. $\succ_k$, and there are two cases in total: $j$ does not get (case 1) a better item than $o_i$ w.r.t. $\succ_j$, or gets one such item (case2).}
	We do not need to consider worlds where both agents $j$ and $k$ get items ranked higher than $o_i$ in $\succ_j$, since in this case, neither of them gets $o_i$.
	
	\vspace{0.5em}\noindent{\bf Concluding the inductive step.} From cases 1-3, we have covered all the worlds that agent $k$ can possibly get item $o_i$, and we obtain that $p_{k,o_i}\le p_{j,o_i}$.
	With the assumption that $P_{k}\sd{j}P_j$ and condition (3) holds for $i'<i$, it follows that  $p_{k,o_i}\ge p_{j,o_i}$.
	Therefore we have $p_{k,o_i}=p_{j,o_i}$, i.e., condition (3) holds for $i$.
	The equality also requires $k\in N_{o_i,r}$ in case (\romannumeral1) and $o_i=\tp{j}{M_r}=\tp{k}{M_r}$ in case (\romannumeral2), i.e., conditions (1) and (2) for $i$.
	
	With the induction above, we prove that $p_{k,o_i}=p_{j,o_i}$ for any $i$. It follows that if $P_k\succ_j P_j$, $P_k=P_j$.
	
	\vspace{0.5em}\noindent{\bf Part 3: $\mathbb E(\text{\am}(R))$ is \sdsp{}.}\quad
    We continue to use the new notations introduced at the beginning of Part 2.
	Without loss of generality, let $\succ_j$ be $o_1\succ_j o_2 \succ_j \cdots$.
	Let profile $R'=(\succ'_j,\succ_{-j})$ where $\succ'_j$ is any preference that agent $j$ misreports, and $Q=\am(R')$.
	Assume that $Q_j\sd{} P_j$.
	We show by mathematical induction, for rank $i=1,2,\dots$ with respect to agent $j$, that the following conditions hold:
	\begin{enumerate}[label={\bf Condition (\arabic*):},align=left,wide,labelindent=0pt,itemsep=0.5em,topsep=1em]
		\item when $j\in N^{w}_{o_i,r}$ in a world $w\in W(R)$,
		for any $w'\in W(R')$ where lotteries and winners are the same as $w$ before round $r$, we have that $j\in N^{w'}_{o_i,r}$, and
		\item $p_{j,o_i}=q_{j,o_i}$.
	\end{enumerate}

	\vspace{0.5em}\noindent{\bf Base case.}
	First we show condition (1) for $i=1$.
	It is easy to see that in any $w\in W(R)$, $j$ applies for $o_1$ at round $1$, i.e., $j\in N^{w}_{o_1,1}$.
	We claim that $j\in N^{w'}_{o_1,1}$ for any $w'\in W(R')$.
	Otherwise, if $j\notin N^{w'}_{o_1,1}$ in some $w'$, we show that both of the possible cases below lead to a contradiction to our assumption that $Q_j\sd{} P_j$.
	
	\begin{itemize}[topsep=0pt,itemindent=1em,itemsep=0em,leftmargin=0pt]
	    \item When $N^{w'}_{o_1,1}\neq\emptyset$, $o_1$ is assigned to some agent in $N^{w'}_{o_1,1}$ in $w'$. It follows that $p_{j,o_1}>q_{j,o_1}=0$, a contradiction to the assumption.
	    \item When $N^{w'}_{o_1,1}=\emptyset$, i.e., $N^{w}_{o_1,1}=\{j\}$, $o_1$ is assigned to the only applicant $j$ in $w$, while she applies for item $o'\neq o_1$ in $w'$ and $o_1\succ_j o'$ trivially.
	    It follows that
    	\begin{equation*}
    	\begin{split}
    		p_{j,o_1}=&Pr(\{w\in W(R)\mid \text{\am}^w(R)(j)=o_1\})=1\\
    		>&1-Pr(\{w\in W(R')\mid \text{\am}^w(R')(j)=o'\})\\
    		\ge&Pr(\{w\in W(R')\mid \text{\am}^w(R')(j)=o_1\})=q_{j,o_1},\\
    	\end{split}
    	\end{equation*}
    	a contradiction to the assumption.
	\end{itemize}
	
	In this way, we have $j\in N^{w'}_{o_1,1}$ for any $w'\in W(R')$, i.e., condition (1) for $i=1$, which means that $\lvert N(l^w_{o_1})\rvert=\lvert N(l^{w'}_{o_1})\rvert$ and
	\begin{equation*}
		\begin{split}
			p_{j,o_1}=&Pr(\{w\in W(R)\mid \text{\am}^w(R)(j)=o_1\}=\frac{1}{\lvert N(l^w_{o_1})\rvert}\\
			=&\frac{1}{\lvert N(l^{w'}_{o_1})\rvert}=Pr(\{w\in W(R')\mid \text{\am}^w(R')(j)=o_1\})=q_{j,o_1},\\
		\end{split}
	\end{equation*}
	 i.e., condition (2) for $i=1$.
	
	\vspace{0.5em}\noindent{\bf Inductive step.} Supposing conditions (1) and (2) hold for $i'<i$, we show that they also hold for $i$.
	First we show condition (1) for $i$.
	For an arbitrary world $w^*\in W(R)$ with $j\in N^{w^*}_{o_i,r}$, let $W_1\subseteq W(R)$ and $W_2\subseteq W(R')$ be the sets of worlds where lotteries and winners are the same as $w^*$ before round $r$ with respect to $R$ and $R'$, respectively.
	By construction of $W_1$ and $W_2$, $Pr(W_1)=Pr(W_2)$.
	For any $w\in W_1$ and $w'\in W_2$, $M^w_r=M^{w'}_r$, $N^{w}_r=N^{w'}_r$, and $N^{w}_{o,r}=N^{w'}_{o,r}$ for any $o\in M^w_r\setminus\{o_i\}$.
	We have that $j\in N^{w}_{o_i,r}$ by the selection of $w^*$, and we claim that $j\in N^{w'}_{o_i,r}$ for any $w'\in W(R')$.
	Otherwise, if $j\notin N^{w'}_{o_i,r}$ in some $w'$, we show that both of the possible cases below lead to a contradiction to our assumption that $Q_j\sd{} P_j$.
	
	\begin{itemize}[topsep=0pt,itemindent=1em,itemsep=0em,leftmargin=0pt]
	    \item 	When $N^{w'}_{o_i,r}\neq\emptyset$, item $o_i$ is assigned to some agent in $N^{w'}_{o_i,r}$ in $w'$. It follows that
	\begin{equation}\label{eq:thm:amp:1}
		\begin{split}
			&Pr(\{w\in W_1\mid \text{\am}^w(R)(j)=o_i\})=Pr(W_1)\\
			>&Pr(\{w\in W_2\mid \text{\am}^w(R')(j)=o_i\})=0.
		\end{split}
	\end{equation}
	With condition (1) for $i'<i$, in world $w'$, agent $j$ can only apply for $o_i$ at round $r'\ge r$ not earlier than she does in $w$, which means that $p_{j,o_i}>q_{j,o_i}=0$ with Eq~(\ref{eq:thm:amp:1}).
	Together with condition (2) for $i'<i$, we have a contradiction to the assumption that $Q_j\sd{} P_j$.
	\item 	When $N^{w'}_{o_i,r}=\emptyset$, i.e., $N^{w}_{o_i,r}=\{j\}$, item $o_i$ is assigned to the only applicant $j$ in $w$ while she applies for item $o'\neq o_i$ in $w'$ and $o_i\succ_j o'$ by the selection.
	It follows that
	\begin{equation*}
		\begin{split}
			&Pr(\{w\in W_1\mid \text{\am}^w(R)(j)=o_i\})=Pr(W_1)\\
			>&Pr(W_2)-Pr(\{w\in W_2\mid \text{\am}^w(R')(j)=o'\})
			\ge Pr(\{w\in W_2\mid \text{\am}^w(R')(j)=o_i\}).
		\end{split}
	\end{equation*}
	This means that $p_{j,o_i}>q_{j,o_i}$, a contradiction to the assumption that condition (2) holds for $i'<i$.
	\end{itemize}

	In this way, we have condition (1) for $i$, which means that $\lvert N(l^w_{o_i})\rvert=\lvert N(l^{w'}_{o_i})\rvert$ and
	\begin{equation*}
		\begin{split}
			&Pr(\{w\in W_1\mid \text{\am}^w(R)(j)=o_i\}=Pr(W_1)\cdot\frac{1}{\lvert N(l^w_{o_i})\rvert}\\
			=&Pr(W_2)\cdot\frac{1}{\lvert N(l^{w'}_{o_i})\rvert}=Pr(\{w\in W_2\mid \text{\am}^w(R')(j)=o_1\}),\\
		\end{split}
	\end{equation*}
	which implies $p_{j,o_i}=q_{j,o_i}$, i.e., condition (2) for $i$.
	
	By mathematical induction, we have that $p_{j,o}=q_{j,o}$ for any $o$, and therefore if $Q_j\sd{} P_j$, that $Q_j=P_j$ .
	
	\vspace{0.5em}\noindent{\bf Part 4: $\mathbb E(\text{\am}(R))$ is \etep{}.}\quad
	
	For agents $j$ and $k$, we prove that $p_{j,o}=p_{k,o}$ for any item $o$ appearing in $\succ_{j,k}$.
	We compare probability of possible worlds where agent $j$ gets $o\in\ucs(\succ_{j,k},o_m)$ with those where agent $k$ gets $o$.
	
	First we consider the world $w$ where $j$ gets $o$ at round $r$ and $k$ gets $o'\in\ucs(\succ_{j,k},o_m)$ at round $r'$.
	Let $w'$ satisfy $j$ gets $o'$ at round $r$, $k$ gets $o$ at round $r'$, and the result of other lotteries keep the same as $w$.
	In $w'$, we see that $k$ wins the lottery for $o$ instead of $j$, and $j$ participates in lotteries at rounds $r+1$ to $r'$ instead of $k$.
	We also see that for every lottery $l_o$, $\lvert N(l_o)\rvert$ keep the same in worlds $w$ and $w'$.
	Therefore we have that $Pr(w)= Pr(w')$.
	
	Then we consider the world $w$ where $j$ gets $o$ at round $r$ and $k$ does not get items appearing in $\succ_{j,k}$.
	Let $o'$ be the last item $k$ applies for in $\succ_{j,k}$ at round $r'$, and $W_j$ be the set of worlds which are the same as $w$ from rounds $1$ to $r'$.
	Here the probability of $W_j$ can also be computed as $Pr(W_j)=\Pi_{l\in L(r^*),r^*\le r'}\frac{1}{\lvert N(l)\rvert}$.
	We construct another set $W_k$ such that for any $w\in W_j$,
	\begin{enumerate*}[label=(\roman*)]
		\item the winners of lotteries are the same as $w$ at round $1$ to $r-1$,
		\item the winner of $l(o)$ is $k$ at round $r$, and any other $l\in L(r)$ is the same as $w$,
		\item $j$ participates in lotteries at rounds $r+1$ to $r'$ instead of $k$.
	\end{enumerate*}
	Then we see that for every lottery $l\in L(r^*)$ with $r^*\le r'$, $\lvert N(l)\rvert$ are the same in any world $w\in W_j$ and $w'\in W_k$.
	Therefore we have that $Pr(W_j)= Pr(W_k)$.
	
	Together we have that $p_{j,o}= p_{k,o}$ for any $o$ appearing in $\succ_{j,k}$.
\end{proof}

\else
\begin{proof}
	Consider an arbitrary instance $(N,M)$ and a strict linear preference profile $R$, and let $P=\mathbb E(\text{\am}(R))$.
	
	\vspace{0.5em}\noindent{\bf Part 1: $\mathbb E(\text{\am}(R))$ is \efcr{}.}\quad

	Let $A=\text{\am}(R)$ be any one of the possible outcomes of \am{} applied to $R$. We show that $A$ is \fhcr{}. 
	We prove by mathematical induction that the items to be assigned at each round $r$ of~\Cref{alg:am} are exactly those in $\tps{A}{r}$ (defined in~\Cref{dfn:fhcr}), and that $A$ is \fhcr{} by showing that every item $o\in \tps{A}{r}$ is assigned to an agent most eager for it in round $r$, i.e.,
	\begin{equation}\label{eq:thm:amchar:1}
		o=\tp{A^{-1}(o)}{M\setminus\bigcup_{r^*<r}\tps{A}{r^*}}.
	\end{equation}
	
	\vspace{0.5em}\noindent{\bf Base case.} When $r=1$, any $o\in \tps{A}{1}$ satisfies that $o=\tp{j}{M}$ for some $j$.
	In~\Cref{alg:am}, all such agents are in $N_o$ on Line~5, and $o$ is assigned to one of them by Line~6 at that round.
	It means that $o=\tp{A^{-1}(o)}{M}$, which is equivalent to Eq~(\ref{eq:thm:amchar:1}) when $r=1$.
	
	\vspace{0.5em}\noindent{\bf Inductive step.} Consider the case that $r>1$. Suppose that for each $r'<r$, $\am{}(R)$ assigns the items in $\tps{A}{r'}$ during round $r'$ and it holds for every item $o'\in\tps{A}{r'}$, that $o'=\tp{A^{-1}(o')}{M\setminus\bigcup_{r^*<r'}\tps{A}{r^*}}$. We will show that at the end of round $r$, for every $o\in \tps{A}{r}$, $o=\tp{A^{-1}(o)}{M\setminus\bigcup_{r^*<r}\tps{A}{r^*}}$.
	By the assumption and Line~3, we have that at the beginning of round $r$:
	\begin{itemize}[leftmargin=*,itemsep=0em,topsep=1em]
	    \item the set of remaining items is $M'=M\setminus\bigcup_{r^*<r}\tps{A}{r^*}$, since items in $\bigcup_{r^*<r}\tps{A}{r^*}$ are allocated before round $r$, and
	    \item for any agent $k\in N'$ who has not received any item yet, it holds that $A(k)\notin \tps{A}{r'}$ for any $r'<r$.
	\end{itemize}
	Therefore, if there exists an agent $j\in N_o\subseteq N'$ (i.e., $o=\tp{j}{M'}$) on Line~5, then an agent in $N_o$ gets $o$ by Line~6, i.e., $A^{-1}(o)\in N_o$, which implies Eq~(\ref{eq:thm:amchar:1}).
	
	By the induction hypothesis we have that Eq~(\ref{eq:thm:amchar:1}) holds for any $o\in \tps{A}{r}$ with $r\ge1$, which means that $A$ satisfies \fhcr{}.
	It follows that $\mathbb E(\text{\am}(R))$ is \efcr{}, and therefore \epopt{} by~\Cref{prop:fhcr}.
	
	\vspace{1em}\noindent{\bf Part 2: $\mathbb E(\text{\am}(R))$ is \sdwef{}.}\quad
	
	Before proceeding with the proof we introduce some notation for convenience:
	\begin{itemize}[label=-,itemindent=1em,leftmargin=0pt,topsep=0pt,itemsep=0pt]
	    \item For ease of exposition, we will refer to each possible execution of \am{} as a ``world'', denoted $w$, and $\text{\am}^w(R)$ to be the corresponding deterministic assignment output by \am{}.
    	It is easy to see that if $w\neq w'$, then $\text{\am}^w(R)\neq \text{\am}^{w'}(R)$.
    	Let $W(R)$ be the set of all possible worlds for the given instance with $R$, and $W$ for short when $R$ is clearly given in the context.
    	The probability of $w$, denoted $Pr(w)$, can be computed according to the lotteries in each round.

	    \item We use $l$ to refer to a lottery and $N(l)$ be the set of agents who participate in $l$.
	    Let $L(w,r)$ denote the set of lotteries in round $r$ of world $w$ ($w$ can be omitted when clear), and $r(w)$  be the total rounds of $w$.
    	Specially, $l^w_o$ refers to the lottery for item $o$ in $w$.
    	Then we have that 
    	$$Pr(w)=\Pi_{l\in L(r),r\le r(w)}\frac{1}{\lvert N(l)\rvert}=\Pi_{o\in M}\frac{1}{\lvert N(l^w_o)\rvert}$$ 
    	since every item can only be allocated once through a lottery.
    	Let $Pr(W')=\sum_{w\in W'}Pr(w)$ for $W'\subseteq W$.
    	If $W'$ is the set of all the worlds with the same lotteries and winners for the first $r$ rounds, then $Pr(W')=\Pi_{l\in L(r'),r'\le r}\frac{1}{\lvert N(l)\rvert}$.
    	For $P=\mathbb E(\text{\am}(R))$, we have that $p_{j,o}=Pr(\{w\in W\mid \am^w(R)(j)=o\})$, i.e., the probability of all the worlds where $j$ gets items $o$.

    	\item For ease of exposition, we use $M^w_{r},N^w_{r},N^w_{o,r}$ to refer to the values of variables $M',N',N_{o}$ at the beginning of each round $r$ during the execution of~\Cref{alg:am} in the world $w$, and omit $w$ when it is clear from the context.
	\end{itemize}
	
	\paragraph{Start of proof of Part 2.} Consider an arbitrary pair of agents $j$ and $k$ such that $P_{k}\sd{j}P_j$ and $j\neq k$.
	Without loss of generality, let $\succ_j$ be $o_1\succ_j o_2 \succ_j \cdots \succ_j o_n$.
	We show by mathematical induction, on the rank $i=1,2,\dots,n$ with respect to $\succ_j$, that the following conditions hold:
	\begin{enumerate}[label={\bf Condition (\arabic*):},align=left,wide,labelindent=0pt,itemsep=0em,topsep=0pt]
		\item if there exits a round $r$ such that $j\in N_{o_i,r}$ and $k\in N_r$ in a certain world $w\in W$, then $k\in N_{o_i,r}$,
		\item if in a certain world $w\in W$, $j$ gets some $o\succ_j o_i$ at round $r'$, then for  $r>r'$ with $k\in N_r$ and $\tp{j}{M_r}=o_i$, we have that $k\in N_{o_i,r}$, and
		\item $p_{k,o_i}=p_{j,o_i}$.
	\end{enumerate}

	\vspace{0.5em}\noindent{\bf Base case.} First we prove conditions (1)-(3) for $i=1$.
	Since no $o\succ_j o_1$, we have condition (2) trivially true.
	For every possible world $w\in W$, we have that $j\in N_{o_1,1}$ and $k\in N_1=N$.
	If  $k\notin N_{o_1,1}$, then $k$ does not participate in the lottery for $o$ , and she does not get $o$ in any $w$, which means that $p_{k,o_1}=0<p_{j,o_1}$, a contradiction to $P_{k}\sd{j}P_j$.
	Therefore we have condition (1) for $i=1$. It follows that $p_{j,o_1}=Pr(\{w\in W\mid \text{\am}^w(R)(j)=o_1\})=Pr(\{w\in W\mid \text{\am}^w(R)(k)=o_1\})=p_{k,o_1}$, i.e., condition (3) holds for $i=1$.
	
	\vspace{0.5em}\noindent{\bf Inductive step.} Assume that conditions (1)-(3) hold for $i'<i$, we show that they also hold for $i$.
	We show that $p_{k,o_i}\le p_{j,o_i}$ by comparing the probabilities of worlds where $j$ gets $o_i$ with those where $k$ gets $o_i$ in the following cases (\romannumeral1)-(\romannumeral3).
	
	\vspace{0.5em}\noindent{\bf Case \romannumeral1:}
	Consider any world $w'\in W$ where agents $j$ do not get items better than $o_i$ according to $\succ_j$, and agents $k$ do not get items better than $o_i$ according to $\succ_k$.

    We first show that we do not need to consider the case that there does not exist a round $r$ such that $j\in N_{o_i,r}$ in world $w'$.
    If such $r$ does not exist, $\text{\am}^{w'}(R)(j)\neq o_i$, and there does not exists $r^*$ with $k\in N_{o_i,r^*}$ either.
	Otherwise, if such $r^*$ exists, $o_i\in M_{r^*}$.
	Let $o_{i^*}=\tp{j}{M_{r^*}}\succ_j o_i$, which means that $i>i^*$.
	By condition (1) for $i'<i$, when $j$ applies for $o_{i'}$, $k$ does too.
	It follows that $i^*\ge i$ since $o_{i^*}\neq o_i$, a contradiction.
	Together we see that both agents do not get $o_i$ in $w'$, and therefore $w'$ is out of discussion.

	Then we consider the case that there exists a round $r$ such that $j\in N_{o_i,r}$.
	Let $W_1$ be the set of worlds where lotteries and winners are the same as $w'$ for any round $r'<r$, and therefore all the worlds in $W_1$ have the same $M_{r},N_{r},N_{o,r}$ as $w'$ for $r$.
	By selection of $w'$, we have that $j,k\in N_{r}$.
	In the following, we compare the probabilities that $j$ and $k$ get $o_i$ in $W_1$ by cases.

	\vspace{0.5em}\noindent{\bf Case \romannumeral1~(a):} 
	If $k\in N_{o_i,r}$, then she participates in the lottery for $o_i$ at round $r$ and her chance to win is equal to $j$'s, which means that
	\begin{equation*}
	\begin{split}
	    & Pr(\{w\in W_1\mid \am^w(R)(j)=o_i\})\\
	    =& Pr(W_1)\cdot\frac{1}{\lvert N(l_{o_i})\rvert}\\
	    =& Pr(\{w\in W_1\mid \am^w(R)(k)=o_i\}).
	\end{split}
	\end{equation*}

	\vspace{0.5em}\noindent{\bf Case \romannumeral1~(b):} 
	If $k\notin N_{o_i,r}$, then by $j\in N_{o_i,r}\neq\emptyset$, $o_i$ is allocated to some agent in $N_{o_i,r}$ and never appears in later rounds, which means that 
	\begin{equation}
 	\label{eq:thm:amp:2:1}
	Pr(\{w\in W_1\mid \am^w(R)(j)=o_i\})> Pr(\{w\in W_1\mid \am^w(R)(k)=o_i\})=0.
	\end{equation}
	
	\vspace{0.5em}\noindent{\bf Case \romannumeral2:}
	In this case, we consider the world where one of agents $j$ and $k$ gets an item better than $o_i$ with respect to $\succ_j$.
	For any world $w_k\in W$ in which agent $j$ gets item $o_h \succ_j o_i$ at round $r'$ and $k$ does not get any item better than $o_i$ with respect to $\succ_k$, let $r$ satisfy $o_i=\tp{j}{M_r}$.
	
    First we show that we do not need to consider the case where such an $r$ does not exist in $w_k$.
	In that case, we have that $k\notin N_{o_i,r^*}$ for any $r^*$.
	Otherwise, if $k\in N_{o_i,r^*}$ for some $r^*$, then $k\in N_{r^*}$.
	Let $o_{i^*}=\tp{j}{M_{r^*}}$.
	We have that $o_{i^*}\succ_j o_i$, i.e, $i^*<i$, and the fact that $o_{i^*}\neq o_i$ contradicts condition (1) for $i^*$ if $j\in N_{r^*}$, and condition (2) for $i^*$ if $j\notin N_{r^*}$.
	Therefore $k\notin N_{o_i,r^*}$ for any $r^*$ if such $r$ does not exist, which means that $\text{\am}^{w_k}(R)(k)\neq o_i$ and we do not need to consider $w_k$.
	
	Then we consider the case where such $r$ exists.
	Recall that agent $j$ gets item $o_h \succ_j o_i$ at round $r'$, and $r$ satisfy $o_i=\tp{j}{M_r}$.
	By condition (1) for $i'\le h$, neither of agents applies for any $o$ with $o_i\succ_j o$ before round $r'$, and it follows that $r>r'$.
	Let $W_k$ be the set of worlds where lotteries and winners are the same as $w_k$ for any round $r^*<r$.
	Correspondingly, we find a set of worlds $W_j$ such that
	\begin{itemize}[topsep=0pt,itemsep=0em,leftmargin=0pt,itemindent=1em]
		\item for any round in $\{1,\dots,r'-1,r'+1,\dots,r\}$, lotteries and winners are the same as $w_k$,
		\item for round $r'$, lotteries are the same as $w_k$, and so do winners except the one for item $o_h$, and
		\item agent $k$ wins the lottery of $o_h$ at round $r'$.
	\end{itemize}
	We have that $W_j\neq\emptyset$, because $k$ participates in the lottery for $o_h$ at round $r'$ since $k\in N_{o_h,r'}$ by condition (1) for $h$, which means that $k$ is possible to win $o_h$ instead of $j$, and then $j$ participates in the same lotteries instead of $k$ does in $w_k$ till round $r$ by condition (2) for $h<i'\le i$.
	By construction of $W_j$ and $W_k$, $Pr(W_j)=Pr(W_k)$.
	For any $w\in W_j$ and $w'\in W_k$, we have that $M^w_{r}=M^{w'}_r$, $j\in N^w_{r}$, $k\in N^{w'}_{r}$, and $N^w_{r}\setminus\{j\}=N^{w'}_{r}\setminus\{k\}$.
	By selection of $r$ such that $o_i=\tp{j}{M_r}$, we obtain that $j\in N^w_{o,r}$.
	In the following, we compare the probabilities that $j$ and $k$ get $o_i$ in $W_j$ and $W_k$ respectively by cases.
	
	\vspace{0.5em}\noindent{\bf Case \romannumeral2~(a):}
	If $o_i=\tp{k}{M_r}$, i.e., $k\in N^{w'}_{o_i,r}$, then by construction of $W_j$ and $W_k$, $N(l^w_{o_i})\setminus\{j\}=N(l^{w'}_{o_i})\setminus\{k\}$.
	    	It follows that $\lvert N(l^w_{o_i})\rvert =\lvert N(l^{w'}_{o_i})\rvert$ and
    	    \begin{equation*}
            	\begin{split}
            	    & Pr(\{w\in W_j\mid \am^w(R)(j)=o_i\})=Pr(W_j)\cdot\frac{1}{\lvert N(l^w_{o_i})\rvert}= Pr(W_k)\cdot\frac{1}{\lvert N(l^{w'}_{o_i})\rvert}\\
            	    =& Pr(\{w\in W_k\mid \am^w(R)(k)=o_i\}).
            	\end{split}
        	\end{equation*}
	
	\vspace{0.5em}\noindent{\bf Case \romannumeral2~(b):}
	If $o_i\neq\tp{k}{M_r}$, i.e., $k\notin N^{w'}_{o_i,r}$, we discuss in case of $N^{w'}_{o_i,r}$.
	When $N^{w'}_{o_i,r}\neq\emptyset$, then $o_i$ is allocated to some agent in $N^{w'}_{o_i,r}$, which means that $Pr(\{w\in W_j\mid \am^w(R)(j)=o_i\})> Pr(\{w\in W_k\mid \am^w(R)(k)=o_i\})=0$.
	When $N^{w'}_{o_i,r}=\emptyset$,  we have that $N^w_{o_i,r}=\{j\}$.
	It means that $j$ is the only applicant for $o_i$, i.e., $\lvert N(l^w_{o_i})\rvert=1$, and therefore gets it in any $w\in W_j$.
	As for agent $k$, she applies for $o'\neq o_i$ at round $r$ in $w'\in W_k$, and $o_i\succ_j o'$ by the selection of $o_i$.
	It follows that
	\begin{equation}
 	\label{eq:thm:amp:2:2}
	\begin{split}
	    & Pr(\{w\in W_j\mid \am^w(R)(j)=o_i\})= Pr(W_j)\cdot\frac{1}{\lvert N(l^w_{o_i})\rvert}=Pr(W_j)\\
	    >& 1 - Pr(\{w\in W_k\mid \am^w(R)(k)=o'\})\ge Pr(\{w\in W_k\mid \am^w(R)(k)=o_i\}).
	\end{split}
	\end{equation}
	

	
	\vspace{0.5em}\noindent{\bf Concluding the inductive step.} From cases \romannumeral1{} and \romannumeral2, we have covered all the worlds that agent $k$ can possibly get item $o_i$, and we obtain that $p_{k,o_i}\le p_{j,o_i}$.
	With the assumption that $P_{k}\sd{j}P_j$ and condition (3) holds for $i'<i$, it follows that  $p_{k,o_i}\ge p_{j,o_i}$.
	Therefore we have $p_{k,o_i}=p_{j,o_i}$, i.e., condition (3) holds for $i$.
	The equality also requires that
	\begin{itemize}[topsep=0pt,itemsep=0em,leftmargin=0pt,itemindent=1em]
	    \item If in the world $w$, $j\in N_{o_i,r}$ and $k\in N_r$, then $w$ belongs to case \romannumeral1{} where agent $j$ do not get items better than $o_i$ according to $\succ_j$ and agents $k$ do not get items better than $o_i$ according to $\succ_k$.
	    By case \romannumeral1, we have that $k\in N_{o_i,r}$, i.e. condition (1) for rank $i$.
	    Otherwise if $k\notin N_{o_i,r}$, then it fits into case~\romannumeral1~(b), and we have Eq~(\ref{eq:thm:amp:2:1}) which leads to $p_{k,o_i}< p_{j,o_i}$, a contradiction to assumption that $P_{k}\sd{j}P_j$.
	    
	    \item If in the world $w$, agent $j$ gets some $o\succ_j o_i$ at round $r'$, and there exists round $r$ such that $k\in N_r$ with $o_i=\tp{j}{M_r}$, then $w$ belongs to case \romannumeral2{} where agent $j$ gets an item better than $o_i$ according to $\succ_j$ while agents $k$ do not get items better than $o_i$ according to $\succ_k$.
	    By case \romannumeral2~(b), we have that $o_i=\tp{k}{M_r}$, i.e. $k\in N_{o_i,r}$, and it follows that condition (2) holds for rank $i$.
	    Otherwise if $o_i\neq\tp{k}{M_r}$, then it fits into case~\romannumeral2~(b), and we have~Eq~(\ref{eq:thm:amp:2:2}) which leads to $p_{k,o_i}< p_{j,o_i}$, a contradiction to assumption that $P_{k}\sd{j}P_j$.
	\end{itemize}

	With the induction above, we prove that $p_{k,o_i}=p_{j,o_i}$ for any $i$. It follows that if $P_k\succ_j P_j$, $P_k=P_j$.
	
	\vspace{1em}\noindent{\bf Part 3: $\mathbb E(\text{\am}(R))$ is \sdsp{}.}\quad
    We continue to use the new notations introduced at the beginning of Part 2.
	Without loss of generality, let $\succ_j$ be $o_1\succ_j o_2 \succ_j \cdots$.
	Let profile $R'=(\succ'_j,\succ_{-j})$ where $\succ'_j$ is any preference that agent $j$ misreports, and $Q=\am(R')$.
	Assume that $Q_j\sd{} P_j$.
	We show by mathematical induction, for rank $i=1,2,\dots$ with respect to agent $j$, that the following conditions hold:
	\begin{enumerate}[label={\bf Condition (\arabic*):},align=left,wide,labelindent=0pt,itemsep=0.5em,topsep=1em]
		\item when $j\in N^{w}_{o_i,r}$ in a world $w\in W(R)$,
		for any $w'\in W(R')$ where lotteries and winners are the same as $w$ before round $r$, we have that $j\in N^{w'}_{o_i,r}$, and
		\item $p_{j,o_i}=q_{j,o_i}$.
	\end{enumerate}

	\vspace{0.5em}\noindent{\bf Base case.}
	First we show condition (1) for $i=1$.
	It is easy to see that in any $w\in W(R)$, $j$ applies for $o_1$ at round $1$, i.e., $j\in N^{w}_{o_1,1}$.
	We claim that $j\in N^{w'}_{o_1,1}$ for any $w'\in W(R')$.
	Otherwise, if $j\notin N^{w'}_{o_1,1}$ in some $w'$, we show that both of the possible cases below lead to a contradiction to our assumption that $Q_j\sd{} P_j$.
	
	\begin{itemize}[topsep=0pt,itemindent=1em,itemsep=0em,leftmargin=0pt]
	    \item When $N^{w'}_{o_1,1}\neq\emptyset$, $o_1$ is assigned to some agent in $N^{w'}_{o_1,1}$ in $w'$. It follows that $p_{j,o_1}>q_{j,o_1}=0$, a contradiction to the assumption.
	    \item When $N^{w'}_{o_1,1}=\emptyset$, i.e., $N^{w}_{o_1,1}=\{j\}$, $o_1$ is assigned to the only applicant $j$ in $w$, while she applies for item $o'\neq o_1$ in $w'$ and $o_1\succ_j o'$ trivially.
	    It follows that
    	\begin{equation*}
    	\begin{split}
    		p_{j,o_1}=&Pr(\{w\in W(R)\mid \text{\am}^w(R)(j)=o_1\})=1\\
    		>&1-Pr(\{w\in W(R')\mid \text{\am}^w(R')(j)=o'\})\\
    		\ge&Pr(\{w\in W(R')\mid \text{\am}^w(R')(j)=o_1\})=q_{j,o_1},\\
    	\end{split}
    	\end{equation*}
    	a contradiction to the assumption.
	\end{itemize}
	
	In this way, we have $j\in N^{w'}_{o_1,1}$ for any $w'\in W(R')$, i.e., condition (1) for $i=1$, which means that $\lvert N(l^w_{o_1})\rvert=\lvert N(l^{w'}_{o_1})\rvert$ and
	\begin{equation*}
		\begin{split}
			p_{j,o_1}=&Pr(\{w\in W(R)\mid \text{\am}^w(R)(j)=o_1\}=\frac{1}{\lvert N(l^w_{o_1})\rvert}\\
			=&\frac{1}{\lvert N(l^{w'}_{o_1})\rvert}=Pr(\{w\in W(R')\mid \text{\am}^w(R')(j)=o_1\})=q_{j,o_1},\\
		\end{split}
	\end{equation*}
	 i.e., condition (2) for $i=1$.
	
	\vspace{0.5em}\noindent{\bf Inductive step.} Supposing conditions (1) and (2) hold for $i'<i$, we show that they also hold for $i$.
	First we show condition (1) for $i$.
	For an arbitrary world $w^*\in W(R)$ with $j\in N^{w^*}_{o_i,r}$, let $W_1\subseteq W(R)$ and $W_2\subseteq W(R')$ be the sets of worlds where lotteries and winners are the same as $w^*$ before round $r$ with respect to $R$ and $R'$, respectively.
	By construction of $W_1$ and $W_2$, $Pr(W_1)=Pr(W_2)$.
	For any $w\in W_1$ and $w'\in W_2$, $M^w_r=M^{w'}_r$, $N^{w}_r=N^{w'}_r$, and $N^{w}_{o,r}=N^{w'}_{o,r}$ for any $o\in M^w_r\setminus\{o_i\}$.
	We have that $j\in N^{w}_{o_i,r}$ by the selection of $w^*$, and we claim that $j\in N^{w'}_{o_i,r}$ for any $w'\in W(R')$.
	Otherwise, if $j\notin N^{w'}_{o_i,r}$ in some $w'$, we show that both of the possible cases below lead to a contradiction to our assumption that $Q_j\sd{} P_j$.
	
	\begin{itemize}[topsep=0pt,itemindent=1em,itemsep=0em,leftmargin=0pt]
	    \item 	When $N^{w'}_{o_i,r}\neq\emptyset$, item $o_i$ is assigned to some agent in $N^{w'}_{o_i,r}$ in $w'$. It follows that
	\begin{equation}\label{eq:thm:amp:1}
		\begin{split}
			&Pr(\{w\in W_1\mid \text{\am}^w(R)(j)=o_i\})=Pr(W_1)\\
			>&Pr(\{w\in W_2\mid \text{\am}^w(R')(j)=o_i\})=0.
		\end{split}
	\end{equation}
	With condition (1) for $i'<i$, in world $w'$, agent $j$ can only apply for $o_i$ at round $r'\ge r$ not earlier than she does in $w$, which means that $p_{j,o_i}>q_{j,o_i}=0$ with Eq~(\ref{eq:thm:amp:1}).
	Together with condition (2) for $i'<i$, we have a contradiction to the assumption that $Q_j\sd{} P_j$.
	\item 	When $N^{w'}_{o_i,r}=\emptyset$, i.e., $N^{w}_{o_i,r}=\{j\}$, item $o_i$ is assigned to the only applicant $j$ in $w$ while she applies for item $o'\neq o_i$ in $w'$ and $o_i\succ_j o'$ by the selection.
	It follows that
	\begin{equation*}
		\begin{split}
			&Pr(\{w\in W_1\mid \text{\am}^w(R)(j)=o_i\})=Pr(W_1)\\
			>&Pr(W_2)-Pr(\{w\in W_2\mid \text{\am}^w(R')(j)=o'\})
			\ge Pr(\{w\in W_2\mid \text{\am}^w(R')(j)=o_i\}).
		\end{split}
	\end{equation*}
	This means that $p_{j,o_i}>q_{j,o_i}$, a contradiction to the assumption that condition (2) holds for $i'<i$.
	\end{itemize}

	In this way, we have condition (1) for $i$, which means that $\lvert N(l^w_{o_i})\rvert=\lvert N(l^{w'}_{o_i})\rvert$ and
	\begin{equation*}
		\begin{split}
			&Pr(\{w\in W_1\mid \text{\am}^w(R)(j)=o_i\}=Pr(W_1)\cdot\frac{1}{\lvert N(l^w_{o_i})\rvert}\\
			=&Pr(W_2)\cdot\frac{1}{\lvert N(l^{w'}_{o_i})\rvert}=Pr(\{w\in W_2\mid \text{\am}^w(R')(j)=o_1\}),\\
		\end{split}
	\end{equation*}
	which implies $p_{j,o_i}=q_{j,o_i}$, i.e., condition (2) for $i$.
	
	By mathematical induction, we have that $p_{j,o}=q_{j,o}$ for any $o$, and therefore if $Q_j\sd{} P_j$, that $Q_j=P_j$ .
	
	\vspace{1em}\noindent{\bf Part 4: $\mathbb E(\text{\am}(R))$ is \etep{}.}\quad
	
	For agents $j$ and $k$, we prove that $p_{j,o}=p_{k,o}$ for any item $o$ appearing in $\succ_{j,k}$.
	We compare probability of possible worlds where agent $j$ gets $o\in\ucs(\succ_{j,k},o_m)$ with those where agent $k$ gets $o$.
	
	First we consider the world $w$ where $j$ gets $o$ at round $r$ and $k$ gets $o'\in\ucs(\succ_{j,k},o_m)$ at round $r'$.
	Let $w'$ satisfy $j$ gets $o'$ at round $r$, $k$ gets $o$ at round $r'$, and the result of other lotteries keep the same as $w$.
	In $w'$, we see that $k$ wins the lottery for $o$ instead of $j$, and $j$ participates in lotteries at rounds $r+1$ to $r'$ instead of $k$.
	We also see that for every lottery $l_o$, $\lvert N(l_o)\rvert$ keep the same in worlds $w$ and $w'$.
	Therefore we have that $Pr(w)= Pr(w')$.
	
	Then we consider the world $w$ where $j$ gets $o$ at round $r$ and $k$ does not get items appearing in $\succ_{j,k}$.
	Let $o'$ be the last item $k$ applies for in $\succ_{j,k}$ at round $r'$, and $W_j$ be the set of worlds which are the same as $w$ from rounds $1$ to $r'$.
	Here the probability of $W_j$ can also be computed as $Pr(W_j)=\Pi_{l\in L(r^*),r^*\le r'}\frac{1}{\lvert N(l)\rvert}$.
	We construct another set $W_k$ such that for any $w\in W_j$,
	\begin{enumerate*}[label=(\roman*)]
		\item the winners of lotteries are the same as $w$ at round $1$ to $r-1$,
		\item the winner of $l(o)$ is $k$ at round $r$, and any other $l\in L(r)$ is the same as $w$,
		\item $j$ participates in lotteries at rounds $r+1$ to $r'$ instead of $k$.
	\end{enumerate*}
	Then we see that for every lottery $l\in L(r^*)$ with $r^*\le r'$, $\lvert N(l)\rvert$ are the same in any world $w\in W_j$ and $w'\in W_k$.
	Therefore we have that $Pr(W_j)= Pr(W_k)$.
	
	Together we have that $p_{j,o}= p_{k,o}$ for any $o$ appearing in $\succ_{j,k}$.
\end{proof}
\fi

\thmabmchar*
\begin{proof}
    We provide the satisfaction part for the sake of completeness.

    {\bf\noindent(Satisfaction)}
    To show \abm{} satisfies \efcr{}, we first prove by mathematical induction that given a profile $R$, $A=\text{\abm}^{\impord{}}(R)$  satisfies \fhcr{} for any $\impord{}$.

    \vspace{0.5em}\noindent{\bf Base case.}
    At round $1$ of \abm$^{\impord{}}$, every agent applies for their top ranked items.
    For any item $o$, let $N^1_{o}=\{j:o=\tp{j}{M}\}$, i.e., the set of agents who rank $o$ highest.
    We also note that $o$ with $N^1_{o}\neq\emptyset$ satisfies $o\in \tps{A}{1}\subseteq M$ trivially.
    According to the priority order $\impord{}$, item $o$ is assigned to the agent with highest priority among $N^1_{o}$.
    It follows that $o=\tp{A^{-1}(o)}{M}$, which meets the requirement of \fhcr{} when $r=1$.
    
    \vspace{0.5em}\noindent{\bf Induction Step.}
    Assume that $A$ meets the requirement of \fhcr{} for any $r'<r$,
    For round $r$, let $M^r$ and $N^r$ be the set of available items and unsatisfied agents at the beginning of that round, respectively.
    By construction, it follows that $M^r=M\setminus\bigcup_{r'<r}\tps{A}{r'}$, and that for any $j\in N^r$, $A(j)\notin\bigcup_{r'<r}\tps{A}{r'}$.
    In~\Cref{alg:abm}, every agent in $N^r$ applies for their top ranked items among $M^r$.
    For any item $o$, let $N^r_{o}=\{j:o=\tp{j}{M^r}\}$, i.e., the set of agents who rank $o$ highest among $M^r$.
    We note that $o$ with $N^r_{o}\neq\emptyset$ satisfies $o\in \tps{A}{r}$ by the construction.
    Then $o$ is assigned to the agent ranked highest in $\impord{}$ among $N^r_{o}$.
    It follows that $o=\tp{A^{-1}(o)}{M^r}$, and we see that it meets the requirement of \fhcr{} for $r$.
    
    By induction, we have that the outcome of ABM$^{\impord{}}(R)$ satisfies \fhcr{} for any profile $R$, and therefore ABM$^{\pi}(R)=\sum \pi(\impord{})*\text{ABM}^{\impord{}}(R)$ satisfies \efcr{} by definition.
\end{proof}

    To prove~\Cref{thm:uprep,thm:familychar}, we show the following claim (where $M_{P,r}$ and $E_{P,r}(o)$ are defined in \Cref{dfn:sdcrf}).

	\begin{lemma}\label{lem:pre}
		Given any preference profile $R$ and any member $f$ of \pcr{}, let $P=f(R)$. Then, for any round $r$, it holds that:
		\begin{enumerate}[label=\rm(\roman*),wide,labelindent=0pt,topsep=0em,itemsep=0pt]
			\item $M_{P,r}=M'$ and $E_{P,r}(o)=N_o$ for each $o\in M_{P,r}$, where $M'$ is the set of items with remaining supply, and $N_o$ is the set of agents who are eager for item $o$ at the beginning of round $r$ during the execution of~\Cref{alg:pcr}.
			\item for any agent $j$ and item $o^*$ with $\tp{j}{M_{P,r-1}}\succ_j o^*\succ_j \tp{j}{M_{P,r}}$, it holds that $p_{j,o^*}=0$.
			\item for any round $r^*>r$, and any item $o\in M_{P,r^*}$, it holds that for any $j\in E_{P,r}(o)$,
		\end{enumerate}
		\begin{equation}\label{lem:pre:1}
			\sum_{o'\in\ucs(\succ_j,o)}p_{j,o'}=1.
		\end{equation}
	\end{lemma}
	
	\begin{proof}
	We prove the claim by mathematical induction for every $r$.

    \vspace{0.5em}\noindent{\bf Base case.} When $r=1$, we see that $s(o)$ is initially set to $1$ with respect to the supply of item and $M_{P,1}=M$ which is also the initial value of $M'$ on Line~3 of~\Cref{alg:pcr}.
	Therefore $E_{P,1}(o)=\{j\mid \tp{j}{M'}\}=N_o$ by Line~4 for round $1$.
	Together we have (\romannumeral1) for $r=1$.

	Besides, since no such $o^*\succ_j o=\tp{j}{M_{P,1}}$ exists for any agent $j$, we have that (\romannumeral2) holds for $r=1$ trivially.
	
	Since $t_j$ is set to $0$ for any $j\in N$ and $\sum_{k\in N_{o}}\int^{1}_{0}\omega_{k}(t){\rm d}t\ge s(o)=1$, $o$ is consumed to exhaustion by agents in $N_o$ at round $1$.
	It also means that $\sum_{k\in E_{P,1}(o)}p_{k,o}=1$, $o\notin M_{P,r^*}$ with $r^*>1$, and therefore (\romannumeral3) holds trivially.
	
	\vspace{0.5em}\noindent{\bf Inductive step.} Supposing that (\romannumeral1)-(\romannumeral3) holds for any $r'<r$, we show that it also holds for $r$.
	In~\Cref{alg:pcr}, at the beginning of round $r$, $M'$ only contains item $o$ with positive supply, i.e., $s(o)>0$, after consumption of previous rounds by Line~8.
	Because (\romannumeral1) holds for $r'<r$, we have that only agents in $\bigcup_{r'<r}E_{P,r'}(o)$ are able to consume $o$ before round $r$, which means that $s(o)=1-\sum_{j\in\bigcup_{r'<r}E_{P,r'}(o)}p_{j,o}$, and therefore $M'=M_{P,r}$.
	Then we have that $N_o=\{j\mid o=\tp{j}{M'}\}=E_{P,r}(o)$.
	Together we have that (\romannumeral1) holds for $r$.
	
	Then we show that (\romannumeral2) holds for $r$.
	For any agent $j$ and item $o^*$ such that $\tp{j}{M_{P,r-1}}\succ_j o^*\succ_j \tp{j}{M_{P,r}}$, it means that  $o^*\notin M_{P,r}$ and agent $j$ cannot get shares of $o^*$ at round $r-1,r$ or later rounds.
	\begin{itemize}[topsep=0pt,itemindent=1em,itemsep=0em,leftmargin=0pt]
	    \item If $\tp{j}{M_{P,r-1}}=\tp{j}{M_{P,r}}$, we have that (\romannumeral2) is trivially true.
	    \item If $\tp{j}{M_{P,r-1}}\neq\tp{j}{M_{P,r}}$ and $o^*= \tp{j}{M_{P,r'}}$ for some $r'<r-1$, we know that $j$ consumes $\tp{j}{M_{P,r-1}}$ at round $r-1$, and $o^*$ at round $r'$ because (\romannumeral1) holds for $r'$, which means that $\tp{j}{M_{P,r-1}}$ and $o^*$ are available at round $r'$.
	    It follows that both items $o^*,\tp{j}{M_{P,r-1}}\in M_{P,r'}$, a contradiction to $\tp{j}{M_{P,r-1}}\succ_j o^*$.
	    Therefore $o^*\neq\tp{j}{M_{P,r'}}$ for any $r'<r-1$, i.e., agent $j$ does not get shares of $o^*$ before round $r-1$.
	\end{itemize}
    Together we have that $p_{j,o^*}=0$, i.e., (\romannumeral2) holds for $r$.

	Finally, we show (\romannumeral3) holds for $r$.
	For any $o\in M'$, if $o\in M_{P,r^*}$ for some $r^*>r$, then
	\begin{equation}\label{lem:pre:2}
		\sum_{k\in\bigcup_{r''\le r}E_{P,r''}(o)}p_{k,o}<\sum_{k\in\bigcup_{r''\ge r^*}E_{P,r''}(o)}p_{k,o}\le1.
	\end{equation}
	Eq~(\ref{lem:pre:2}) implies that $o$ is still available after the consumption at round $r$.
	Wiht $\omega_{j}(t)=0$ when $t>1$ for any $j\in N$, it follows that 
	\begin{equation}\label{{lem:pre:3}}
	    p_{j,o}=
	    \begin{cases}
	        \int^{1}_{t_j}\omega_{j}(t){\rm d}t, & t_j<1\\
	        0, &  t_j\ge1.\\
	    \end{cases}
	\end{equation}
	
	With (\romannumeral1) for $r'$, we know that for any $o'=\tp{j}{M_{P,r'}}$ with $j\in N_{o}$ and $r'<r$, agent $j\in E_{P,r'}(o')$, which means that $o'$ is consumed by $j$ if agent $j$ is not satisfied at round $r'$.
	Moreover,  $o'\in\ucs(\succ_j,o)$ due to $M_{P,r}\subseteq M_{P,r'}$, and we have that items consumed by $j$ in time period $[0, t_j]$ are not worse then $o$ according to $\succ_j$.
	Then with Eq~(\ref{{lem:pre:3}}) and (\romannumeral2) for $r'\le r$, $\sum_{\hat{o}\in\ucs(\succ_j,o)}p_{j,\hat{o}}= \int^{1}_{0}\omega_{j}(t){\rm d}t$=1, i.e., Eq~(\ref{lem:pre:1}), which complete the proof of (\romannumeral3) for current $r$.
	\end{proof}

\thmuprep*

\begin{proof}
	Given an instance with $R$, let $P=\text{\upre}(R)$.

	\vspace{0.5em}\noindent{\bf Part 1: $\upre{}(R)$ is \sdwef{}.}\quad
	
	Assume that there exist agents $j$ and ${j'}$ such that $P_{j'} \sd{j} P_j$.
	Without loss of generality, let $o_r$ be the item such that $j\in N_{o_r}$ at round $r$, and we have the following claim:
	\begin{claim}\label{clm:thm:pcrp:1}
		For $r'<r$, if $o_r\neq o_{r'}$, then $o_r\succ_j o_{r'}$.
	\end{claim}
 	The claim holds because by~\Cref{lem:pre}~(\romannumeral1), $o_r=\tp{j}{M_{P,r}}$ and $o_{r'}\in M_{P,r'}\subseteq M_{P,r}$.

	We prove by mathematical induction that the following conditions hold for any round $r$ with $t_j<1$:
	\begin{enumerate}[label={\bf Condition (\arabic*):},align=left,wide,labelindent=0pt,itemsep=0pt,topsep=0pt]
		\item $t_{j'}=t_j$,
		\item $p_{j,o'}=p_{j',o'}=0$ for any $o'$ with $o_{r-1}\succ_j o'\succ_j o_r$,
		\item $j'\in N_{o_r}$, and
		\item $p_{j,o_r}=p_{j',o_r}$.
	\end{enumerate}

	\vspace{0.5em}\noindent{\bf Base case.}
	With $j\in N_{o_1}$ at round $1$, we know that $t_k=0$ for every $k$, i.e., condition (1) holds for $r=1$, and $o_1=\tp{j}{M'}=\tp{j}{M}$ by~\Cref{lem:pre}~(\romannumeral1).
	
	Condition (2) is trivially true since no item $o'\succ_j o_1$ exists.
	
	Then, we show that condition (3) holds for $r=1$.
	Item $o_1$ is consumed to exhaustion at this round by Line~6.1 because $\sum_{k\in N_{o_1}}\int^{t_k+1}_{t_k}\omega_{k}(t){\rm d}t\ge s(o_1)=1$, which means that no agent can get $o_1$ at any round $r^*>1$.
	Therefore $p_{j,o_1}>0$ due to $j\in N_{o_1}$.
	If $j'\notin N_{o_1}$, then $j'$ does no consume $o_1$, which means that $p_{j',o_1}=0<p_{j,o_1}$, a contradiction to the assumption that $P_{j'} \sd{j} P_j$.
	Then we have that $j'\in N_{o_1}$, i.e., condition (3) holds for $r=1$.
	
	Because $t_j=t_{j'}=0$, $p_{j,o_1}=\int_{0}^{\rho_{o_1}}\omega_j(r){\rm d}t=\int_{0}^{\rho_{o_1}}\omega_{j'}(r){\rm d}t=p_{j',o_1}$ by Eq~(\ref{eq:eatingfunction}), i.e., condition (4) for $r=1$.
	
	\vspace{0.5em}\noindent{\bf Inductive step.}
	Supposing that conditions (1)-(4) hold for any $r'<r$,
	we show that they also hold for $r$ with $t_j<1$.
	
	We see that condition (1) trivially holds for $r$, i.e., $t_{j'}=t_j$ due to the fact that by condition (3) for any $r'<r$, both $t_{j'}$ and $t_j$ increase by the same value $\rho_{o_{r'}}$ on Line~6.1 in each round $r'$.
	
	Then, we prove condition (2) for $r$.
	Here $o_r\neq o_{r-1}$, because otherwise we know that after the consumption at round $r-1$, $o_{r-1}$ is not exhausted, which means that agent $j\in N_{o_{r-1}}$ is satisfied.
	It follows that $t_j\ge 1$ at the beginning of round $r$, which we do not need to consider.
	
	For any $o'$ with $o_{r-1} \succ_j o'\succ_j o_r$, $p_{j,o'}=0$ by~\Cref{lem:pre}~(\romannumeral1) and~(\romannumeral2), and we show that $p_{j',o'}=0$.
	We have that $o'\notin M'$ at round $r$ because $j\in N_{o_r}$, i.e., $\tp{j}{M'}=o_r$, which means that $o'$ is unavailable for round $r^*\ge r$.
	With~\Cref{clm:thm:pcrp:1}, $o'$ is not consumed by $j$ at any round $r'<r$, and therefore it is also not consumed by $j'$ according to condition (3) for $r'<r$.
	Then $o'$ is never consumed by $j$ or $j'$, which means that $p_{j,o'}=p_{j',o'}=0$, i.e., condition (2) for $r$.
	
	Next, we prove condition (3) for $r$.
	We consider the following cases.
	\begin{itemize}[topsep=0pt,itemindent=1em,itemsep=0em,leftmargin=0pt]
	    \item If $o_r\in M_{P,r+1}$, then by~\Cref{clm:thm:pcrp:1} and condition (2) for $r'\le r$, it must hold that $\sum_{o^*\succ_j o_r}p_{j,o^*}\allowbreak=\int_{0}^{t_j}\omega_j(r){\rm d}t=t_j<1$.
	    By~\Cref{lem:pre}~(\romannumeral1), $j\in N_{o_r}=E_{P,r}(o)$.
	    By~\Cref{lem:pre}~(\romannumeral3) and $o_r\in M_{P,r+1}$, $\sum_{o^*\in\ucs(\succ_j,o_r)}p_{j,o^*}=1$.
	    By condition (1) for $r$ that $t_{j'}=t_{j}<1$, $j'$ is not satisfied at the beginning of round $r$.
	If $j'\in N_{o'}$ with $o'\neq o_r$, it means that $j'$ consumes $o'$ at round $r$ and $p_{j',o'}>0$, and $o_r\succ_j o'$ since $o_r=\tp{j}{M'}$.
	By conditions (2) and (4) for $r'<r$, and condition (2) for $r$ which we just prove, we have that
	\begin{equation}\label{eq:thm:pcrp:1}
		\sum_{o^*\succ_j o_r}p_{j,o^*}=\sum_{o^*\succ_j o_r}p_{j',o^*}.
	\end{equation}
	Therefore,
	\begin{equation*}
	       \sum_{o^*\in\ucs(\succ_j,o_r)}p_{j',o^*}<1-p_{j',o'}<1=\sum_{o^*\in\ucs(\succ_j,o_r)}p_{j,o^*},
	\end{equation*}
	 a contradiction to the assumption that $P_{j'} \sd{j} P_j$.
	    \item
	If $o_r\notin M_{P,r+1}$, then $o_r$ is consumed to exhaustion by agents in $N_{o_r}$ at round $r$, and no agent consumes $o_r$ after round $r$.
	Since $j\in N_{o_r}$ and $t_j<1$, $p_{j,o_r}>0$.
	If $j'\notin N_{o_r}$, then $j'$ does not consume $o_r$ at round $r^*\ge r$.
	Moreover, $j'$ also does not consume $o_r$ before round $r$ by condition (3) for $r'<r$, which means that $p_{j',o_r}=0$.
	With Eq~(\ref{eq:thm:pcrp:1}), we have that $\sum_{o^*\in\ucs(\succ_j,o_r)}p_{j,o^*}>\sum_{o^*\in\ucs(\succ_j,o_r)}p_{j',o^*}$, a contradiction to the assumption that $P_{j'} \sd{j} P_j$.
	\end{itemize}
	Together we show that $j\in N_{o_r}$.
	
	Finally condition (4) holds for $r$ trivially because by Eq~(\ref{eq:eatingfunction}), and conditions (1) and (3) for $r$, we have that
	\begin{equation*}
	    	p_{j,o_r}=\int_{t_j}^{t_j+\rho_{o_r}}\omega_j(r){\rm d}t=\int_{t_{j'}}^{t_{j'}+\rho_{o_r}}\omega_{j'}(r){\rm d}t\allowbreak=p_{j',o_r}.
	\end{equation*}

	By the induction, we have conditions (2) and (4) for any $r$, i.e., $p_{j,o}=p_{k,o}$ for any item $o$, which means that $P_{j'}=P_j$ if $P_{j'} \sd{j} P_j$.

	\vspace{0.5em}\noindent{\bf Part 2: $\upre{}(R)$ is \etep{}.}\quad
	
	We show that before consuming items not in $\succ_{j,k}$, $t_j=t_k$, and $j$ and $k$ consume the same item in each round by mathematical induction based on rounds.
	
	\vspace{0.5em}\noindent{\bf Base case.} At round $1$, we know that both $j$ and $k$ consume the most preferred item $o$ in $\succ_{j,k}$.
	We have $p_{j,o}=p_{k,o}$ by Eq~(\ref{eq:eatingfunction}) and $t_j=t_k=0$ which is set initially at the beginning of~\Cref{alg:pcr}.
	
	\vspace{0.5em}\noindent{\bf Inductive step.} Supposing that $j$ and $k$ consume the same item and get the same shares for each round $r'<r$, we prove that this is also the true for round $r$.
	By the inductive assumption, we trivially have that $t_j=t_k$ at the beginning of $r$.
	Let $j$ consume $o$, and $k$ consume $o'$.
	Here we do not need to consider the case that both $o,o'$ not in $\succ_{j,k}$.
	If $o\neq o'$, we assume that $o\succ o'$ without loss of generality, and therefore $o$ must be in $\succ_{j,k}$.
	It means that $o$ is available at round $r$, but $k$ consumes $o'$, a contradiction to the selection of top items.
	Therefore $o=o'$, and  $p_{j,o}=p_{k,o}$ by $t_j=t_k$ and Eq~(\ref{eq:eatingfunction}).
	
	By induction we have that $p_{j,o}=p_{k,o}$ for every $o$ in $\succ_{j,k}$.
\end{proof}

    \Cref{lem:sdcrf} below illustrates how shares of items must be allocated in order to satisfy \sdcrf{}, which is used in proving the uniqueness part of~\Cref{thm:familychar} and impossibility results for \sdcrf{}.

    
    \begin{lemma}\label{lem:sdcrf}
    Given $P$ satisfying \sdcrf{}, for any $r$ and $o\in M_{P,r}$, we define the remaining shares of item $o$ excluding those owned by agents in $E_{P,r'}(o)$ with any $r'<r$, 
		\begin{equation*}
		    \begin{split}
		        s_{P,r}(o)=1-\sum_{k\in\bigcup_{r'<r}E_{P,r'}(o)}p_{k,o},
		    \end{split}
		\end{equation*}
		and the remaining demand of agent $j$ for items ranked below the item $\tp{j}{M_{P,r-1}}$,
		\begin{equation*}
		    \begin{split}
		        d_{P,r}(j)=1-\sum_{o'\in\ucs(\succ_j,\tp{j}{M_{P,r-1}})}p_{j,o'}.
		    \end{split}
		\end{equation*}
		For any $j\in E_{P,r}(o)\neq\emptyset$,
		\begin{enumerate}[label=\rm(\roman*),wide,labelindent=0pt,topsep=0em,itemsep=0pt]
			\item for any $o^*$ with $\tp{j}{M_{P,r-1}}\succ_j o^* \succ_j o$, it holds that $p_{j,o^*}=0$.
			\item if the total remaining demand of agents eager for item $o$ surpasses  the remaining shares of $o$, i.e. $\sum_{k\in E_{P,r}(o)}d_{P,r}(k)\ge s_{P,r}(o)$, then $o\notin  M_{P,r^*}$ for any $r^*>r$ and remaining shares of $o$ are allocated to these agents, i.e., $\sum_{k\in E_{P,r}(o)}p_{k,o}=s_{P,r}(o)$.
			\item if the total remaining demand of agents eager for item $o$ does not surpass the remaining shares of $o$, i.e., $\sum_{k\in E_{P,r}(o)}d_{P,r}(k)\le s_{P,r}(o)$, then these agents' demands are satisfied by shares of $o$, i.e., $p_{j,o}=d_{P,r}(j)$.
		\end{enumerate}
    \end{lemma}
    
    \begin{proof}
	(\romannumeral1) By the condition, we have that $o^*\notin M_{P,r}$ , i.e., $\sum_{k\in\bigcup_{r'<r}E_{P,r'}(o)}p_{k,o}=1$, which means that $j\notin E_{P,^*}(o^*)$ with $^*\ge r$.
	We also have that $j\notin E_{P,r'}(o^*)$ for any $r'<r$, since $\tp{j}{M_{P,r'}}\succ_j o^*$.
	Together we have that $j$ does not eager for $o^*$, which leads to $p_{j,o^*}=0$.
	
	(\romannumeral2) Assume that $o\in  M_{P,r^*}$ with $r^*>r$.
	By (\romannumeral1) for $r$ and $\sum_{k\in E_{P,r}(o)}d_{P,r}(k)\ge s_{P,r}(o)$, there exists agent $j'\in E_{P,r}(o)$ such that $\sum_{o'\in\ucs(\succ_{j'},o)}p_{j',o'}=d_{P,r}(j')+p_{j',o}<1$, a contradiction to $P$ satisfying \sdcrf{}.
	
	(\romannumeral3) Assume that there exists $j'\in E_{P,r}(o)$ with $p_{j',o}<d_{P,r}(j')$.
	Then by $s_{P,r+1}(o)>s_{P,r}(o)-\sum_{k\in E_{P,r}(o)}d_{P,r}(k)\ge0$, we have that $\sum_{o'\in\ucs(\succ_{j'},o)}p_{j',o'}<1$ and $o\in M_{P,r+1}$, a contradiction.
\end{proof}

\thmfamilychar*
\begin{proof}
	{\bf(Satisfaction)}
    Let $P=\pcr{}_{\omega}(R)$ where $\omega$ is any collection of eating functions.
	By~\Cref{lem:pre} (\romannumeral3), we have that for any item $o\in M_{P,r}$, $\sum_{o'\in\ucs(\succ_j,o)}p_{j,o'}=1$ for any $j\in E_{P,r'}(o)$ with $r'<r$, which means that $P$ satisfies \sdcrf{}.
	
	\vspace{0.5em}\noindent{\bf (Uniqueness)}
	Given $Q$ satisfying \sdcrf{}, we prove that it coincides with the outcome $P=\pcr{}_{\omega}(R)$ where the eating functions in $\omega$ are as defined in  Eq~(\ref{eq:thm:familychar:2.5}) for each agent $j$:
	\begin{equation}\label{eq:thm:familychar:2.5}
		\omega_j(t)=
		\begin{cases}
			
			n~\cdot~q_{j,o}, & {\mbox{$t\in[\frac{r-1}{n},\frac{r}{n}],$ where $r=\min(\{\hat{r}\mid j\in E_{Q,\hat{r}}(o)\})$, }}\\
			
			
			0, & \mbox{others.}
			
		\end{cases}
	\end{equation}
	We prove by mathematical induction that the following conditions hold for any round $r$:
	\begin{enumerate}[label={\bf Condition (\arabic*):},align=left,wide,labelindent=0pt,itemsep=0pt,topsep=0pt]
		\item $M_{P,r}=M_{Q,r}$, and $E_{P,r}(o)=E_{Q,r}(o)$ for each $o\in M_{Q,r}$.
		\item for any $j\in E_{Q,r}(o)$, if agent $j$ is not satisfied by items in $\ucs(\succ_j,\allowbreak\tp{j}{M_{Q,r-1}}$, i.e., $\sum_{o'\in\ucs(\succ_j,\tp{j}{M_{Q,r-1}})}p_{j,o'}<1$, the start time $t_j=(r-1)/n$ and the consumption time $\rho_o=1/n$, and
		\item for any $j\in E_{Q,r}(o)$ and $o\in M_{Q,r}$, $p_{j,o}=q_{j,o}$.
	\end{enumerate}
	
	\vspace{0.5em}\noindent{\bf Base case.} When $r=1$, we trivially have that $M_{Q,1}=M'=M$ and $E_{Q,1}(o)=N_o$ for any $o\in M$ at round $1$ in~\Cref{alg:pcr}, and each $j\in N_o$ consumes $o$.
	With~\Cref{lem:pre}~(\romannumeral1), we have that condition (1) holds for $r=1$.
	
	Then we show condition (2) holds for $r=1$.
	By Line~2, $s(o)$, the supply of $o$, is set to $1$ for any $o\in M'$, and for any $j\in E_{Q,r}(o)$, $t_j$ is set to $0=(r-1)/n$.
	Since $\sum_{k\in N_{o}}\int^{1}_{t_k}\omega_{k}(t){\rm d}t\ge s(o)=1$, $\rho_o=\min\{\rho\mid\sum_{k\in N_{o}}\int^{\rho}_{0}\omega_{k}(t){\rm d}t= s(o)\}$ by Line~6.1, and $o$ is consumed to exhaustion.
	We also have that $\sum_{k\in N_o}p_{k,o}=1$ for $P$.
	Otherwise, $o\in M_{P,2}$, and there exists $j'\in E_{P,1}(o)$ with $\sum_{o'\in\ucs(\succ_j,o)}p_{j',o'}=p_{j',o}<1$, a contradiction to $P$ satisfying \sdcrf{} by the satisfaction part above.
	Similarly, with $Q$ satisfying \sdcrf{}, $\sum_{k\in N_o}q_{j,o}=1=\sum_{k\in N_o}p_{k,o}$.
	Therefore by Eq~(\ref{eq:thm:familychar:2.5}), $\sum_{k\in N_{o}}\int^{1/n}_{0}\omega_{k}(t){\rm d}t=1$, which means that $\rho_o$, the time for consuming $o$, is exactly $1/n$.
	Together we have that condition (2) holds for $r=1$.
	
	With $\rho_o=1/n$, $p_{j,o}=\int^{1/n}_{0}\omega_{j}(t){\rm d}t=q_{j,o}$, i.e., condition (3) holds for $r=1$.
	
	\vspace{0.5em}\noindent{\bf Inductive step.}
	Supposing that conditions (1)-(3) hold for $r'<r$, we show they also hold for $r$.
	First, since conditions (1) and (3) holds for $r'<r$, we trivially have that condition (1) holds for $r$.

	Next we show condition (2) holds for $r$.
	By~\Cref{lem:pre}~(\romannumeral1) and condition (1) for $r$ which we just prove, it holds that $M_{Q,r}=M'$ and $E_{Q,r}(o)=N_o$ for each $o\in M_{Q,r}$.
	By Line~8 and  condition (2) for $r-1$, we have that for any $j\in E_{Q,r}(o)$ with $\sum_{o'\in\ucs(\succ_j,\tp{j}{M_{Q,r-1}})}p_{j,o'}<1$, $t_j=(r-1)/n$. 
	
	Then we show that the consumption time $\rho_o=1/n$.
	Let $N'_o=\{k\in N_o\mid \sum_{o'\in\ucs(\succ_k,\tp{j}{M_{Q,r-1}})}p_{k,o'}<1\}$, and we see that agent $j'\in N_o\setminus N'_o$ does not consume $o$ at round $r$ since they have been satisfied with $\int_0^{t_j}\omega_j(t){\rm d}t=\sum_{o'\in\ucs(\succ_{j'},\tp{j}{M_{Q,r-1}})}p_{j',o'}\ge1$.
	Besides, we define $M^*=\{o^*\mid k\in E_{Q,r'}(o^*)\text{ with }r'<r\}=\{o^*\mid k\in E_{P,r'}(o^*)\text{ with }r'<r\}$ by conditions (1) and (3) for $r'<r$.
	We show that $\rho_o=1/n$ in both cases about $s(o)$, now the remaining shares of $o$ at round $r$ in~\Cref{alg:pcr}.
	We also note that with~\Cref{lem:pre}~(\romannumeral1), $s(o)=\sum_{k\notin\bigcup_{r'<r}E_{P,r'}(o)}p_{k,o}=\sum_{k\notin\bigcup_{r'<r}E_{Q,r'}(o)}q_{k,o}$, which means that $s(o)$ is also the total shares of $o$ owned by agents not in $E_{Q,r'}(o)$ with $r'<r$ in assignment $Q$.
	
	\begin{itemize}[topsep=0pt,itemindent=1em,itemsep=0em,leftmargin=0pt]
	    \item 	If $\sum_{k\in N_{o}}\int^{1}_{t_k}\omega_{k}(t){\rm d}t\ge s(o)$, then $o$ is consumed to exhaustion, i.e. $\sum_{k\in N'_o}p_{k,o}\allowbreak =s(o)$.
    	Otherwise, assume for the sake of contradiction that $\sum_{k\in N'_o}p_{k,o}<s(o)$.
    	Then by the assumption, $o\in M_{P,r+1}$, and there exists $j'\in E_{P,r}(o)$ with $p_{j',o}<\int^{1}_{t_j}\omega_{j}(t){\rm d}t$, which means that $\sum_{o'\in\ucs(\succ_{j'},o)}p_{j',o'}=\sum_{o'\in M^*}p_{j',o} + p_{j',o}<\int^{1}_{0}\omega_{j}(t){\rm d}t=1$ by~\Cref{lem:sdcrf}~(\romannumeral1), a contradiction to $P$ satisfying \sdcrf{}.
    	Since $Q$ also satisfies \sdcrf{}, we have that $o\notin M_{Q,r+1}$, and with conditions (1) and (3) for $r'<r$, we have that $\sum_{k\in N'_o}q_{k,o}=s(o)=\sum_{k\in N'_o}p_{k,o}$.
    	We also note that $r=\min(\{\hat{r}\mid j\in E_{Q,\hat{r}}(o)\})$ for any $j\in N'_o$, because otherwise $j\in E_{Q,r'}(o)=E_{P,r'}(o)$ with $r'<r$ and $o\in M_{P,r}$, while $\sum_{o'\in\ucs(\succ_j,o)}p_{j,o'}\le\sum_{o'\in\ucs(\succ_j,\tp{j}{M_{Q,r-1}})}p_{j,o'}<1$, a contradiction to $P$ satisfying \sdcrf{}.
    	Then by Eq~(\ref{eq:thm:familychar:2.5}) and~Line~6.1 of~\Cref{alg:pcr},
        \begin{equation*}
          \rho_o=\min\{\rho\mid\sum_{k\in N_{o}}\int^{t_k+\rho}_{t_k}\omega_{k}(t){\rm d}t= s(o)\}=1/n.
        \end{equation*}
	    \item 	If $\sum_{k\in N_{o}}\int^{1}_{t_k}\omega_{k}(t){\rm d}t< s(o)$, then $o\in M_{P,r+1}$, and we have that all the agents in $N'_o$ are satisfied, i.e., $ \sum_{k\in N'_o}p_{k,o}=\sum_{k\in N_{o}}\int^{t_k+\rho_o}_{t_k}\omega_{k}(t){\rm d}t=\sum_{k\in N_{o}}\int^{1}_{t_k}\omega_{k}(t){\rm d}t$.
	    Otherwise, there exists $j'\in E_{P,r}(o)$ who is not satisfied with $\sum_{o'\in\ucs(\succ_{j'},o)}p_{j',o'}=\sum_{o'\in M^*}p_{j',o} + p_{j',o}<1$ by~\Cref{lem:sdcrf}~(\romannumeral1), a contradiction to $P$ satisfying \sdcrf{}.
	    With conditions (1) and (3) for $r'<r$, we also have that 
	    \begin{equation}\label{eq:thm:familychar:3}
	        \sum_{k\in N'_o}q_{k,o}\le\sum_{k\in N'_o}(1-\sum_{o'\in M^*}q_{k,o'})=\sum_{k\in N'_o}(1-\sum_{o'\in M^*}p_{k,o'})=\sum_{k\in N'_o}\int^{1}_{t_k}\omega_{k}(t){\rm d}t=\sum_{k\in N'_o}p_{k,o}.
	    \end{equation}
	    We also claim that $\sum_{k\in N'_o}q_{j,o}=\sum_{k\in N'_o}(1-\sum_{o'\in M^*}q_{k,o'})$ in Eq~(\ref{eq:thm:familychar:3}).
	    Otherwise, $o\in M_{Q,r+1}$ since $\sum_{k\in N'_o}q_{k,o}<\sum_{k\in N'_o}p_{k,o}<s(o)$, and there exists $j'\in N'_o$ with $\sum_{o'\in\ucs(\succ_{j'},o)}q_{j',o}=\sum_{o'\in M^*}q_{j',o} + q_{j',o}<1$ by~\Cref{lem:sdcrf}~(\romannumeral1), a contradiction to $Q$ satisfying \sdcrf{}.
        With $r=\min(\{\hat{r}\mid j\in E_{Q,\hat{r}}(o)\})$ for any $j\in N'_o$, and~Eq~(\ref{eq:thm:familychar:2.5}) and~(\ref{eq:thm:familychar:3}),  we obtain that $\rho_o=1/n$.
	\end{itemize}
	
	Together we have that condition (2) holds for $r$.
	
	Finally we show that condition (3) holds for $r$.
	For any $j\in E_{Q,r}(o)$, if $j\in E_{Q,r'}(o)$ with some $r'<r$, then we have the proof trivially by the fact that condition (3) holds for $r'$.
	Next we consider the case that $j\notin E_{Q,r'}(o)$, i.e., $o\neq\tp{j}{M_{Q,r'}}$ with any $r'<r$.
	We show $q_{j,o}=p_{j,o}$ in both of the possible cases below:
	\begin{itemize}[topsep=0pt,itemindent=1em,itemsep=0.5em,leftmargin=0pt]
	    \item If $\sum_{o'\in\ucs(\succ_j,\tp{j}{M_{Q,r-1}})}p_{j,o'}<1$, then by the fact that condition (2) for $r$ which we just proved above, $p_{j,o}=\int^{r/n}_{(r-1)/n}\omega_{k}(t){\rm d}t=q_{j,o}$.
	    \item 	If $\sum_{o'\in\ucs(\succ_j,\tp{j}{M_{Q,r-1}})}p_{j,o'}=1$, then $\sum_{o'\in\ucs(\succ_j,\tp{j}{M_{P,r-1}})}p_{j,o'}=1$ by condition (1) for $r-1$, which means that $j$ is satisfied before round $r$ by~\cref{lem:sdcrf}~(\romannumeral1) and does not consume $o$ since $o\neq\tp{j}{M_{Q,r'}}=\tp{j}{M_{P,r'}}$ for any $r'<r$.
	    Therefore $p_{j,o}=0$.
	    As for $Q$, by condition (3) for $r'<r$, $\sum_{o'\in\ucs(\succ_j,\tp{j}{M_{Q,r-1}})}q_{j,o'}\ge \sum_{o'\in\ucs(\succ_j,\tp{j}{M_{Q,r-1}})}p_{j,o'}=1$.
    	Because $o\neq\tp{j}{M_{Q,r-1}}$ and $o\in M_{P,r}=M_{Q,r}\subseteq M_{Q,r-1}$ by condition (1) for $r$, we have that $\tp{j}{M_{Q,r-1}}\succ_j o$, and therefore $q_{j,o}=0=p_{j,o}$.
	\end{itemize}
	Together we have that condition (3) holds for $r$.

	From the induction above, we have that conditions (1) and (3) hold for any $r$, i.e., for any $r$, we have that $M_{P,r}=M_{Q,r}$, $E_{P,r}(o)=E_{Q,r}(o)$ for any $o\in M_{Q,r}$, and $p_{j,o}=q_{j,o}$ for any $j\in E_{Q,r}(o)$.
	In~\Cref{alg:pcr}, shares of $o$ are only allocated to agents in $E_{P,r}(o)$ in each round, and $o$ is exhausted at the end, which means that $p_{j',o}=0$ if $j'\notin E_{P,r}(o)=E_{Q,r}(o)$ for any $r$.
	With the fact the the supply of all the items are fully allocated to agents, it follows that $q_{j',o}=0$ if $j'\notin E_{Q,r}(o)$ for any $r$ by condition (3).
	Together we have that $P=Q$.
\end{proof}

\propimpsdcfri*
\begin{proof}
    For ease of reading, we recall the preference profile $R$ and assignment $Q$ used in~\Cref{prop:impefcr1}, which are used in the following proof.
  	
  	\vspace{1em}\noindent
	\begin{minipage}{\linewidth}
		\centering
		\begin{minipage}{0.4\linewidth}
			\begin{center}
			    Preference Profile $R$\\
				$\succ_1$: $a\succ_1 c\succ_1 b\succ_1 d$,\\
				$\succ_2$: $a\succ_2 c\succ_2 b\succ_2 d$,\\
				$\succ_3$: $a\succ_3 b\succ_3 c\succ_3 d$,\\
				$\succ_4$: $b\succ_4 a\succ_4 d\succ_4 c$.\\
			\end{center}
		\end{minipage}
		\begin{minipage}{0.4\linewidth}
			\centering
			\begin{center}
				\centering
				\begin{tabular}{c|cccc}
					\multicolumn{5}{c}{Assignment $Q$}\\
					&  a & b & c & d\\\hline
					1 & $\frac{1}{3}$ & $0$ & ? & ?\\
					2 & $\frac{1}{3}$ & $0$ & ? & ?\\
					3 & $\frac{1}{3}$ & $0$ & ? & ?\\
					4 & $0$ & $1$ & $0$ & $0$\\
				\end{tabular}
			\end{center}
		\end{minipage}
	\end{minipage}\vspace{1em}
    
    For any assignment $P$ satisfying \sdcrf{}, we have that $E_{P,1}(a)=\{1,2,3\}$ and $E_{P,1}(b)=\{4\}$.
    It follows that $\sum_{k\in E_{P,1}(a)}d_{P,1}(k)>s_{P,1}(a)$, and therefore $a\notin  M_{P,r}$ with $r>1$ by~\Cref{lem:sdcrf}~(\romannumeral2), which means that only agents $1,2$ and $3$ get shares of $a$.
    It also follows that agent $4$ fully gets $b$ for the same token.
    Then any assignment satisfying \sdcrf{} and \etep{} (implied by \sdef{}) is in the form of $Q$, but $Q$ does not satisfy \sdef{} as we have shown in~\Cref{prop:impefcr1}.
\end{proof}

\propimpsdcfrii*

\begin{proof}

  Assume such mechanism $f$ exists.
	Let $R$ be:
	\begin{equation*}
		\begin{split}
			\succ_1:~ &a\succ_1 b\succ_1\cdots\succ_1 h\succ_1 c,\\
			\succ_2:~ &a\succ_2 h\succ_2\cdots\succ_2 b\succ_2 c,\\
			\succ_{3\text{-}7}:~ &c\succ d\succ e \succ f\succ g\succ b\succ h\succ a,\\
			\succ_8:~ &c\succ_8 d\succ_8 b\succ_8 e \succ_8 f\succ_8 g\succ_8 h\succ_8 a.
		\end{split}
	\end{equation*}
	Let $P=f(R)$.
	Recall the notations in~\Cref{lem:sdcrf} that $s_{P,r}(o)=1-\sum_{k\in\bigcup_{r'<r}E_{P,r'}(o)}p_{k,o}$, and $d_{P,r}(j)=1-\sum_{o'\succ_j o}p_{j,o'}$  for any $j\in  E_{P,r}(o)$.
	
	\begin{itemize}[topsep=0pt,itemindent=1em,itemsep=0.5em,leftmargin=0pt]
	\item[-] 	For $r=1$, by~\Cref{lem:sdcrf}~(\romannumeral2), since $E_{P,1}(a)=\{1,2\}$ and $\sum_{k\in E_{P,1}(a)}d_{P,1}(k)>s_{P,1}(a)$, only agents $1$ and $2$ gets shares of $a$.
	It follows that only agents in $E_{P,1}(c)=\{3,\dots,8\}$ gets $c$ for the same token.
	Then we have $p_{1,a}=p_{2,a}=1/2$, $p_{j,c}=1/6$ for $j\in \{3,\dots,8\}$ by \etep{}.
	\item[-] 	For $r=2$, $M_{P,2}=\{b,d,\dots,h\}$, $E_{P,2}(b)=\{1\}$, $E_{P,2}(h)=\{2\}$, and $E_{P,2}(d)=\{3,\dots,8\}$.
	With $\sum_{k\in E_{P,2}(d)}d_{P,2}(k)>s_{P,2}(d)$, we have that $p_{j,d}=1/6$ for $j\in \{3,\dots,8\}$ by~\Cref{lem:sdcrf}~(\romannumeral2) and \etep{}.
	With $\sum_{k\in E_{P,2}(b)}d_{P,2}(k)\le s_{P,2}(b)$, $p_{1,b}=d_{P,2}(1)=1/2$ by~\Cref{lem:sdcrf}~(\romannumeral3), and it follows that $p_{2,h}=d_{P,2}(2)=1/2$ for the same token.
	\item[-] 	For $r=3$, $M_{P,3}=\{b,e,\dots,h\}$, $E_{P,3}(b)=\{8\}$, and $E_{P,3}(e)=\{3,\dots,7\}$.
	With \Cref{lem:sdcrf}~(\romannumeral2) and $\sum_{k\in E_{P,3}(d)}d_{P,3}(k)>s_{P,3}(d)$, we have that $p_{8,b}=1/2$.
	With $\sum_{k\in E_{P,3}(e)}d_{P,3}(k)>s_{P,3}(e)$, $p_{j,e}=1/5$ for $j\in \{3,\dots,7\}$ by \etep{}.
	\end{itemize}
	With the analysis above, we have assignment $P$ in the following form.
	
	\begin{center}
		\centering
		\begin{tabular}{c|cccccccc}
			\multicolumn{9}{c}{Assignment $P$}\\
			&  a & b & c & d & e & f & g & h\\\hline
			$1$ & $1/2$ & $1/2$ & $0$ & $0$ & $0$ & $0$ & $0$ & $0$\\
			$2$ & $1/2$ & $0$ & $0$ & $0$ & $0$ & $0$ & $0$& $1/2$\\
			$3$-$7$ & $0$ & $0$ & $1/6$ & $1/6$ & $1/5$ & $?$ & $?$ & $?$\\
			$8$ & $0$ & $1/2$ & $1/6$ & $1/6$ & $0$ & $?$ & $?$ & $?$\\
		\end{tabular}
	\end{center}

	If agent $8$ misreports her preference as
	\begin{equation*}
		\succ'_8:c\succ'_8 d\succ'_8 e\succ'_8 b\succ'_8 f\succ'_8 g\succ'_8 h\succ'_8 a,
	\end{equation*}
	then let $P'=f(R')$ for $R'=(\succ'_8,\succ_{-8})$.
	\begin{itemize}[topsep=0pt,itemindent=1em,itemsep=0.5em,leftmargin=0pt]
	\item[-] The analysis for $P'$ with $r=1$ and $2$ is the same as $P$.
	\item[-] 	For $r=3$, $M_{P',3}=\{b,e,\dots,h\}$ and $E_{P',3}(e)=\{3,\dots,8\}$.
	With $\sum_{k\in E_{P',3}(e)}d_{P',3}(k)>s_{P',3}(e)$, $p'_{j,e}=1/6$ for $j\in \{3,\dots,8\}$ by \etep{}.
	\item[-] 	For $r=4$, $M_{P',4}=\{b,f,g,h\}$, $E_{P',4}(b)=\{8\}$, and $E_{P',4}(f)=\{3,\dots,7\}$.
	With \Cref{lem:sdcrf}~(\romannumeral3) and $\sum_{k\in E_{P',4}(b)}d_{P',4}(k)=s_{P',4}(b)$, $p'_{8,b}=d_{P',4}(8)=1/2$.
	\end{itemize}
	Then we obtain the assignment $P'$ in the following form.
	
	\begin{center}
		\centering
		\begin{tabular}{c|cccccccc}
			\multicolumn{9}{c}{Assignment $P'$}\\
			&  a & b & c & d & e & f & g & h\\\hline
			$1$ & $1/2$ & $1/2$ & $0$ & $0$ & $0$ & $0$ & $0$ & $0$\\
			$2$ & $1/2$ & $0$ & $0$ & $0$ & $0$ & $0$ & $0$& $1/2$\\
			$3$-$7$ & $0$ & $0$ & $1/6$ & $1/6$ & $1/6$ & $?$ & $?$ & $?$\\
			$8$ & $0$ & $1/2$ & $1/6$ & $1/6$ & $1/6$ & $0$ & $0$ & $0$\\
		\end{tabular}
	\end{center}
	
	We see that $P'_8$ strictly dominates $P_8$ for $\sum_{o\in\ucs(8,e)}p_{8,o}=5/6<1=\sum_{o\in\ucs(8,e)}p'_{8,o}$ and $\sum_{o\in\ucs(8,o')}p_{8,o}\le \sum_{o\in\ucs(8,o')}p'_{8,o}$ for other $o'\in M$, a contradiction to the fact that $f$ is \sdspa{}.
\end{proof}

\propimpefcrsdcfr*

\begin{proof}
	Assume that there exists a mechanism $f$ satisfying \sdcrf{} and \etep{}.
	Let $P=f(R)$ for the following preference profile $R$, and then we show that $P$ is not \efcr{}, 
	\begin{equation*}
		\begin{split}
			\succ_{1,2}:~&a_1\succ_1 a_2\succ_1 a_3\succ_1\text{others}\\
			\succ_3:~&a_1\succ_3 a_2\succ_3 a_4\succ_3\text{others}\\
			\succ_{4,5}:~&b_1\succ_4 b_2\succ_4 b_3\succ_4\text{others}\\
			\succ_6:~&b_1\succ_6 b_2\succ_6 b_4\succ_6\text{others}\\
			\succ_{7\text{-}17}:~&c_1\succ_7 c_2\succ_7 c_3\succ_7 c_4 \succ_7 c_5 \succ_7 c_6 \succ_7\text{others}\\
			\succ_x:~&c_1\succ_x c_2\succ_x c_3\succ_x a_3 \succ_x b_3 \succ_x c_5 \succ_x c_4\\
			& \succ_x c_6 \succ_x\text{others}\\
		\end{split}
	\end{equation*}
	
	\begin{itemize}[topsep=0pt,itemindent=1em,itemsep=0.5em,leftmargin=0pt]
	\item[-] For $r=1$, $M_{P,1}=M,E_{P,1}(a_1)=\{1,2,3\},E_{P,1}(b_1)=\{4,5,6\}$ and $E_{P,1}(c_1)=\{7,\dots,17,x\}$.
	By~\Cref{lem:sdcrf}~(\romannumeral2), item $a_1,b_1,c_1\notin M_{P,2}$ since $d_{P,1}(j)=1$ for any $j\in N$.
	Then we have $p_{j,a_1}=1/3$ for $j\in E_{P,1}(a_1)$, $p_{i',b_1}=1/3$ for $j'\in E_{P,1}(b_1)$ and $p_{j^*,c_1}=1/12$ for $j^*\in E_{P,1}(c_1)$ by \etep{}.
	\item[-] 	For $r=2$, $M_{P,2}=\{a_2,a_3,a_4,b_2,b_3,b_4,c_2,\dots,c_6,\allowbreak \dots\},\allowbreak E_{P,2}(a_2)=\{1,2,3\},E_{P,2}(b_2)=\{4,5,6\}$ and $E_{P,2}(c_2)=\{7,\dots,17,x\}$.
	Similar to $r=1$, by~\Cref{lem:sdcrf}~(\romannumeral2), $a_2,b_2,c_2\notin M_{P,3}$.
	Then we have $p_{j,a_2}=1/3$ for $j\in E_{P,2}(a_2)$, $p_{i',b_2}=1/3$ for $j'\in E_{P,2}(b_2)$ and $p_{j^*,c_2}=1/12$ for $j^*\in E_{P,2}(c_2)$ by \etep{}.
	\item[-] 	For $r=3$, $M_{P,3}=\{a_3,a_4,b_3,b_4,c_3,\dots,c_6,\allowbreak \dots\},$ $E_{P,3}(a_3)=\{1,2\},$ $E_{P,3}(a_4)=\{3\},$ $E_{P,3}(b_3)\allowbreak=\{4,5\},E_{P,3}(b_4)=\{6\}$, and $E_{P,3}(c_3)=\{7,\dots,17,x\}$.
	By~\Cref{lem:sdcrf}~(\romannumeral3), $p_{j,a_3}=1/3$ for $j\in E_{P,3}(a_3)$, $p_{3,a_4}=1/3$, $p_{j',b_3}=1/3$ for $j'\in Ej_{P,3}(b_3)$, and $p_{6,b_4}=1/3$.
	Since $\sum_{\hat{o}\in\ucs(\succ_j,\tp{j}{M_{P,3}})}=1$ for any agent $j\in N'=\{1,\dots,6\}$, their allocations have been determined and we do not need to consider them for $r>3$.
	\item[-] 	For $r=4$, $M_{P,4}=\{a_3,a_4,b_3,b_4,c_4,c_5,c_6,\dots\}$.
	Then $E_{P,4}(c_4)\setminus N'=\{7,\dots,17\}$ and $E_{P,4}(a_3)\setminus N'=\{x\}$.
	By~\Cref{lem:sdcrf}~(\romannumeral2), $a_3,c_4\notin M_{P,5}$, $p_{x,a_3}=1/3$ and $p_{j',c_4}=1/11$ for $j'\in E_{P,4}(c_4)$ by \etep{}.
	\item[-] 	For $r=5$, $M_{P,5}=\{a_4,b_3,b_4,c_5,c_6,\dots\}$, $E_{P,5}(c_5)\setminus N'=\{7,\dots,17\}$ and $E_{P,5}(b_3)\setminus N'=\{x\}$.
	By~\Cref{lem:sdcrf}~(\romannumeral2), $b_3,c_5\notin M_{P,6}$, $p_{x,b_3}=1/3$ and $p_{j',c_5}=1/11$ for $j'\in E_{P,5}(c_5)$ by \etep{}.
	\item[-] 	For $r=6$, $M_{P,5}=\{a_4,b_4,c_6,\dots\}$, $E_{P,6}(c_6)\setminus N'=\{7,\dots,17,x\}$.
	By~\Cref{lem:sdcrf}~(\romannumeral2), $c_6\notin M_{P,7}$ and $p_{j',c_6}=1/12$ for $j'\in E_{P,6}(c_6)$ by \etep{}.
	\end{itemize}
	
	We show the part of $P$ which has been determined by $r\le 6$ in the following $P$(\romannumeral1) for agents $\{1,2,3,x\}$ over items $\{a_1,\dots,a_4\}$, $P$(\romannumeral2) for agents $\{4,5,6,x\}$ over items $\{b_1,\dots,b_4\}$, and $P$(\romannumeral3) for agents $\{7,\dots,17,x\}$ over items $\{c_1,\dots,c_6\}$.

	\vspace{1em}\noindent
	\begin{minipage}{\linewidth}
		\centering
		\begin{minipage}{0.4\linewidth}
	\begin{center}
		\centering
		\begin{tabular}{c|cccc}
			\multicolumn{5}{c}{Assignment $P$(\romannumeral1)}\\
			&  $a_1$ & $a_2$ & $a_3$ & $a_4$\\\hline
			1 & $1/3$ & $1/3$ & $1/3$ & $0$\\
			2 & $1/3$ & $1/3$ & $1/3$ & $0$\\
			3 & $1/3$ & $1/3$ & $0$ & $1/3$\\
			x & $0$ & $0$ & $1/3$ & $0$\\
		\end{tabular}
	\end{center}
		\end{minipage}
		\begin{minipage}{0.4\linewidth}
			\centering
	\begin{center}
		\centering
		\begin{tabular}{c|cccc}
			\multicolumn{5}{c}{Assignment $P$(\romannumeral2)}\\
			&  $b_1$ & $b_2$ & $b_3$ & $b_4$\\\hline
			4 & $1/3$ & $1/3$ & $1/3$ & $0$\\
			5 & $1/3$ & $1/3$ & $1/3$ & $0$\\
			6 & $1/3$ & $1/3$ & $0$ & $1/3$\\
			x & $0$ & $0$ & $1/3$ & $0$\\
		\end{tabular}
	\end{center}
		\end{minipage}
	\end{minipage}\vspace{1em}

	\begin{center}
		\centering
		\begin{tabular}{c|cccccc}
			\multicolumn{7}{c}{Assignment $P$(\romannumeral3)}\\
			&  $c_1$ & $c_2$ & $c_3$ & $c_4$ & $c_5$ & $c_6$  \\\hline
			$7$-$17$ & $1/12$ & $1/12$ & $1/12$ & $1/11$ & $1/11$ & $1/12$  \\
			$x$ & $1/12$ & $1/12$ & $1/12$ & $0$ & $0$ & $1/12$  \\
		\end{tabular}
	\end{center}
	
	There exists an assignment $A$ with $A(x)=c_6$ among the deterministic assignments which constitute the convex combination for $P$.
	In the following, we prove that none of such $A$ is \fhcr{}.
	According to $P$, $a_3$ is assigned to one of $\{1,2\}$ in $A$ since agent $x$ does not get it.
	Due to the fact that $\succ_1=\succ_2$, let $A(1)=a_3$ without loss of generality.
	With the fact that only agents in $\{1,2,3\}$ can get $\{a_1,a_2\}$, we have that agents $\{2,3\}$ get $\{a_1,a_2\}$.
	It follows that $b_3$ is assigned to one of $\{4,5\}$ and $\{4,5,6\}$ get $\{b_1,b_2,b_3\}$ for the same token.
	Due to the fact that $\succ_4=\succ_5$, let $A(4)=b_3$ without loss of generality, and therefore  $\{5,6\}$ get $\{b_1,b_2\}$.
	Agents in $\{7,\dots,17\}$ get the rest items, and for ease of exposition, let agent $j_i$ with $i\in\{1,\dots,6\}$ satisfy $j_i\in\{7,\dots,17\}$ and $j_i=A^{-1}(c_i)$.
	We further have the following analysis about checking if $A$ satisfies \fhcr{}:
	
	\begin{itemize}[topsep=0pt,itemindent=1em,itemsep=0.5em,leftmargin=0pt]
	\item[-]	For $r=1$, $\tps{A}{1}=\{a_1,b_1,c_1\}$ because $M_1=M$, $\tp{j}{M_1}=a_1$ for $j\in\{1,2,3\}$, $\tp{j'}{M_1}=b_1$ for $j'\in\{4,5,6\}$,  and $\tp{j^*}{M_1}=c_1$ for $j^*\in\{7,\dots,11,x\}$.
	Then one of $\{2,3\}$ gets $a_1$, one of $\{4,5\}$ gets $b_1$, and agent $j_1$ gets $c_1$ by \fhcr{}.
	\item[-]  	For $r=2$, no matter which $j\in\{2,3\}$ gets $a_1$ and which $j'\in\{4,5\}$ gets $b_1$, $\tps{A}{2}=\{a_2,b_2,c_2\}$ because for $M_2=M\setminus\tps{A}{1}$, $\tp{j}{M_2}=a_2$ for $j\in\{1,2,3\}$, $\tp{j'}{M_2}=b_2$ for $j'\in\{4,5,6\}$, and $\tp{j^*}{M_2}=c_2$ for $j^*\in\{7,\dots,11,x\}$.
	Then the rest one of $\{2,3\}$ gets $a_2$, the rest one of of $\{4,5\}$ gets $b_2$, and agent $j_2$ gets $c_2$ by \fhcr{}.
	\item[-] 	For $r=3$, we do not consider $j'\in\{2,3,5,6,j_1,j_2\}$ because $A(j')\in\bigcup_{r'<3}\tps{A}{r'}$.
	We obtain that $\tps{A}{3}=\{a_3,b_3,c_3\}$ because for $M_3=M\setminus\bigcup_{r'<3}\tps{A}{r'}$, $\tp{1}{M_3}=a_3$, $\tp{4}{M_3}=b_3$, and $\tp{j}{M_3}=c_3$ for $j\in\{7,\dots,11,x\}$.
	Then agent $1$ gets $a_3$, agent $4$ gets $b_3$, and agent $j_3$ gets $c_3$ by \fhcr{}.
	\item[-] 	For $r=4$, we do not consider $j'\in\{1,\dots,6,j_1,j_2,j_3\}$.
	We obtain that $\tps{A}{4}=\{c_4,c_5\}$ because for $M_4=M\setminus\bigcup_{r'<4}\tps{A}{r'}$, $\tp{j}{M_4}=c_4$ for $j\in\{7,\dots,11\}$ and $\tp{x}{M_4}=c_5$.
	However, we have that $A^{-1}(c_5)=j_5$ and $\tp{j_5}{M_4}=c_4$, which violates \fhcr{}.
	\end{itemize}
	
    With the analysis above, we have that $A$ does not satisfy \fhcr{}, and therefore $P$ does not satisfy \efcr{}, which means that $f$ does not satisfy \efcr{}, \sdcrf{}, and \etep{} simultaneously.
\end{proof}

\section{Acronyms}\label{sec:app:acronyms}

\begin{table}[htb]

	\centering
	\begin{tabular}{l|l}
		Abbr. & full names \\\hline
		ABM & adaptive Boston mechanism~\cite{alcalde1996implementation,Dur2019modified}\\
		
		\am{} & eager Boston mechanism\\
		
		BM & Boston mechanism~\cite{Kojima2014:Boston}\\
				
		PR & probabilistic rank~\cite{Chen2021:Theprobabilistic}\\
		
		\pcr{}& probabilistic respecting eagerness\\
				
		PS & probabilistic serial~\cite{Bogomolnaia01:New} \\
		
		RP & random priority~\cite{Abdulkadiroglu98:Random} \\

		\upre{}& uniform probabilistic respecting eagerness\\
	\end{tabular}
	\caption{Acronyms for mechanisms used in this paper.}\label{tab:abbrm}
\end{table}

\vskip 0.2in
\bibliography{citation}
\bibliographystyle{theapa}

\end{document}